\documentclass[ejs]{imsart}

\RequirePackage{amsthm,amsmath,amsfonts,amssymb}
\RequirePackage[numbers]{natbib}
\RequirePackage[colorlinks,citecolor=blue,urlcolor=blue]{hyperref}%% uncomment this for coloring bibliography citations and linked URLs
\RequirePackage{graphicx}%% uncomment this for including figures

\RequirePackage{bm}
\usepackage{accents}
\usepackage{multirow}
\usepackage[algosection,linesnumbered,lined,boxed,commentsnumbered]{algorithm2e}
\usepackage{float}
\usepackage{booktabs}
\usepackage{cprotect}
\startlocaldefs

\newcommand{\Real}{\mathbb{R}}

\newcommand{\fpr}[1]{\left(#1\right)}
\newcommand{\tpr}[1]{\left[#1\right]}
\newcommand{\spr}[1]{\left\{#1\right\}}

\newcommand{\norm}[1]{\left\lVert#1\right\rVert}
\newcommand{\V}[1]{{\bm{\mathbf{\MakeLowercase{#1}}}}} % vector
\newcommand{\M}[1]{{\bm{\mathbf{\MakeUppercase{#1}}}}} 
\newcommand{\tr}{\operatorname{tr}} % Trace
\newcommand{\abs}[1]{\left\lvert#1\right\rvert}

\def\bA{\mathbf{A}}
\def\bG{\mathbf{G}}
\def\bI{\mathbf{I}}

\def\bH{\mathbf{H}}

\def\bM{\mathbf{M}}
\def\bX{\mathbf{X}}
\def\bY{\mathbf{Y}}

\def\bZ{\mathbf{Z}}

\def\bB{\mathbf{B}}
\def\bC{\mathbf{C}}

\def\bU{\mathbf{U}}
\def\bV{\mathbf{V}}

\def\btZ{\widetilde{\bZ}}

\def\bG{\mathbf{G}}
\def\bI{\mathbf{I}}

\def\bH{\mathbf{H}}

\def\bM{\mathbf{M}}
\def\bX{\mathbf{X}}
\def\bY{\mathbf{Y}}
\def\bR{\mathbf{R}}

\def\bZ{\mathbf{Z}}

\def\bB{\mathbf{B}}
\def\bC{\mathbf{C}}

\def\bU{\mathbf{U}}
\def\bV{\mathbf{V}}

\def\bAMo{\mathbf{A}_{(M_{\text{o}})}}

\def\bHMo{\mathbf{H}_{(M_{\text{o}})}}

\def\bXMo{\mathbf{X}_{M_{\text{o}}\,\cdot}}

\def\bBMo{\mathbf{B}_{\cdot\,M_{\text{o}}}}

\def\bVMo{\mathbf{V}_{(M_{\text{o}})}}

\newcommand{\vc}{\textrm{vec}}

\def\bV{\mathbf{V}}
\def\bU{\mathbf{U}}

\def\bGamma{\boldsymbol{\Gamma}}

\def\bSigma{\boldsymbol{\Sigma}}

\def\bhatY{\widehat{\mathbf{Y}}}
\def\bhatB{\widehat{\mathbf{B}}}
\def\bhatSigma{\widehat{\bSigma}}

\DeclareMathOperator*{\argmin}{arg\,min}

\newcommand{\ubar}[1]{\underaccent{\bar}{#1}}

\newcommand*{\bigchi}{\mbox{\Large$\chi$}}% big chi
\newcommand{\sighat}[1]{\widehat{\M{\Sigma}}_{#1}}

\theoremstyle{plain}

\newtheorem{theorem}{Theorem}[section]
\newtheorem{lemma}[theorem]{Lemma}

\newtheorem{condition}[theorem]{Condition}
\theoremstyle{remark}
\newtheorem{definition}[theorem]{Definition}

\endlocaldefs

\begin{document}

\begin{frontmatter}
\title{The EAS approach to variable selection for multivariate response data in high-dimensional settings}
%\title{A sample article title with some additional note\thanksref{T1}}
\runtitle{EAS for variable selection in multivariate regression}
%\thankstext{T1}{A sample of additional note to the title.}

\begin{aug}
\author[a]{\fnms{Salil}  \snm{Koner}\ead[label=e1,mark]{Salil.Koner@duke.edu}}
\and
\author[b]{\fnms{Jonathan P}  \snm{Williams}\ead[label=e2]{jwilli27@ncsu.edu}}
  
\runauthor{Koner and Williams}
\address[a]{Duke University, Durham, NC. \printead{e1}}
\address[b]{North Carolina State University, Raleigh, NC. \printead{e2}}
\end{aug}

\begin{abstract}
In this paper, we develop an {\em epsilon admissible subsets} (EAS) model selection approach for performing group variable selection in the high-dimensional multivariate regression setting.  This EAS strategy is designed to estimate a posterior-like, generalized fiducial distribution over a parsimonious class of models in the setting of correlated predictors and/or in the absence of a sparsity assumption.  The effectiveness of our approach, to this end, is demonstrated empirically in simulation studies, and is compared to other state-of-the-art model/variable selection procedures.  Furthermore, assuming a matrix-Normal linear model we show that the EAS strategy achieves {\em strong model selection consistency} in the high-dimensional setting if there does exist a sparse, true data generating set of predictors.  In contrast to Bayesian approaches for model selection, our generalized fiducial approach completely avoids the problem of simultaneously having to specify arbitrary prior distributions for model parameters and penalize model complexity; our approach allows for inference directly on the model complexity. \textcolor{black}{Implementation of the method is illustrated through yeast data to identify significant cell-cycle regulating transcription factors.}

% Furthermore, our framework has been designed with particular consideration for applying the method to discretized functional data.  Special care is needed for extending our EAS strategy to the functional data setting, but we demonstrate how this can be done with numerical illustrations.  
\end{abstract}



\end{frontmatter}
 
%%%%%%%%%%%%%%%%%%%%%%%%%%%%%%%%%%%%%%%%%%%%%%%%%%%%%

\section{Introduction}\label{sec: Intro}

With the advent of modern data collection technologies in many real-life applications, multiple responses are simultaneously collected that are characterized by a set of explanatory variables. Data of this structure falls under the scope of multivariate regression. Examples arise in chemometrics \citep{frank1993statistical}, genome-wide association studies (GWAS) \citep{boulesteix2005predicting}, etc. More often than not, the number of predictors $p$ is much larger than the number of observed multivariate response vectors $n$. Parsimoniously modeling the variability in the response without overfitting is necessary to enhance prediction accuracy. For example, in GWAS, identification of key genetic markers out of millions that are associated with a univariate or multivariate phenotype is of scientific interest \citep{vounou2010discovering}. Model/variable selection is a statistical framework that has been widely popular in this context. A naive way to approach this multivariate problem is to model the components of the response as separate univariate regressions on the predictors, and to employ existing selection techniques available in the univariate setting.  However, since the multiple responses for each subject are often correlated, prediction error can be minimized substantially if one uses the inherent association between the responses effectively \citep{breiman1997predicting}. Finding the best model in the context of multivariate linear regression (MLR) without ignoring the correlation between the responses, especially in a high-dimensional setting is a challenging task and has received much attention in the literature over the last decade. 

We develop an EAS approach for {\em group} variable selection in high-dimensional MLR settings. The EAS procedure was originally developed for high-dimensional univariate regression settings of \cite{williams2019nonpenalized}, and has been extended to the vector auto-regression setting in \cite{williams2019eas}.  However, an EAS procedure has not been constructed for the multivariate regression setting, nor has there been built a group selection mechanism for EAS selection. In contrast to Bayesian model selection approaches, we consider a generalized fiducial (GF) inference approach \citep{hannig2016generalized} that explicitly estimates the GF distribution over the class of all subsets of predictors; whereas frequentist and most Bayesian approaches exclusively focus on coefficient estimation to perform variable selection. In a high-dimensional setting, it is very problematic for a variable selection procedure to over-rely on the magnitude of the estimated regression coefficients because they lack identifiability and are numerically unstable. Moreover, most variable selection procedures for multivariate regression do {\em not} account for the correlation structure of the multivariate responses. Our multivariate group EAS procedure is designed to inherently accommodate the arbitrary covariance structure of the response. As explained in \cite{lee2012simultaneous}, accounting for the correlation is important because, for example, in the case of the least absolute shrinkage and selection operator (LASSO) estimator the shrinkage criterion is affected by the magnitude and sign of the correlation between the responses. 

Mathematically, under a standard sparsity and Gaussian errors assumption with a general covariance structure we prove that our proposed EAS procedure achieves {\em strong model selection} consistency, as defined in \cite{narisetty2014bayesian}. That is, we show that over the class of all $\epsilon${\em -admissible} subsets of the predictors/groups, the GF probability of a true sparse model converges to one in probability as the sample size goes to infinity, and the number of predictors/groups is allowed to grow sub-exponentially as a function of the sample size. Additionally, as a next paper in the series of papers to investigate EAS model selection strategies in various settings, the algorithm we propose improves on the computational efficiency and stability of the algorithms proposed in \cite{williams2019nonpenalized, williams2019eas}. We provide user-friendly \verb1R1 software to implement our EAS procedure, available at \url{https://github.com/SalilKoner/EAS}.

Next, in various simulation scenarios, we demonstrate that our EAS procedure is either competitive with or outperforms the state-of-the-art Bayesian or frequentist approaches, based on various metrics such as prediction error, misclassification rate, and average proportion of correct model selections. Moreover, consistent with our theory, our EAS method does an excellent job in assigning a very high probability mass to the true model compared to the other Bayesian methods. 

%To highlight the novelty of our EAS method compared to \cite{williams2019nonpenalized}, firstly, the residual sum of squares (RSS) analog in MLR is a random matrix rather than the univariate random variable in the univariate linear regression setting. Standard well-established concentration inequalities on RSS for the univariate case do not apply to the determinant or the eigenvalues of the error covariance matrix in MLR, for which much less standard theory exists. Secondly, our EAS methodology enjoys variable selection consistency even if the tuning parameter $\epsilon$ does not grow with the sample size, markedly different from \cite{williams2019nonpenalized}. Lastly, the so-called GF ``Jacobian term'' for the MLR model along with the construction of a closed form expression of the GF probability distribution requires complicated matrix calculus, thus entailing significantly more challenges than in the univariate case. 

\textcolor{black}{There are four distinct elements of novelty of our EAS procedure for the MLR setting compared to the univariate case \cite{williams2019nonpenalized}. To start with the methodological novelties: 1) Our EAS method is designed to inherently accommodate an arbitrary covariance structure of the response in performing variable selection. 2) The second added difficulty, specific to GF inference, is that the so-called GF ``Jacobian term'' for the MLR model (derived in Appendix~\ref{sec: suppGFI}) is non-standard and exerts substantial influence on the resulting GF model probabilities that we construct, whereas the Jacobian term in the univariate linear regression setting is concise and involves components readily relatable to the likelihood function. 3) The theoretical novelties, primarily imparted by accounting for an arbitrary covariance structure, constitute derivation of non-asymptotic concentration bounds on ratios of determinants (or eigenvalues) of error covariance matrices for different models (i.e., subsets of covariates) in the multivariate setting, for which much less standard theory exists. 4) The key computational improvement of our EAS method is that the user does not need to scale the tuning parameter $\epsilon$ as a function of the sample size to select the best model, in contrast to previous EAS developments such as \cite{williams2019nonpenalized}.}

Prior choice/specification in contemporary Bayesian approaches are typically not chosen because they reflect true prior knowledge/beliefs, but they are tailored to simplify computational complexities and/or achieve desirable large sample/frequentist properties.  While this practice is pragmatic, it is a violation of fundamental Bayesian principles.  In contrast, GF inference is an equally principled framework that has an appeal to objective Bayesian perspectives, but it does not suffer from the arbitrary choice of prior specification.  \textcolor{black}{The GF approach is to solve an inverse problem resulting in parameter values most consistent in reconciling the observed data with random draws from the distribution of the auxiliary variable in the data generating equation.  Holding the data fixed, these parameter values then inherit a probability distribution via the distribution of the auxiliary variable.  It can be argued that this procedure effectively has implicit prior knowledge built in, but the subjectivist Bayesian approach (assuming a likelihood function) goes a step further by imposing/requiring additional prior knowledge (in the form of a prior distribution specification) that is exogenous to the data generating model.} Moreover, it has been shown that GF inference exhibits large sample Bernstein-von Mises type properties that guarantee the nominal coverage of credible sets, similar to such theory for Bayesian posteriors \citep{sonderegger2014fiducial}.  See \cite{hannig2016generalized} for a full introduction of GF inference.  
%In our paper, we introduce a GF-based inferential approach for the MLR setting under the Gaussian data generating assumption with a general covariance structure.

%\textcolor{black}{We remark that there are two substantial complications that arise in the extension from the high-dimensional univariate linear regression setting investigated in \cite{williams2019nonpenalized}.  Concentration bounds on ratios of residual sums of squares for different models in the univariate setting become concentration bounds on ratios of determinants (or eigenvalues) of error covariance matrices for different models in the multivariate setting, for which much less standard theory exists.  The second added difficulty, specific to GF inference, is that the so-called GF ``Jacobian term'' for the MLR model (derived in our Supplementary Material) is mathematically complicated and exerts substantial influence on the resulting GF model probabilities that we construct, whereas the Jacobian term in the univariate linear regression setting \citep{williams2019nonpenalized} is concise and involves components readily relatable to the likelihood function.}

\textcolor{black}{Standard model selection techniques for univariate linear regression via Mallow's Cp and other types of information criterion have been extended to multivariate regression; see \cite{sparks1983multivariate, fujikoshi1997modified} and the references therein.  
Since the inception of LASSO \citep{tibshirani1996regression}, penalized methods introducing sparsity in the regression coefficients have engulfed the MLR literature. Notable contributions are the simultaneous variable selection using $L_\infty$SVS from \citep{turlach2005simultaneous} and $L_2$SVS from \cite{simila2006common}; the remMap procedure in \cite{peng2010regularized} employs an elastic net type penalty to identify {\em master predictors}. Estimation of the regression coefficients taking into account the correlation between the responses was pioneered in  \cite{rothman2010sparse}, which was later extended in \cite{lee2012simultaneous}. Assuming that the groups are known, a multivariate sparse group LASSO strategy is proposed in \cite{li2015multivariate} to impose the group structure, which was later augmented in \cite{wilms2018algorithm} for simultaneous covariance estimation. Very recently, de-sparsified/de-biased LASSO was developed in \cite{chevalier2020statistical, bellec2021chi} to overcome the bias induced by the penalty, by extending the existing techniques developed for the univariate setting \citep{van2014asymptotically}.
From a Bayesian perspective, a stochastic search variable selection (SSVS) procedure for MLR is employed in \cite{brown1998multivariate, brown2002bayes}. The multivariate Bayesian group LASSO using a spike-and-slab prior is developed in \cite{liquet2017bayesian} to perform the variable selection. To mitigate the computational issues of spike-and-slab priors for the large $p$ scenario, continuous global-local shrinkage priors are introduced in \cite{bai2018high}. Recently, an expectation-maximization based maximum a posteriori (MAP) estimation procedure was formulated for fast simultaneous variable and covariance selection using continuous shrinkage priors in \cite{deshpande2019simultaneous}}.

So far the NP-hard problem of best subset selection has been seemingly conveniently handled by assuming sparsity in the true data-generating model. However, in a high-dimensional setting, especially when there is a strong degree of collinearity amongst the predictors, there may not be a unique model that fits the data best, and so the concept of a {\em true model} is somewhat vague. Moreover, the typical $\ell_1$ and $\ell_2$ regularization methods shrink coefficients to zero only based on their magnitude, which is again unreliable in the presence of multicollinearity. Our EAS procedure provides a fresh perspective on variable selection in the MLR setting by defining an admissibility condition for candidate models. The admissibility criterion relates to the idea that any candidate model, as defined by a set of predictors, is redundant if there exists a subset model that explains the variation in the response as well as the candidate model. Thus, while we prove important mathematical properties of the EAS approach under a sparsity assumption, in finite sample data analyses the key functionality of our EAS approach is not to determine the model that is necessarily the true set of predictors but to identify a parsimonious model that explains the data as well as the true model, if it exists.  This criterion is meaningful even in the absence of a sparsity assumption.  We refer to the criterion as $\epsilon${\em -admissibility}, defined in Section \ref{sec: EASmethods}.  The key characteristic of the EAS-based GF distribution is that it assigns very negligible probability to models that fail the $\epsilon$-admissibility criterion, and in doing so significantly reduces the class of candidate models to choose from. This is the intuition for why, assuming sparsity, the procedure achieves strong model selection consistency.

\textcolor{black}{An advantage of our GF-based variable selection procedure over the frequentist counterparts is that it provides an estimate of the model probabilities, derived from the posterior-like GF distribution of the model parameters. Even many recent standard Bayesian approaches, such as Multivariate Bayesian model with Shrinkage Priors (MBSP) from \cite{bai2018high}, multivariate spike-and-slab prior (mSSL) from \cite{deshpande2019simultaneous}, and Spike-and-Slab Group LASSO (SSGL) from \cite{bai2020spike} are not developed to compute relative model probabilities; rather they are designed to estimate the MAP probability model parameters. A notable exception is the Bayesian spike-and-slab prior group LASSO from \cite{liquet2017bayesian} that is capable of providing relative model probabilities, but is unfortunately not suitable for high-dimensional settings.  In a finite sample scenario, relative model probabilities are useful because they give a certain degree of confidence in choosing one model over another, and they reflect a useful discrimination between competing models, especially in the high-dimensional setting.  Moreover, the $\epsilon$-admissibility criterion in our EAS framework takes into account the covariance structure of the multivariate response. In a multi-response setting, the noise associated with a particular component of the response may be significantly higher than the other components. Many of the variable selection procedures such as {remMap} \cite{peng2010regularized}, $L_2$SVS  \cite{simila2006common}, or reduced rank regression \cite{chen2012sparse, velu2013multivariate}, neither consider this difference in noise levels nor the intra-dependence in the multivariate response.} 

The rest of the paper is organized as follows. In Section~\ref{sec: EASmethods} we layout the EAS methodology and highlight differences from other EAS approaches.  Next, the computational algorithm to implement the method is presented in Section~\ref{sec: EAScomputation}.  A few essential non-asymptotic results characterizing the meaningfulness of the EAS procedure along with the main consistency result are stated in Section~\ref{sec: EAStheory}. The proof of all the results are relegated to the \textcolor{black}{Section~\ref{sec: suppPROOF} of the Appendix}. Finite sample numerical results covering both $n > p$ and $p > n$ cases are illustrated in Section~\ref{sec: EASsimstudy}. \textcolor{black}{Section~\ref{sec: EASrealdata} provides an illustration of the practical application of the procedure using yeast cell cycle data.}  Computer codes to reproduce all empirical results are available at \textcolor{black}{\url{https://github.com/SalilKoner/EAS}}. 

\section*{Notations} Throughout the course of the paper we will use the following notations. For an vector $\V{a} \in \Real^n$, $\norm{\V{a}} =: \sqrt{\sum_{i=1}^n a_i^2}$ denotes vector $2$-norm; for any matrix $\bA$, $\norm{\bA}$ refers to the spectral norm (i.e., $\norm{\bA} := \sup_{\V{x}, \norm{\V{x}}=1} \norm{\bA \V{x}}$); $\norm{\bA}_{\textrm{F}}$ denotes the Frobenius norm; and $\norm{\bA}_{\max} := \max_{i,j} \abs{a_{ij}}$ denotes the max norm. For any symmetric matrix $\bA$, $\lambda_{\min}(\bA)$ and $\lambda_{\max}(\bA)$ denotes the minimum and maximum eigenvalues, respectively, of the matrix $\bA$. For an $m \times r$ random matrix $\bX$, $\bX \sim \textrm{Matrix-Normal}_{m,r}(\bM, \bU, \bV)$ is equivalent to $\textrm{vec}(\bM) \sim \textrm{N}_{m,r}(\textrm{vec}(\bM), \bV \otimes \bU)$ \citep[see][definition 2.2.1]{gupta2018matrix}. For a $m \times r$ random matrix $\bX$, $\bX \sim \textrm{T}_{m,r}(\nu, \bM, \bU, \bV)$ means $\bX$ follows a matrix $t$-distribution with mean $\bM$, scale matrices $\bU$, $\bV$, and degrees of freedom $\nu$ \citep[see][definition 4.2.1]{gupta2018matrix}. For any set $M$, $\abs{M}$ denotes the cardinality of $M$. For any event $A$, $\mathrm{I}(A)$ denotes the indicator function that the event happens. Lastly, for two sequences $a_n$ and $b_n$, $a_n = o(b_n)$ means $\lim_{n \to \infty} a_n/b_n = 0$.  

\section{Methodology} \label{sec: EASmethods}
In MLR, $n$ pairs of examples $(\M{x}_i, \M{Y}_i)$ are observed, where $\M{Y}_i := (Y_{i1}, \dots, Y_{iq})^\top$ is the $q$-dimensional multivariate response for $i$-th subject, and $\M{X}_i := (X_{1i}, \dots, X_{pi})^\top$ contains the values of $p$ predictor variables that are presumed to be associated with the response. The response is expressed as the linear model,
$$\M{Y}_i = \sum_{j=1}^p \bB_{j}X_{ji} + \bA\bU_i,$$ 
where $\bB_j := (b_{1j}, \dots, b_{qj})^\top$  is a $q$-dimensional regression coefficient vector where $b_{\ell j}$ captures the effect of $j$-th predictor on $\ell$-th coordinate of the multivariate response, $\bA$ is $q \times q$ matrix, and $\bU_i \in \Real^{q\times 1}$ is the $i$-th error vector \textcolor{black}{with mean $\V{0}$ and $\textrm{Var}(\bU_i) = \bI_q$}, so that $\textrm{Var}(\bY_i) = \bA\bA^\top$. Further denoting $\bY := [\bY_1, \dots, \bY_n] \in \Real^{q \times n}$ as the horizontal column stacked response, and $\bU := [\bU_1, \dots, \bU_n]$ as the corresponding $q \times n$ dimensional error matrix, the multivariate regression model with a sample of size $n$ is summarized as,
\begin{equation} \label{eqn: EASdgm}
\bY = \Big[\sum_{j=1}^p \bB_j X_{j1}, \cdots, \sum_{j=1}^p \bB_j X_{jn}\Big] + \bA\bU = \bB\bX + \bA\bU,
\end{equation}
where $\bB := [\bB_1, \dots, \bB_p] \in \Real^{q \times p}$ is the coefficient matrix and $\bX := [\bX_1, \dots, \bX_n] \in \Real^{p \times n}$ is the design matrix.

 \textcolor{black}{In the context of variable selection, for any index set $M \subseteq \{1,\dots,p\}$, let $\mathbf{X}_{M\,\boldsymbol\cdot}$ denote the matrix with {\em rows} comprised of the rows of $\mathbf{X}$ corresponding to the indices in $M$.  In subsetting the rows of $\mathbf{X}$ for variable selection, the {\em columns} of $\mathbf{B}$ must be subset to only those corresponding to the indices in $M$; take $\mathbf{B}_{\boldsymbol\cdot\,M}$ to be the column-subsetted coefficient matrix.  The subscript $(M)$, as in the $q \times q$ matrix $\mathbf{A}_{(M)}$, simply denotes an association with the index set $M$, rather than any subsetting of the rows/columns.}  Accordingly, conditional on index set/model $M$, expression (\ref{eqn: EASdgm}) reduces to,
\begin{equation}\label{eqn: EASgroupedmodel}
\M{Y} = \bB_{\boldsymbol\cdot\,M}\bX_{M\,\boldsymbol\cdot} + \bA_{(M)}\bU.
\end{equation}
Observe that variable selection in this multivariate model is in fact a group selection problem because the active columns (i.e., the groups) of the coefficient matrix $\bB$ are being selected.  This fact establishes the need for the a group selection mechanism within the variable selection procedure.  Nonetheless, the EAS variable selection procedure that we develop seamlessly accommodates the additional problem of selecting among known/posited groups of predictors, rather than the power set of the predictors $1,\dots,p$ (i.e., the {\em natural grouping}).  Our presentation will focus on EAS methodology for solving the problem of grouped variable selection for MLR under the natural grouping, but the methodology and theoretical results cover (as a sub-case) the simpler case when a class of predictor groups are known.  In that case, simply restrict the class of candidate models that the algorithm is allowed to choose from. 

In the case of Gaussian error, we introduce the notion of a {\em true} data generating model as the assumption that,
\begin{equation} \label{eqn: truemodel}
\bY \sim \textrm{Matrix-Normal}_{q,n}\fpr{\bB_{\boldsymbol\cdot\,M_{\text{o}}}^{0}\bX_{M_{\text{o}}\,\boldsymbol\cdot}, \bV_{(M_{\text{o}})}^{0}, \bI_n},
\end{equation}
for some fixed (but unknown) $M_{\text{o}} \subseteq \spr{1,2,\dots,p}$, and some fixed (but unknown) parameter matrices $\bB_{\boldsymbol\cdot\,M_{\text{o}}}^{0}$ and $\bV_{(M_{\text{o}})}^{0} =  \bA_{(M_{\text{o}})}^{0}\bA_{(M_{\text{o}})}^{0\top}$, for some positive definite matrix $\bA_{(M_{\text{o}})}^{0}$. Note that in the case of the covariance matrix $\bV_{(M_{\text{o}})}^{0}$, the subscript $M_{\text{o}}$ simply denotes the association with the index set $M_{\text{o}}$ but it is {\em not} constructed by subsetting some more general matrix $\bV^0$.  Analogous to \cite{williams2019nonpenalized}, the index `o' in $M_{\text{o}}$ is in reference to the term `oracle', and the superscript `0' emphasizes that the quantities $\bB_{\boldsymbol\cdot\,M_{\text{o}}}^{0}$ and $\bV_{(M_{\text{o}})}^{0}$ are fixed quantities (in contrast to their GF/Bayesian-like random variable counterparts $\bB_{\boldsymbol\cdot\,M_{\text{o}}}$ and $\bV_{(M_{\text{o}})}$, to be introduced shortly).  The matrix normal notation in~(\ref{eqn: truemodel}) is a compact way of saying that the multivariate responses $\bY_1, \dots, \bY_n$ are independent and identically distributed multivariate normal random vectors with mean $\bB_{\boldsymbol\cdot\,M_{\text{o}}}^{0}\bX_{M_{\text{o}}\,\boldsymbol\cdot}$ and $q \times q$ covariance matrix $\bV_{(M_{\text{o}})}^{0}$.

The objective of our paper is to develop a methodology that identifies {\em a non-redundant} (i.e., $\epsilon$-admissible) set $M_{\text{o}}$ (the true model or possibly a related sub-model of the true model) out of the $2^p$ candidate sets in the power set of $\spr{1,2,\dots,p}$.  In the special case that the true model is sparse, the objective is to show that the method will identify {\em the} true model $M_{\text{o}}$, as the sample size is taken to infinity.  The notion of  $\epsilon$-admissibility is defined in Definition \ref{defn: hfunction}.  To build up to this definition, the following model-based perspective is required. 

For any index set $M \subseteq \spr{1,\dots,p}$, we assume the conditional distribution,
\begin{equation} \label{eqn: EAStruesampdistY}
    \bY \vert \bB_{\boldsymbol\cdot\,M}, \bV_{(M)} \sim \textrm{Matrix-Normal}_{q,n}\fpr{\bB_{\boldsymbol\cdot\,M}\bX_{M\,\boldsymbol\cdot}, \bV_{(M)}, \bI_n},
\end{equation}
where $\bB_{\boldsymbol\cdot\,M}$ and $\bV_{(M)} = \bA_{(M)}\bA_{(M)}^\top$ are random matrices that reflect the uncertainty in not knowing the true data generating model nor true values of its parameter matrices.  Under the Gaussian error assumption, the quantities $\bB_{\boldsymbol\cdot\,M}$ and $\bV_{(M)}$ are expected to be centered, respectively, around the least squares estimator $\widehat{\bB}_{\boldsymbol\cdot\,M} := \bY\bX_{M\,\boldsymbol\cdot}^\top\fpr{\bX_{M\,\boldsymbol\cdot}\bX_{M\,\boldsymbol\cdot}^\top}^{-1}$ and the restricted maximum likelihood estimator $\widehat{\bV}_M := \widehat{\bSigma}_{(M)}/(n- \abs{M})$, with $\widehat{\bSigma}_{(M)} := \bY(\bI_n - \bH_{(M)})\bY^\top$ and $\bH_{(M)} := \bX_{M\,\boldsymbol\cdot}^\top\fpr{\bX_{M\,\boldsymbol\cdot}\bX_{M\,\boldsymbol\cdot}^\top}^{-1}\bX_{M\,\boldsymbol\cdot}$.  Note that the matrix $\bH_{(M)}$ is the orthogonal projection onto the row space of $\bX_{M\,\boldsymbol\cdot}$.  Moreover, given the true data generating model (\ref{eqn: truemodel}), it follows that $\mathbb{E}(\bhatY) = \mathbb{E}(\bhatB_{\boldsymbol\cdot\,M}\bX_{M\,\boldsymbol\cdot}) = \bB_{\boldsymbol\cdot\,M_{\text{o}}}^{0}\bX_{M_{\text{o}}\,\boldsymbol\cdot}\bH_{(M)}$, and so if $M \supseteq M_{\text{o}}$ then $\mathbb{E}(\bhatY) = \bB_{\boldsymbol\cdot\,M_{\text{o}}}^{0}\bX_{M_{\text{o}}\,\boldsymbol\cdot}$. This means that any collection of predictors with linear span containing the oracle predictors, $M_{\text{o}}$, is as good at explaining variation in the data as the true model (in terms of residual sum of squares). However, such a large/redundant set of predictors lacks efficiency in terms of prediction accuracy.  Exploiting this idea, the critical definition supporting our methodology is presented next.
\begin{definition} \label{defn: hfunction}
A $q \times \abs{M}$ regression coefficient matrix $\bB_{\boldsymbol\cdot\,M}$ coupled with an index set $M \subseteq \spr{1,\dots,p}$ is said to be {\em $\epsilon$-admissible} if $h_\epsilon(\bB_{\boldsymbol\cdot\,M}) = 1$ where,
\begin{align} \label{eqn: EAShfuncdefn}
    h_\epsilon(\bB_{\boldsymbol\cdot\,M}) = 
        \mathrm{I}\spr{\frac{1}{2}\left \lVert \bhatSigma_{(M)}^{-1/2} \fpr{\bB_{\boldsymbol\cdot\,M}\bX_{M\,\boldsymbol\cdot} - \bB_{\min}\bX} \right \rVert_\textrm{F}^2 \geq \epsilon} \textrm{I}\{\abs{M} < n - q\},
\end{align}
where $\bB_{\min}$ is the solution to the optimization problem,
\begin{align*}
\underset{\bB \in \Real^{q \times p}}{\argmin}  \left \lVert \bhatSigma_{(M)}^{-1/2} \fpr{\bB_{\boldsymbol\cdot\,M}\bX_{M\,\boldsymbol\cdot} - \bB\bX} \right \rVert_\textrm{F}^2 \ \ \text{ subject to } \ \ \abs{\spr{j: \norm{\bB_j} \neq 0}} \leq \abs{M}-1.
\end{align*}
\end{definition}

Definition~\ref{defn: hfunction} characterizes a notion of redundancy for any set of predictors, indexed by $M$. The quantity $\left \lVert  \bB_{\boldsymbol\cdot\,M}\bX_{M\,\boldsymbol\cdot} - \bB\bX \right \rVert_\textrm{F}^2$ captures the difference in prediction of the model $M$ from all models with fewer predictors.  Any model $M$ that is not $\epsilon$-admissible is redundant in the sense that there exists a subset of fewer predictors that approximately linearly spans the same subspace.  This very notion of redundancy makes the EAS method different from the traditional regularization-based approaches, where redundancy is expressed as a model containing negligible or zero magnitude regression coefficients.  Nonetheless, Definition~\ref{defn: hfunction} encompasses the traditional notion of redundancy because, if one column of $\bB_{\boldsymbol\cdot\,M}$ is equal to zero then $h_\epsilon(\bB_{\boldsymbol\cdot\,M})$ is $0$ for all $\epsilon > 0$. Additionally, as a consequence of the rows of $\bX_{M\,\boldsymbol\cdot}$ spanning a finite-dimensional vector space, $h_\epsilon(\bB_{\boldsymbol\cdot\,M})$ assigns value zero to all models $M$ with $\abs{M} > n$, by definition.  
As a consequence, the EAS procedure inherently reduces the difficultly of the model selection problem from $2^{p}$ candidate models to $2^{n}$, a fact that is fundamental to the scalability of the EAS procedure for high-dimensional settings.  Furthermore, if $\bX_{M\,\boldsymbol\cdot}$ does not have full row rank, then $h_\epsilon(\bB_{\boldsymbol\cdot\,M})$ is again zero by its construction.

The above definition of the $h$-function is well-defined in the sense that for $M$ with $\abs{M} < n - q$, the matrix $\bhatSigma_{(M)}$ is invertible with probability $1$. Lemma~\ref{Lemma: MinEigenofRSSMatrix} in Section~\ref{sec: EAStheory} justifies that the minimum eigenvalue of $\bhatSigma_{(M)}$ diverges away from $0$ for large $n$.  Next, for identifiability of a sparse true model, $M_{\text{o}}$, the choice of $\epsilon$ must not be so large that it classifies $M_{\text{o}}$ as redundant.  Conversely, if $\epsilon$ is chosen too small, then many redundant models might also satisfy the $\epsilon$-admissibility criterion.  It will be seen throughout the remainder of the paper that this trade-off analysis is the crux of the theoretical underpinnings of the EAS approach.

A distinction between the definition of the $h$-function defined in~(\ref{eqn: EAShfuncdefn}) and the one defined in \cite{williams2019nonpenalized} is the introduction of the empirical error covariance matrix $\bhatSigma_{(M)}$. Normalization by the square-root of error covariance matrix is common in LASSO-type model selection strategies such as concomitant multi-task regression \citep{massias2018generalized}, multivariate square-root LASSO \citep{van2016chi, bertrand2019handling} to name a few. For any fixed model $M$, a common assumption is that the quantity $\left \lVert  \bB_{\boldsymbol\cdot\,M}\bX_{M\,\boldsymbol\cdot} - \bB\bX \right \rVert_\textrm{F}^2$ is on the order of $n$ (for large $n$). As a result, the optimal choice of $\epsilon$ as derived in \cite{williams2019nonpenalized} turned out to be a function of both $n$ and $|M_{\text{o}}|$. Moreover, the form of their suggested $\epsilon$ is somewhat unintuitive as it was derived purely from a theoretical result. On the other hand, since $\bhatSigma_{(M)}$ is also on the order of $n$ (for large $n$), adjusting for the inverse square-root of $\bhatSigma_{(M)}$, as in (\ref{eqn: EAShfuncdefn}), proportionately scales $\left \lVert  \bB_{\boldsymbol\cdot\,M}\bX_{M\,\boldsymbol\cdot} - \bB\bX \right \rVert_\textrm{F}^2$. This enables us to choose the threshold $\epsilon$ via simple grid search based on some metric such as cross-validation (CV) technique or information criterion (IC), independently of $n$ and $|M|$. We implement a computationally efficient version of the EAS algorithm that bypasses the original repeated sampling based pseudo-marginal Markov chain Monte Carlo (MCMC) algorithm. A detailed description of the computational procedure we propose is presented in Section \ref{sec: EAScomputation}.

With our $\epsilon$-admissibility notion of redundancy now defined, we are ready to build the statistical framework that will facilitate its use in model selection.  The GF inference approach adopted in \cite{williams2019nonpenalized} remains an advantageous pathway for constructing a posterior-like probability distribution over the class of candidate models; one that is principled in the sense of Bernstein-von Mises asymptotics, but also avoids the problem of prior choice/specification. 

\textcolor{black}{To introduce the mechanics of GF inference, assume a random variable $Z$ has a forward data generation equation that can be expressed as $Z = G(U, \theta)$, where $G$ is some known deterministic function, $U$ is a pivotal quantity whose distribution is known, and $\theta$ is some fixed but unknown parameter(s) of interest in the space $\Theta$.  Given an observed data set $\mathbf{z} = (z_1, \dots, z_n)^\top$ of independent instances of the random variable $Z$, GF inference aims to find the best fitting $\vartheta$ such that $\lVert \mathbf{z} - G(\mathbf{U}, \vartheta) \rVert $ is minimized, resulting in a random variable $\theta^* := \theta^*(\bU)$. The distribution of the random variable $\theta^*(\bU)$ is termed the GF distribution of the unknown parameter $\theta$. When $Z$ is a continuous random variable, under certain regularity conditions \citep[as stated in][mostly dealing with the smoothness of $G$]{hannig2016generalized} the GF distribution of $\theta$ can be expressed as 
\begin{equation}
 r(\vartheta \mid \mathbf{z}) = \frac{f(\mathbf{z}, \vartheta)J(\mathbf{z}, \vartheta)}{\int_{\Theta} f(\mathbf{z}, \widetilde{\vartheta})J(\mathbf{z}, \widetilde{\vartheta})\;d\widetilde{\vartheta}},
\end{equation}
where $f(\mathbf{z}, \vartheta)$ is the likelihood function and 
\begin{equation} \label{eqn: JacobianForm}
J(\mathbf{z}, \vartheta) := D\bigg(\frac{d}{d\vartheta}G(\mathbf{u}, \vartheta)\Big|_{\mathbf{u}=G^{-1}(\mathbf{z},\vartheta)} \bigg),
\end{equation}
with $D(\bC) := \sqrt{\det\bC^\top \bC}$ for a matrix argument $\bC$.  The Jacobian-like quantity $J(\mathbf{z}, \vartheta)$ results from inverting the data generating equation assuming the inverse $G^{-1}(\mathbf{z}, \vartheta)$ exists. 
}

\textcolor{black}{
To illustrate, if $Z_1, \dots, Z_n \overset{iid}{\sim} N(\mu,\sigma^2)$, then $G(U, (\mu, \sigma^2)) = \mu + \sqrt{\sigma^2} U$ where $U \sim N(0,1)$.  In this case, 
\[
J(\mathbf{z}, (\mu, \sigma^2)) = D([\mathbf{1}_n \,\,-\sigma^{-2}(\mathbf{z} - \mu \mathbf{1}_n)/2]) = \sqrt{n}\sigma^{-2}s_z/2,
\]
where $\bar{z} := \sum_{i=1}^n z_i$ and $s_z^2 := \sum_{i=1}^n (z_i-\bar{z})^2$. Then the GF distribution takes the form
\begin{equation*}
r(\mu, \sigma^2 \mid \mathbf{z}) \propto (\sigma^2)^{-n/2-1}e^{-\sigma^{-2}\sum_{1}^{n}(z_i-\mu)^2/2},
\end{equation*}
which implies that the conditional GF distribution of $(\mu \mid \sigma)$ is $N(\bar{z}, \sigma^2/n)$, and the marginal GF distribution of $\sigma^2$ is Inverse-Gamma$\{(n-1)/2, s_z^2/2\}$.  This solution is consistent with the posterior distribution of $(\mu, \sigma^2)$ constructed via flat prior specification \citep[page 65,][]{gelman1995bayesian}.
}

In contrast to the univariate regression model (as in \cite{williams2019nonpenalized}), however, construction of the GF distribution in the MLR setting is {\em not} a simple extension and is accompanied by unique challenges, especially so for dealing with the arbitrary covariance matrix of the multivariate response vectors. 

The difficulty in deriving and studying an expression for a GF distribution, for most continuous data models, is that it requires deriving a complicated function of partial derivatives with respect to the unknown model parameters.  Interesting data models for which the GF density can be expressed analytically (up to a normalizing constant), such as we will show for the MLR with arbitrary coefficient and covariance matrices, are interesting in their own right for their contribution to the growing literature on GF inference.  The assumptions and materials for deriving/computing a GF distribution are provided in \cite{hannig2016generalized}.  In the remainder of this section we provide the details relevant to our methodology.

Given an index set, $M$, the unknown parameters in model~(\ref{eqn: EAStruesampdistY}) are $\bB_{\boldsymbol\cdot\,M}$ and $\bA_{(M)}$. As in the GF inference setup, re-express the data generating equation (\ref{eqn: EAStruesampdistY}) as,
\begin{equation} \label{eqn: EASdatageneqn}
\bY = \bB_{\boldsymbol\cdot\,M}\bX_{M\,\boldsymbol\cdot} + \bA_{(M)}\bU =: G\big(\bU, (\bB_{\boldsymbol\cdot\,M}, \bA_{(M)})\big),
\end{equation}
where $\bU \sim \textrm{Matrix-Normal}_{q,n}(\V{0}, \bI_q, \bI_n)$.  As prescribed in Theorem 1 of \cite{hannig2016generalized}, the GF density of the parameters $(\bB_{\boldsymbol\cdot\,M}, \bA_{(M)})$ can be expressed as,
\[
r(\bB_{\boldsymbol\cdot\,M}, \bA_{(M)} \mid \bY) := \frac{  f\big(\bY,(\bB_{\boldsymbol\cdot\,M}, \bA_{(M)})\big) J\big(\bY,(\bB_{\boldsymbol\cdot\,M}, \bA_{(M)})\big)  }{  \int f\big(\bY,(\widetilde{\bB}_{\boldsymbol\cdot\,M}, \widetilde{\bA}_{(M)})\big) J\big(\bY,(\widetilde{\bB}_{\boldsymbol\cdot\,M}, \widetilde{\bA}_{(M)})\big) \ d(\widetilde{\bB}_{\boldsymbol\cdot\,M}, \widetilde{\bA}_{(M)})  },
\]
where $f$ is the matrix normal likelihood function, and the {\em Jacobian} term defined in (\ref{eqn: JacobianForm}).  After some routine matrix calculations the Jacobian term reduces to, 
\begin{equation*}
J\big(\bY,(\bB_{\boldsymbol\cdot\,M}, \bA_{(M)})\big) = \fpr{\det \M{A}_{(M)}\M{A}_{(M)}^\top}^{-q/2} \fpr{\det \bX_{M\,\boldsymbol\cdot}\bX_{M\,\boldsymbol\cdot}^\top}^{q/2} \fpr{\det \widehat{\M{\Sigma}}_{(M)}}^{q/2}.
\end{equation*}
Accordingly, restricting $M$ to the class of $\epsilon$-admissible models yields the GF density, 
\begin{align*}
r_\epsilon(\bB_{\boldsymbol\cdot\,M}, \bA_{(M)} \mid \bY) \propto \frac{ e^{- \frac{1}{2}\tr\tpr{\bR_{(M)}\bV_{(M)}^{-1}}} }{ \fpr{\det \bV_{(M)}}^{(n+q)/2} } \fpr{\det \bX_{M\,\boldsymbol\cdot}\bX_{M\,\boldsymbol\cdot}^\top}^{q/2} \fpr{\det \widehat{\M{\Sigma}}_{(M)}}^{q/2} h_\epsilon\fpr{\bB_{\boldsymbol\cdot\,M}},
\end{align*}
where $\bR_{(M)} := \fpr{\M{Y}- \bB_{\boldsymbol\cdot\,M}\bX_{M\,\boldsymbol\cdot}} \fpr{\M{Y}- \bB_{\boldsymbol\cdot\,M}\bX_{M\,\boldsymbol\cdot}}^\top$.  Note the dependence of $r_\epsilon(\cdot \mid \bY)$ on the choice of $\epsilon$.  Further, $r_\epsilon(\cdot \mid \bY)$ should not be confused with the notation for a conditional probability density function, but should be understood to reflect the fact that the GF distribution is a function of the observed data $\bY$. 

Moving along, analogous to a Bayesian model selection approach, we construct a probability distribution over all $\epsilon$-admissible index sets as the marginal distribution,
{\small
\begin{align*}
r_\epsilon(M \mid \bY) &\propto \int_{\Real^{q \times \abs{M}}} \int_{\Real^{q \times q}} r_\epsilon(\bB_{\boldsymbol\cdot\,M}, \bA_{(M)} \mid \bY) \ d\M{A}_{(M)} \ d\bB_{\boldsymbol\cdot\,M} \\
&= \fpr{\det \bX_{M\,\boldsymbol\cdot}\bX_{M\,\boldsymbol\cdot}^\top}^{\frac{q}{2}} \fpr{\det \widehat{\M{\Sigma}}_{(M)}}^{\frac{q}{2}} \int h_\epsilon\fpr{\bB_{\boldsymbol\cdot\,M}} \int \frac{ e^{- \frac{1}{2}\tr\tpr{\bR_{(M)}\bV_{(M)}^{-1}}} }{ \fpr{\det \bV_{(M)}}^{(n+q)/2} } \ d\bA_{(M)} \ d\bB_{\boldsymbol\cdot\,M}.
\end{align*}
}\noindent \textcolor{black}{We simplify this expression as equation (\ref{eqn: EASMassfunctionmodelM}), stated next, and provide a detailed account of the intermediate steps in Section~\ref{sec: suppGFI} of the Appendix.} The derivation of GF model involve integration over the domain of positive definite matrices and non-trivial matrix algebra that are far more complex than the case of univariate linear regression setting. 

%Note that the Jacobian term for the MLR model (derived in our Supplementary Material) is mathematically complicated and exerts substantial influence on the resulting GF model probabilities in equation (\ref{eqn: EASMassfunctionmodelM}), whereas the Jacobian term in the univariate linear regression setting \citep{williams2019nonpenalized} is concise and involves components readily relatable to the likelihood function.

\begin{equation}\label{eqn: EASMassfunctionmodelM}
r_\epsilon(M \mid \bY) \propto \Gamma_q\fpr{\frac{n-\abs{M}}{2}}\pi^{\frac{q\abs{M}}{2}}\fpr{\det \widehat{\M{\Sigma}}_{(M)}}^{-\fpr{\frac{n-\abs{M}-q}{2}}}\mathbb{E}\tpr{h_\epsilon\fpr{\bB_{\boldsymbol\cdot\,M}}},
\end{equation}
where the expectation is taken with respect to the density of matrix $t$-distribution, i.e.,
\begin{equation}\label{eqn: matrixt}
\bB_{\boldsymbol\cdot\,M} \sim \textrm{T}_{q,\abs{M}}\fpr{n-\abs{M}-q+1, \widehat{\bB}_{\boldsymbol\cdot\,M}, \bhatSigma_{(M)}, \fpr{\bX_{M\,\boldsymbol\cdot}\bX_{M\,\boldsymbol\cdot}^\top}^{-1}}.
\end{equation}
Note that the GF distribution of $\bB_{\boldsymbol\cdot\,M}$ is concentrated around the least squared estimator, $\widehat{\bB}_{\boldsymbol\cdot\,M}$, defined previously.

Observe in~(\ref{eqn: EASMassfunctionmodelM}) that for models with $\abs{M} > n - q$, the $h$-function is zero by definition, and thus $r_\epsilon(M \mid \bY)$ is trivially zero. This probability mass function has the interpretation as the relative likelihood of the model $M$ versus that of all other candidate models in the class of $\epsilon$-admissible models.  It becomes clear from the expression (\ref{eqn: EASMassfunctionmodelM}) that $r_\epsilon(M \mid \bY)$ is largely driven by the inverse of the empirical error covariance matrix, and that the $h$-function delivers a multiplicative effect on the probability.  For a large redundant model, $M$, we expect that the determinant of the empirical error covariance is small relative to that of $M_{\text{o}}$, and so we leverage the choice of $\epsilon$ such that $\mathbb{E}\tpr{h_\epsilon\fpr{\bB_{\boldsymbol\cdot\,M}}}$ controls the value of $r_\epsilon(M \mid \bY)$.  This insight is formalized in Section \ref{sec: EAStheory}.    

As we illustrate in the remainder of this paper, the GF mass function $r_\epsilon(M \mid \bY)$ serves as a vehicle for model selection and inference.  In the next section we discuss the details of the computations, and provide an algorithm to generate samples from this GF distribution.

\section{Model estimation and computational techniques} \label{sec: EAScomputation}

\textcolor{black}{In order to generate samples from $r_\epsilon(M \mid \bY)$ we must be able to compute $\mathbb{E}\tpr{h_\epsilon\fpr{\bB_{\boldsymbol\cdot\,M}}}$. Although the expectation is with respect to a matrix $t$-distribution, the complex expression for $h_\epsilon(\bB_{\boldsymbol\cdot\,M})$ makes the form of its expectation intractable, and so standard MCMC techniques do not apply. This issue is typical of all the previously developed EAS implementations.  The MLR EAS analogue of the previous EAS approaches is to employ a pseudo-marginal MCMC algorithm by estimating $\mathbb{E}\tpr{h_\epsilon\fpr{\bB_{\boldsymbol\cdot\,M}}}$ with the average of a large number of random samples from the GF distribution of $\bB_{\boldsymbol\cdot\,M}$ (i.e., its matrix $t$-distribution). An important remark is that in contrast to the previous EAS articles, in our empirical investigations we find that rather than generating a sample of matrices $\bB_{\boldsymbol\cdot\,M}$ as in (\ref{eqn: matrixt}) for approximating $h_\epsilon(\bB_{\boldsymbol\cdot\,M})$ with an a sample mean, it suffices to take $h_\epsilon(\widehat{\bB}_{\boldsymbol\cdot\,M})$ as a point estimate of $\mathbb{E}\tpr{h_\epsilon\fpr{\bB_{\boldsymbol\cdot\,M}}}$, where $\widehat{\bB}_{\boldsymbol\cdot\,M}$ is the least squares estimator for model $M$. This is likely partly due to our construction for the $h$ function having a better scaling with $\epsilon$ than in earlier developments of EAS approaches, and the fact that the distribution in (\ref{eqn: matrixt}) is centered at $\widehat{\bB}_{\boldsymbol\cdot\,M}$. This adjustment makes the implementation of EAS highly efficient, and is supported by the competitive performance exhibited in extensive numerical studies, summarized in Section \ref{sec: EASsimstudy}. We briefly discuss the algorithms of EAS for MLR case, next.}

%We extend these algorithms to the MLR case, and discuss them briefly, here.

\textcolor{black}{From Definition \ref{defn: hfunction}, evaluating $h_\epsilon(\widehat{\bB}_{\boldsymbol\cdot\,M})$ can be formulated as the mixed integer quadratic program (MIQP) with quadratic constraints,
\begin{align*}
\vc\fpr{\bB_{\min}} = \underset{\V{b}, z_{1},\dots,z_{p}}{\argmin}\Big\{ \frac{1}{2}  \V{b}^\top\M{Q}_{(M)}\V{b} - \V{a}_{(M)}^\top\V{b}\Big\},  
\end{align*}
subject to $\V{b} \in \Real^{qp}$, $z_j \in \{0,1\}$ for $j \in \{1,\dots,p\}$,
\begin{align*}
\V{b}^\top\M{C}_j\V{b} &\leq \mathcal{M}_U z_j^2, \qquad \text{and} \qquad \sum_{j=1}^p z_j \leq \abs{M}-1,
\end{align*} 
where $\M{Q}_{(M)} := (\M{X}\M{X}^\top) \otimes \widehat{\bSigma}_{(M)}^{-1}$, $\V{a}_{(M)} := \fpr{\M{X}\bX_{M\,\boldsymbol\cdot}^\top \otimes \widehat{\bSigma}_{(M)}^{-1}} \vc(\widehat{\bB}_{\boldsymbol\cdot\,M})$, $\M{C}_j $ is a block diagonal matrix with $p$ blocks of $q \times q$ zero matrices, except $\bI_q$ in the $j$-th block, and $\mathcal{M}_U > 0$ is a properly chosen constant.  In particular, the quantity $\mathcal{M}_U$ must be chosen large enough so that $\mathcal{M}_U > \max_j \norm{(\bB_{\min})_j}_2^2$ \citep{bertsimas2016best}.  Further, since $\bB_{\min}$ is not a-priori known, \cite{bertsimas2016best} provides a data-driven formula to specify $\mathcal{M}_U$ in the MIQP, which can be solved with any MIQP solver, such as} \verb1CPLEX1.

\begin{algorithm}[H]
\textcolor{black}{
\SetAlgoLined
\KwIn{Input a model with index set $M$, least squared estimator $\widehat{\bB}_{\boldsymbol\cdot\,M}$, the full design matrix $\bX$, $\bhatSigma_{(M)}$, a pre-specified $\epsilon > 0$ and an initial solution $\bB_{init}$ with number of columns with non-zero norm less than $\abs{M}-1$;}
\KwOut{Value of $h_\epsilon(\widehat{\bB}_{\boldsymbol\cdot\,M})$;} 
Calculate $L = \lambda_{\max}(\bX\bX^\top)\lambda_{\min}^{-1}(\bhatSigma_{(M)}) $ and set $\bB_{cur} = \bB_{init}$\;
Calculate the objective function $g(\bB_{cur}) = \frac{1}{2} \lVert  \bhatSigma_{(M)}^{-1/2} (\widehat{\bB}_{\boldsymbol\cdot\,M}\bX_{M\,\boldsymbol\cdot} - \bB_{cur}\bX) \rVert_\textrm{F}^2$\;
\While{$\textrm{diff} > \textrm{threshold}$ or $g(\bB_{cur}) > \epsilon$}{
Calculate $\bB = \bB_{cur} - \frac{1}{L} \bhatSigma_{(M)}^{-1} \fpr{\bB_{cur}\bX\bX^\top -  \widehat{\bB}_{\boldsymbol\cdot\,M}\bX_{M\,\boldsymbol\cdot}\bX^\top}$ \;
Obtain indices $i_1, i_2, \dots, i_{\abs{M}-1}$ such that $\norm{\bB_{i_1}} \geq \norm{\bB_{i_2}} \dots \geq \norm{\bB_{i_{\abs{M}-1}}} \geq \norm{\bB_{i_{\abs{M}}}} \dots \geq \norm{\bB_{i_p}}$ \;
Set $\bB_j = 0$ for all $j \in \spr{1,2,\dots,p} \backslash \spr{i_1,i_2,\dots,  i_{\abs{M}-1}}$ \;
Calculate $\mathrm{diff} = \abs{g(\bB_{cur}) -g(\bB)}$\;
Update $\bB_{cur} = \bB$ \;
}
\Return $h_\epsilon(\widehat{\bB}_{\boldsymbol\cdot\,M}) = \bI\fpr{g(\bB_{cur}) > \epsilon}$;
 \caption{\footnotesize  Pseudocode for projected gradient descent to compute $h_\epsilon(\widehat{\bB}_{\boldsymbol\cdot\,M})$.}
 \label{algo: hfunction}
 }
\end{algorithm}

\textcolor{black}{Although a single MIQP is typically fast to solve in practice, we need to compute $h_\epsilon(\widehat{\bB}_{\boldsymbol\cdot\,M})$ for different models $M$ at each step of the MCMC, so further streamlining of the computations are needed.  First, observe that a solution, $\bB_{\min}$, to the MIQP is not always necessary to evaluate $h_\epsilon(\widehat{\bB}_{\boldsymbol\cdot\,M})$; rather $h_\epsilon(\widehat{\bB}_{\boldsymbol\cdot\,M}) = 0$ if there exists any $\bB$ (satisfying the MIQP constraints) such that $\frac{1}{2}\lVert \bhatSigma_{(M)}^{-1/2} (\widehat{\bB}_{\boldsymbol\cdot\,M}\bX_{M\,\boldsymbol\cdot} - \bB\bX) \rVert_\textrm{F}^2 < \epsilon$.  If this `stopping' condition is met prior to obtaining $\bB_{\min}$, then the MIQP solver can be terminated early.  Second, as an alternative to an explicit MIQP solver, a discrete first-order gradient-descent based algorithm proposed in \cite{bertsimas2016best} can be implemented for a crude but super-efficient computation of $h_\epsilon(\widehat{\bB}_{\boldsymbol\cdot\,M})$, when $h_\epsilon(\widehat{\bB}_{\boldsymbol\cdot\,M}) = 0$.  This `projected gradient-descent' algorithm is advocated as a warm start to the MIQP in \cite{bertsimas2016best}, and the pseudocode for our implementation of it is given in Algorithm~\ref{algo: hfunction}.  Note that the gradient of objective function in the optimization problem in Definition~\ref{defn: hfunction} is Lipschitz continuous with Lipschitz constant $L = \lVert \bhatSigma_{(M)}^{-1} \rVert \lVert \bX\bX^\top \rVert$.  For $\widehat{\bB}_{\boldsymbol\cdot\,M}$ that are {\em not} $\epsilon$-admissible, we observed in empirical experimentation that by initializing Algorithm~\ref{algo: hfunction} at $\widehat{\bB}_{\boldsymbol\cdot\,M}$, with its column having minimum norm set to zero, it usually finds a solution to determine that $h_\epsilon(\widehat{\bB}_{\boldsymbol\cdot\,M})$ is zero within a few iterations. }

\begin{algorithm}[H]
 \textcolor{black}{
\SetAlgoLined
\KwIn{An index set $M$ and $\epsilon > 0$}
\KwOut{A new index set $M^*$} 
Calculate \begin{math}
   \widetilde{M} = \begin{cases}
   M \cup \{\textrm{a new covariate} \} & w.p. \;\; 1/3 \\
   M  \;\backslash \; \{\textrm{a existing covariate} \} & w.p. \;\; 1/3 \\
   M \;\backslash \; \{\textrm{a existing covariate} \} \cup \{\textrm{a new covariate} \} & w.p. \;\; 1/3
\end{cases}
\end{math}\;
Calculate LS estimator of $\widehat{\bB}_{\boldsymbol\cdot\,M}$ and  $\widehat{\M{\Sigma}}_{(M)}$ corresponding to ${M}$\;
Calculate LS estimator of $\widehat{\bB}_{\boldsymbol\cdot\,\widetilde{M}}$ and  $\widehat{\bSigma}_{(\widetilde{M})}$ corresponding to the proposal ${\widetilde{M}}$\;
% Generate $ \bB_{\boldsymbol\cdot\,M}(k) \sim T_{q, \lvert M \rvert } \fpr{n-(\lvert M \rvert+q)+1,\widehat{\bB}_{\boldsymbol\cdot\,M}, , \widehat{\M{\Sigma}}_{(M)}, \fpr{\bX_{M\,\boldsymbol\cdot}\bX_{M\,\boldsymbol\cdot}^\top}^{-1} } $,  for $k=1,2,\dots,N$ \;
% Generate $\bB_{\boldsymbol\cdot\,\widetilde{M}}(k)  \sim  T_{q,\lvert{\widetilde{M}\rvert}} \fpr{n-(\lvert \widetilde{M} \rvert+q)+1,\widehat{\bB}_{\boldsymbol\cdot\,\widetilde{M}}, \widehat{\bSigma}_{(\widetilde{M})}, \fpr{\bX_{\widetilde{M}\,\boldsymbol\cdot}\bX_{\widetilde{M}\,\boldsymbol\cdot}^\top}^{-1} } $, for $k=1,2,\dots,N$\;
Compute $h_\epsilon(\widehat{\bB}_{\boldsymbol\cdot\,M})$ and $h_\epsilon(\widehat{\bB}_{\boldsymbol\cdot\,\widetilde{M}})$ using Algorithm~\ref{algo: hfunction}\;
Calculate $\widehat{r}_\epsilon\fpr{M \mid \M{Y} }$ and 
  $\widehat{r}_\epsilon\fpr{\widetilde{M} \mid \M{Y} }$ as in~(\ref{eqn: rhat}) \;
 $M^* = \begin{cases} 
 \widetilde{M}
& w.p. \;\; \rho\fpr{M, \widetilde{M}} \\
M & w.p. \;\; 1- \rho\fpr{M, \widetilde{M}} \\
\end{cases}
$, where $\rho\fpr{M, \widetilde{M}} := \min \left\{\frac{ \widehat{r}_\epsilon\fpr{\widetilde{M} \mid \M{Y} }}{\widehat{r}_\epsilon\fpr{M \mid \M{Y} }}, 1 \right\}$ \;
  \caption{\footnotesize Pseudocode for one step of the MCMC algorithm to estimate $r_\epsilon(M \mid \bY)$.  Note that weights can be used for randomly adding/dropping/switching covariates in the proposed index set $\widetilde{M}$ in line 2 (e.g., correlation-based weights). In that case, the Metropolis-Hasting ratio $\rho(M, \widetilde{M})$ in line 7 needs to be updated accordingly.}
 \label{algo: MCMC}
 }
\end{algorithm}

\textcolor{black}{Now that we have a computationally efficient algorithm for computing the $h_\epsilon(\widehat{\bB}_{\boldsymbol\cdot\,M})$, the remaining task is to demonstrate the mechanism for generating samples from the GF distribution of $M$. Estimating $\mathbb{E}\tpr{h_\epsilon\fpr{\bB_{\boldsymbol\cdot\,M}}}$ by $h_\epsilon(\widehat{\bB}_{\boldsymbol\cdot\,M})$, the GF probability mass function $r_\epsilon(M \mid \bY)$ can be approximated as
\begin{equation}\label{eqn: rhat}
\widehat{r}_\epsilon(M \mid \M{Y})  := \Gamma_q\fpr{\frac{n-\abs{M}}{2}}\pi^{\frac{q\abs{M}}{2}}\fpr{\det \widehat{\M{\Sigma}}_{(M)}}^{-\fpr{\frac{n-\abs{M}-q}{2}}}h_\epsilon\fpr{\widehat{\bB}_{\boldsymbol\cdot\,M}}.
\end{equation}
We present the pseudocode for our implementation in Algorithm~\ref{algo: MCMC}. We demonstrate empirically in Section \ref{sec: EASsimstudy} that this approximation gives results that are competitive with the state-of-the-art Bayesian and frequentist methods for the MLR, both in terms of performance and computation time.}

%with respect to the location-scale matrix $t$-distribution specified in (\ref{eqn: matrixt}).  Then the Grouped Independence Metropolis-Hastings (GIMH) algorithm \citep{andrieu2009pseudo} yields samples (asymptotically) from $r_\epsilon(M \mid \bY)$ by, for each step of the algorithm, proposing index sets $M$ and averaging the value of $r_\epsilon(M, \bB_{\boldsymbol\cdot\,M}(k) \mid \M{Y})$ over $N$ generated matrices $\bB_{\boldsymbol\cdot\,M}(1),\dots, \bB_{\boldsymbol\cdot\,M}(N)$ from the matrix $t$-distribution in (\ref{eqn: matrixt}).  

\section{Theoretical Results}\label{sec: EAStheory}

The main objective of this section is to establish the consistency of our model selection procedure, particularly in the high-dimensional setting (i.e., $p \gg n$) with the assumption that the true model is sparse.  We begin by stating and describing essential conditions and necessary supporting results to show that $r_\epsilon(M_{\text{o}} \mid \bY)$ converges in probability to $1$ as $n \to \infty$.  Our strong model selection consistency result is stated as Theorem \ref{Lemma: Mainresult}.  Throughout this section we a-priori fix the following values.  
%The constants $c, C \in (0,\infty)$ that appear in our results are taken as special fixed values that depend only on the sub-Gaussian norm of a centered multivariate Gaussian vector with covariance matrix $\bV_{(M_{\text{o}})}^{0}$. 
Let $\ubar{\lambda}_v$ and $\bar{\lambda}_v$ be the minimum and maximum eigenvalues of the true covariance matrix, $\bV_{(M_{\text{o}})}^{0}$, respectively.  Denote by $\mathbb{P}_y(\cdot)$ be the probability measure associated with the sampling distribution of the response $\bY$, as in (\ref{eqn: truemodel}), and denote by $\mathbb{P}(\cdot)$ the probability measure associated with the GF distribution of the parameters. Similarly, denote by $\mathbb{E}_y(\cdot)$ the expectation with respect to the sampling distribution of the response $\bY$, as in (\ref{eqn: truemodel}), and denote by $\mathbb{E}(\cdot)$ the expectation with respect to the GF distribution of the parameters.

The major theoretical intricacies that we deal with while extending from the high-dimensional univariate linear regression case are, first, the residual sum of squares in the multivariate linear model is no longer a scalar, indeed a $q \times q$ matrix. As evident from the equation (\ref{eqn: EASMassfunctionmodelM}), in order to establish asymptotic variable selection selection consistency of the true model, we must derive the concentration bound on the ratio of the determinant of residual sum of squares of matrix for any arbitrary model to that for the true model. This calls for a lower bound on the size of minimum eigenvalue of the $\widehat{\bSigma}_{(M)}$, for which the exisiting theoretical results are sparse, compared to the well-established standard chi-squared tail bounds that apply to the RSS in univariate setting. Second, the non-asymptotic bound for the ratio of the determinant of error covariance matrix is sharper in the sense that it does not pivot upon the growth of $\boldsymbol{\Delta}_{(M)} := \bB_{\boldsymbol\cdot\,M_{\text{o}}}^{0}\bX_{M_{\text{o}}\,\boldsymbol\cdot}(\bI_n- \bH_{(M)})\bX_{M_{\text{o}}\,\boldsymbol\cdot}^\top\bB_{\boldsymbol\cdot\,M_{\text{o}}}^{0\top}$. Consequently, we establish strong variable selection consistency of the true model without imposing any rate on $\boldsymbol{\Delta}_{(M)}$ \citep[first part of Condition 3.2]{williams2019nonpenalized}, quite distinctly from the univariate paper. Third, the concentration bounds for the $\mathbb{E}(h_\epsilon(\bB_{\boldsymbol\cdot\,M}))$ are derived explicitly as a function of $n$, in contrast to Theorem 3.7 and 3.8 of \cite{williams2019nonpenalized} where the bounds are expressed as a function of $\epsilon$, which demands for an additional non-intuitive assumption \citep[Condition 3.4]{williams2019nonpenalized} on the rate of $\epsilon$. Theorem~\ref{Lemma: Mainresult} in our paper guarantees that such a condition is not essential to establish variable selection consistency as long as $\epsilon$ satisfies Condition \ref{cond: upperboundepsilon} and \ref{cond: lowerboundepsilon}, which are fundamental to  the EAS methodology.  

%\textcolor{black}{There are two substantial complications that arise in the extension from the high-dimensional univariate linear regression setting investigated in \cite{williams2019nonpenalized}.  Concentration bounds on ratios of residual sums of squares for different models in the univariate setting become concentration bounds on ratios of determinants (or eigenvalues) of error covariance matrices for different models in the multivariate setting, for which much less standard theory exists.  The second added difficulty, specific to GF inference, is that the Jacobian term for the MLR model (derived in the Supplementary Material) is mathematically complicated and exerts substantial influence on the resulting model probabilities in equation (\ref{eqn: EASMassfunctionmodelM}), whereas the Jacobian term in the univariate linear regression setting \citep{williams2019nonpenalized} is concise and involves components readily relatable to the likelihood function.}

Conditon~\ref{cond: eigenvalueV} requires that the true covariance matrix of the response is positive-definite and finite, which implies that none of the components of the multivariate response are degenerate and they all have finite second moments.  Since the dimension of the multivariate response, $q$, is fixed, this assumption is rather routine.  

\begin{condition}[{\em Non-singularity of true covariance}] \label{cond: eigenvalueV}
The dimension of the multivariate response vector, $q$, is fixed, and $0 < \ubar{\lambda}_v < \bar{\lambda}_v < \infty$. 
\end{condition} 

In our methods, it is important that $\bV_{(M_{\text{o}})}^{0}$ is positive-definite so that there exists a positive-definite, consistent estimator of it, for example, $\bhatSigma_{(M_{\text{o}})}/(n-\abs{M_{\text{o}}})$.  This estimator plays an essential role in our definition of the $h$ function (among other roles).  In particular, Condition~\ref{cond: eigenvalueV} makes it possible that, for large $n$, the minimum eigenvalue of $\bhatSigma_{(M)}/(n-\abs{M})$ is bounded away from $0$ with high probability, for an important class of models.  This fact is established in Lemma~\ref{Lemma: MinEigenofRSSMatrix}, presented next.

\begin{lemma}\label{Lemma: MinEigenofRSSMatrix}
Assume the data generating model (\ref{eqn: truemodel}), and that Condition~\ref{cond: eigenvalueV} holds.  Then for sufficiently large $n$, and for any fixed $0 < \tau < 1$, and for any $M \subseteq \spr{1,2,\dots,p}$ satisfying $n > \abs{M} + 2q$,
\begin{align*}
\mathbb{P}_y\fpr{\lambda_{\min}\big(\widehat{\bSigma}_{(M)} \big) \geq \tau(n-\abs{M})\ubar{\lambda}_v}
 >  1- e^{- (1-\sqrt{\tau})^2(n-\abs{M})/2}
\end{align*}
 \end{lemma}

Condition~\ref{cond: supelementsB} specifies that the number of predictors are allowed to grow at sub-exponential rate with the sample size, $n$, ensuring that our method is suitable to perform in the high-dimensional setting.  Model selection consistency with this size of $p$ relative to $n$ is on par with the state-of-the-art results in the literature \citep{bai2018high}. 

\begin{condition}\label{cond: supelementsB} 
For some fixed $\alpha \in (0,1)$, $\log p = o(n^{1-\alpha})$.
\end{condition}

For a given model $M$, if $\abs{M}$ is on the order of $n$, then the row space of $\bX_{M\,\boldsymbol\cdot}$ might span $\Real^{n}$ leading to a rank deficient empirical error covariance matrix, $\bhatSigma_{(M)}$.  Accordingly, $r_\epsilon(M \mid \bY)$ from equation (\ref{eqn: EASMassfunctionmodelM}) will be undefined in this case.  Since we assume that the data arise from the non-degenerate statistical model (\ref{eqn: truemodel}), we must exclude index sets $M$ with $n^\alpha < \abs{M} < n-q$ for some fixed $\alpha \in (0,1)$ arising in Condition \ref{cond: supelementsB}.  Recall, however, that the $h$-function already assigns the value $0$, by definition, to $M$ with $\abs{M} \ge n - q$, and that $q$ is small and fixed.  Throughout the remainder of this section, we assume $\alpha \in (0,1)$ to be some a-priori fixed value with $|M_{\text{o}}| \le n^{\alpha}$. \textcolor{black}{Condition~\ref{cond: supelementsB} is quintessential for LASSO to achieve variable selection consistency in high-dimensional settings \citep[Theorem 3 of][]{zhao2006model}}.  It will be observed in the coming results that the fraction $\alpha$ can be interpreted as a tuning parameter that balances the maximum model size to be considered versus the rate of convergence of the EAS procedure.

Moving along, Condition~\ref{cond: eigenvalX} ensures that the design matrix for the true model has full row rank,
%and that its minimum and maximum squared singular values, as a proportion of the sample size, are finite and bounded away from 0.  
This type of restricted eigenvalue condition is routinely needed in the variable selection literature; \textcolor{black}{e.g., see Condition 6 of \cite{zhao2006model} and \cite{lahiri2021necessary} in the context of LASSO}.

%\begin{condition}[{\em Restricted eigenvalues}]  \label{cond: eigenvalX}
%There exists constants $c_*$ and $c^*$ such that,
%\begin{equation*}
%0 < c_* \leq \lambda_{\min}\fpr{\bX_{M_{\text{o}}\,\boldsymbol\cdot}\bX_{M_{\text{o}}\,\boldsymbol\cdot}^\top/n} < \lambda_{\max}\fpr{\bX_{M_{\text{o}}\,\boldsymbol\cdot}\bX_{M_{\text{o}}\,\boldsymbol\cdot}^\top/n} \leq c^* < \infty, \quad \forall n \in \mathbb{N}.
%\end{equation*}
%\end{condition}

\begin{condition}\label{cond: eigenvalX}
For the true model $M_{\text{o}}$, $\bX_{M_{\text{o}}\,\boldsymbol\cdot}\bX_{M_{\text{o}}\,\boldsymbol\cdot}^\top$ is non-singular.
\end{condition}

%Thus, bounding the maximum eigenvalue of the idempotent matrix $\bI_n - \bH_{(M)}$ by $1$, Condition~\ref{cond: eigenvalX} gives,
%\[
%\lambda_{\max}(\M{\Delta}_{M}) \leq n\lambda_{\max}(\bX_{M_{\text{o}}\,\boldsymbol\cdot}\bX_{M_{\text{o}}\,\boldsymbol\cdot}^\top/n)\lambda_{\max}(\bB_{\boldsymbol\cdot\,M_{\text{o}}}^{0}\bB_{\boldsymbol\cdot\,M_{\text{o}}}^{0\top}) \leq n c^*\norm{\bB_{\boldsymbol\cdot\,M_{\text{o}}}^{0}}^2,
%\]
%so that, by Condition~\ref{cond: supelementsB}, $V_{0,n} \to 0$ as $n \to \infty$ in Lemma \ref{Lemma: MinEigenofRSSMatrix}.  This result is summarized in Corollary~\ref{lemma: MinEigenofRSSMatrixSimpler}, stated next.
%
%
%
%
%\begin{corollary} \label{lemma: MinEigenofRSSMatrixSimpler}
%Assume that the conditions of Lemma~\ref{Lemma: MinEigenofRSSMatrix} are satisfied, and further assume Conditions~\ref{cond: supelementsB} and~\ref{cond: eigenvalX}. Then the $V_{0,n}(M, \alpha_1, \alpha_2)$ in Lemma \ref{Lemma: MinEigenofRSSMatrix} can be replaced by,
%\begin{align*}
%\widetilde{V}_{0,n}(M, \alpha_1, \alpha_2) & := 2e^{-c\spr{\kappa(\alpha_2\ubar{\lambda}_v)\sqrt{n-\abs{M}}-C\sqrt{q}}^2} \\
%& \hspace{.25in}  + \frac{2\;\norm{\bB_{\boldsymbol\cdot\,M_{\text{o}}}^{0}}\sqrt{c^*\bar{\lambda}_v}}{\sqrt{\pi}\alpha_1\ubar{\lambda}_v\sqrt{n-\abs{M}}}\exp\fpr{-\frac{\alpha_1^2\ubar{\lambda}_v^2(n-\abs{M})}{16c^*\norm{\bB_{\boldsymbol\cdot\,M_{\text{o}}}^{0}}^2\bar{\lambda}_v}} \mathrm{I}(\M{\Delta}_M \neq 0).
%\end{align*}
%\end{corollary}

Next, in order to the show that $r_\epsilon(M_{\text{o}} \mid \bY) \to 1$ in probability, we must show that $r_\epsilon(M \mid \bY)/r_\epsilon(M_{\text{o}} \mid \bY) \to 0$ in probability, at a rate vanishing faster than $2^{-n^{\alpha}}$ uniformly for every model $M \neq M_{\text{o}}$ with $\abs{M} \leq n^\alpha$.  Recall that the probability mass function $r_\epsilon(M \mid \bY)$, in equation (\ref{eqn: EASMassfunctionmodelM}), is proportional to a polynomial of the inverse of the determinant of the empirical error covariance matrix.  That being so, we must bound the ratio of determinant, as in Theorem \ref{Lemma: RatioofRSS}.  This ratio is analogous to the ratio of RSS that commonly appears in univariate model selection problems, though, the multivariate situation is much more complicated requiring delicate handling of minimum eigenvalues close to zero.  That being so, this result is interesting in its own right for (high-dimensional) MLR.

For for any model $M$ with $|M| \le n^{\alpha}$, it is understood that the determinant of $\sighat{(M_{\text{o}})}$ as a proportion of that of $\sighat{(M)}$ will behave differently depending on whether $M \subsetneq M_{\text{o}}$ or $M \not \subseteq M_{\text{o}}$. In the first case, when $M \subsetneq M_{\text{o}}$, the ratio will be strictly less than 1 since $M$ is missing at least one oracle predictor.  Conversely, the extreme scenario in the other case is that $M \supset M_{\text{o}}$, in which case the ratio exceeds 1, but by some bound that converges to 1 as $n$ tends to infinity.

\begin{theorem}\label{Lemma: RatioofRSS}
%******************************************************
%    Bounding the ratio of RSS : Fundamental theorem
%******************************************************
Assume Conditions ~\ref{cond: eigenvalueV}, \ref{cond: supelementsB}, and~\ref{cond: eigenvalX}.  Then for sufficiently large $n$, the following approximations hold. 

Case 1: This case pertains to the models $M \subsetneq M_{\text{o}}$. 
{\footnotesize \begin{equation*}
    \mathbb{P}_y\fpr{\bigcap_{M: M \subsetneq M_{\text{o}}} \spr{\M{Y}: \fpr{\frac{\det \sighat{M_{\text{o}}}}{\det \sighat{M}}}^{\frac{n-\abs{M}-q}{2}}  \leq e^{-qn^\alpha\log\abs{M_{\text{o}}}} }} \geq 1 - V_{1,n},
\end{equation*}}
where,
{\footnotesize \begin{align*}
    V_{1,n} := & \abs{M_{\text{o}}}q\exp\fpr{-\frac{\xi_{n,\abs{M_{\text{o}}}}n^\alpha\log\abs{M_{\text{o}}}\ubar{\lambda}_v}{2\bar{\lambda}_v} - 0.09 \fpr{n - \abs{M_{\text{o}}}} + \abs{M_{\text{o}}} \log(\abs{M_{\text{o}}})}  \\
    & \quad +  2\abs{M_{\text{o}}}\exp\fpr{-0.04(n-\abs{M_{\text{o}}}) + \abs{M_{\text{o}}}\log\abs{M_{\text{o}}}} ,
\end{align*}}
with $\xi_{n,\abs{M}} := 1 - \frac{2n^\alpha\log\abs{M_{\text{o}}}}{n - \abs{M} - q}$ such that $\xi_{n,\abs{M}} \in (0,1)$ for large $n$. \\

Case 2: This case pertains to the models $M \not \subseteq M_{\text{o}}$ such that $\abs{M} \leq n^\alpha$.
{\footnotesize \begin{equation*}
    \mathbb{P}_y\fpr{\bigcap_{\substack{M : M \not\subseteq M_{\text{o}} \\ \;\;\quad\abs{M} \leq n^\alpha}} \spr{\M{Y}: \fpr{\frac{\det \sighat{M_{\text{o}}}}{\det \sighat{M}}}^{\frac{n-\abs{M}-q}{2}}  \leq e^{q\fpr{n^{\alpha}\log\fpr{n-\abs{M_{\text{o}}}}+\abs{M}\log p}} }} \geq 1 - V_{2,n},
\end{equation*}}
where,
{\footnotesize\begin{equation*}
    V_{2,n}  := 2q\exp\fpr{- \frac{n^\alpha}{2}\log\fpr{n-\abs{M_{\text{o}}}} + \alpha\log n + \frac{(\abs{M_{\text{o}}}+1)}{2}\log\fpr{\zeta_{n,n^\alpha}} },
\end{equation*}}
and 
$\zeta_{n,\abs{M}} := 1 + \frac{2\fpr{n^{\alpha}\log\fpr{n-\abs{M_{\text{o}}}} + \abs{M}\log p}}{n - \abs{M} - q}$.
\end{theorem}

Observe that both quantities $V_{1,n}$ and $V_{2,n}$ vanish exponentially fast for large $n$, by Condition~\ref{cond: supelementsB}. There are two key facts that we learn from Theorem \ref{Lemma: RatioofRSS}.  The first is that the ratio of the determinants, of the empirical error covariances raised to the power on the order of $n$, will drive $r_\epsilon(M \mid \bY)/r_\epsilon(M_{\text{o}} \mid \bY)$ to 0 for $M \subsetneq M_{\text{o}}$ (i.e., Case 1).  The second fact is that the ratio of the determinants, of the empirical error covariances raised to the power on the order of $n$, will perhaps grow at a sub-exponential rate for large $n$, for $M \not \subseteq M_{\text{o}}$ such that $\abs{M} \leq n^\alpha$ (i.e., Case 2).  As such, the role of the $h$-function is to control the explosive nature of these ratios for models with redundant predictors.  The following two theorems establish that $\mathbb{E}(h_\epsilon(\bB_{\boldsymbol\cdot\,M}))$ and $\mathbb{E}(h_\epsilon(\bB_{\boldsymbol\cdot\,M_{\text{o}}}))$ are adept at accomplishing this task. Sufficient conditions on the choice of $\epsilon$ are stated in the order that they are needed.

\begin{condition}[{\em $\epsilon$-admissibility}] \label{cond: upperboundepsilon}
The size of the true model $\abs{M_{\text{o}}}$ is less than $n^\alpha$. Moreover, for large $n$, the true model $M_{\text{o}}$ satisfies,
{\footnotesize\begin{equation*}
    \frac{1}{36q(n-\abs{M_{\text{o}}})}\left\lVert \fpr{\bV_{(M_{\text{o}})}^{0}}^{-1/2} \fpr{\M{B}^0_{M_{\text{o}}}\bX_{M_{\text{o}}\,\boldsymbol\cdot} - \widetilde{\M{B}}_{\min}\M{X}}\right\rVert_{\textrm{F}}^2 \;\;>\; \epsilon,
\end{equation*}}
where $\widetilde{\bB}_{\min}$ is the solution to the optimization problem,
{\footnotesize\begin{align*}
    \widetilde{\bB}_{\min} :=  \underset{\bB \in \Real^{q \times p}}{\argmin}  \;\left \lVert \fpr{\bV_{(M_{\text{o}})}^{0}}^{-1/2} \fpr{\bB_{\boldsymbol\cdot\,M_{\text{o}}}^{0}\bX_{M_{\text{o}}\,\boldsymbol\cdot} - \bB\bX} \right \rVert_\textrm{F}^2,
\end{align*}}
subject to $\abs{\spr{j: \norm{\bB_j} \neq 0}} \leq \abs{M_{\text{o}}}-1$.
\end{condition}

\textcolor{black}{Condition~\ref{cond: upperboundepsilon} provides the maximum rate of growth for the size of the true model. This is analogous to the {\em sparsity} assumption for the LASSO \cite[Condition 7 of ][]{zhao2006model} }. It also furnishes an upper bound for the choice of $\epsilon$ that is sufficient for the identifiability of the true model and coefficients, as in Definition \ref{defn: hfunction} of the $h$ function.  Notice that given an $\epsilon$, the smaller the norm of the regression coefficient matrix $\M{B}^0_{\cdot\,M_{\text{o}}}$, the more difficult it becomes to identify the true model as $\epsilon$-admissibile.  It is in this sense that $\epsilon$-admissibility defines redundancy both in the sense of correlated predictors and in the sense of predictors with weak signal (after scaling for the response covariance). \textcolor{black}{This is related to the `beta-min' condition discussed for variable selection via LASSO (\cite[Section 7.4 of][]{buhlmann2011statistics} and \cite[Condition 8 of][]{zhao2006model}}). 

%Recall that Condition~\ref{cond: supelementsB} is a sparsity condition imposed on the size of $\norm{\bB_{\boldsymbol\cdot\,M_{\text{o}}}^{0}}$, rather than directly on the size of $|M_{\text{o}}|$.  As such, Condition~\ref{cond: supelementsB} allows for increasing $|M_{\text{o}}|$ so long as the components of $\bB_{\boldsymbol\cdot\,M_{\text{o}}}^{0}$ are diminishing in magnitude, but not necessarily zero.  However, Condition~\ref{cond: upperboundepsilon} rules out the identifiability of $\bB_{\boldsymbol\cdot\,M_{\text{o}}}^{0}$ in that case.  Thus, the combined implications of Conditions \ref{cond: upperboundepsilon} and \ref{cond: supelementsB} for identifiability of the true model, is a dense coefficient matrix $\bB_{\boldsymbol\cdot\,M_{\text{o}}}^{0}$ and an upper bound on the size of $|M_{\text{o}}|$.

With the addition of Condition \ref{cond: upperboundepsilon}, Theorem~\ref{Lemma: Ehtruemodel} ensures that the oracle model is $\epsilon$-admissible.  In our proof strategy, this theorem provides a non-asymptotic probabilistic guarantee that $\mathbb{E}\fpr{h_\epsilon\fpr{\bB_{\boldsymbol\cdot\,M_{\text{o}}}}}$ in the denominator of $r_\epsilon(M\mid\bY)/r_\epsilon(M_{\text{o}}\mid\bY)$ is bounded away from zero, so long as $\epsilon$ is not too large.

\begin{theorem} \label{Lemma: Ehtruemodel}
Assume Conditions~\ref{cond: eigenvalueV} and~\ref{cond: eigenvalX}. Then, for every $\epsilon > 0$ satisfying Condition \ref{cond: upperboundepsilon},
{\footnotesize \begin{align*}
   \mathbb{P}_y\tpr{ \mathbb{E}\fpr{h_\epsilon\fpr{\bB_{\boldsymbol\cdot\,M_{\text{o}}}}} \geq 1 - \exp\fpr{-\frac{\epsilon\fpr{n - \abs{M_{\text{o}}}}}{36} + \frac{q\abs{M_{\text{o}}}}{2}} - 2\exp\fpr{-\frac{1}{8}\spr{\sqrt{n - \abs{M_{\text{o}}}}-2\sqrt{q}}^2}} \\
    \geq 1 - V_{3,n},  \hspace{60 mm} 
\end{align*}}
where,
{ \footnotesize
\begin{align*}
    V_{3,n} &:= \exp\fpr{-\frac{\epsilon\fpr{n-\abs{M_{\text{o}}}}\ubar{\lambda}_{v}}{4\;\bar{\lambda}_{v}} + \frac{q\abs{M_{\text{o}}}}{2}} + \exp\fpr{-0.04(n-\abs{M_{\text{o}}})} +  \exp\fpr{-0.15q(n-\abs{M_{\text{o}}})}. 
\end{align*}}
\end{theorem}

Lastly, to justify that $r_\epsilon(M\mid\bY)/r_\epsilon(M_{\text{o}}\mid\bY) \to 0$ for all redundant models, it remains to establish that $\mathbb{E}\fpr{h_\epsilon\fpr{\bB_{\boldsymbol\cdot\,M}}}$ vanishes rapidly for all models $M \not \subseteq M_{\text{o}}$ with $|M| \le n^{\alpha}$ (recall the cases in Theorem \ref{Lemma: RatioofRSS}). This brings us to the final major supporting result, Theorem \ref{Lemma: Ehbigmodels}, for establishing strong model selection consistency. However, in contrast to the the upper bound condition on $\epsilon$ in Condition~\ref{cond: upperboundepsilon}, a lower bound condition on $\epsilon$ is sufficient to ensure that $h_\epsilon\fpr{\bB_{\boldsymbol\cdot\,M}}$ assigns negligible probability mass to redundant models via $r_\epsilon(M\mid\bY)$.

\begin{condition}[{\em Redundancy}] \label{cond: lowerboundepsilon}
For any model $M$ with $M \not\subseteq M_{\text{o}}$ with $\abs{M} \leq n^\alpha$, for large $n$,
{\footnotesize\begin{equation*}
    \frac{9}{\ubar{\lambda}_v\fpr{n-\abs{M}}}\left \lVert \bB_{\boldsymbol\cdot\,M_{\text{o}}}^{0} \bX_{M_{\text{o}}\,\boldsymbol\cdot}\fpr{\bH_{(M)} - \M{H}_{(M)(-1)}} \right \rVert_{\mathrm{F}}^2 < \epsilon,
\end{equation*}}
where $\bH_{(M)(-1)} := \bH_{(M\backslash \{j^*\})}$ is the projection matrix for the size that is constructed after omitting the predictor $j^*$ from the model $M$ that minimizes
{\footnotesize\begin{align*}
    j^* = \underset{j \in M}{\argmin} \left \lVert \bB_{\boldsymbol\cdot\,M_{\text{o}}}^{0} \bX_{M_{\text{o}}\,\boldsymbol\cdot}\fpr{\bH_{(M)} - \M{H}_{(M\backslash
    \{j\})}} \right \rVert_{\mathrm{F}}^2.
\end{align*}}
\end{condition}

Condition \ref{cond: lowerboundepsilon} is sufficient for showing that $\mathbb{E}(h_\epsilon(\bB_{\boldsymbol\cdot\,M})) \to 0$ in probability for all redundant models, and further characterizes the non-$\epsilon$-admissible notion for redundancy. The quantity on the left side of the condition is the mean difference in the prediction between models $M$ and $M\backslash\{j^*\}$. Condition \ref{cond: lowerboundepsilon} implies that models $M$ with $M \not\subseteq M_{\text{o}}$ are redundant in the sense that they contain at least one predictor whose omission will not change the mean predicted response by more than $\epsilon$, as measured by the properly scaled squared Frobenius norm. \textcolor{black}{This requires that none of the predictors in model $M\setminus M_{\text{o}}$ can be replaced by some predictors in the true model to provide a significantly better prediction than $\epsilon$ (in the appropriate scale). Intuitively, this means that the correlations between the predictors in the true model and the ones not in the true model cannot be large. Resembling the notion of the {\em irrepresentability} condition necessary for LASSO model selection consistency \citep{zhao2006model}, this is to say that the irrelevant covariates cannot be well-represented by any of the covariates in the true model.}

Condition~\ref{cond: upperboundepsilon} coupled with Condition~\ref{cond: lowerboundepsilon} provides the crucial interval for the choice of $\epsilon$ within which the oracle model is identifiable and the EAS procedure achieves strong model selection consistency.  Notably, due to the appropriate scaling of quantities in the $h$ function, this interval neither depends on the sample size $n$ nor the size of the model $M$. 

\begin{theorem}\label{Lemma: Ehbigmodels}
Assume Conditions~\ref{cond: eigenvalueV},~\ref{cond: supelementsB}, and~\ref{cond: eigenvalX}. Then, for sufficiently large $n$, and for every $\epsilon > 0$ satisfying Condition~\ref{cond: lowerboundepsilon},
{\footnotesize \begin{align*}
    \mathbb{P}_y\tpr{\bigcap_{\substack{M \not\subseteq M_{\text{o}} \\ \abs{M} \leq n^\alpha}} \spr{\bY:\mathbb{E}\fpr{h_\epsilon\fpr{\bB_{\boldsymbol\cdot\,M}}} \leq \exp\fpr{-\frac{\epsilon\fpr{n - \abs{M}}}{36} +  \frac{q\abs{M}}{2}} + 2\exp\fpr{-\frac{1}{8}\spr{\sqrt{n - \abs{M}}-2\sqrt{q}}^2}}} \\
    \geq 1 - V_{4,n}, \hspace{50 mm}
\end{align*} }
where,
{\footnotesize 
\begin{align*}
    V_{4,n} &:= n^\alpha\spr{\exp\fpr{-\frac{\epsilon\fpr{n-n^\alpha}\ubar{\lambda}_v}{36\;\bar{\lambda}_v} + \frac{qn^\alpha}{2} + n^\alpha\log p} + 2\exp\fpr{ -0.04(n-n^{\alpha}) + n^\alpha\log p}}.
\end{align*}}
\end{theorem}

Theorem~\ref{Lemma: Ehbigmodels} is a non-asymptotic concentration bound for $\mathbb{E}\fpr{h_\epsilon\fpr{\bB_{\boldsymbol\cdot\,M}}}$ that applies uniformly over all model $M \not\subseteq M_{\text{o}}$.  This is the critical theoretical aspect of the $h$ function that compensates for the explosive nature of the ratios of the determinants of the empirical error covariances raised to the power on the order of $n$, uniformly over all models $M \not \subseteq M_{\text{o}}$ such that $\abs{M} \leq n^\alpha$, as exhibited in Case 2 of Theorem \ref{Lemma: RatioofRSS}.  

To this point in the article, sufficient analysis has be constructed to argue the {\em pairwise} model selection consistency result that $r_\epsilon(M\mid\bY)/r_\epsilon(M_{\text{o}}\mid\bY) \to 0$ in probability for any $M \ne |M_{\text{o}}|$ with $\abs{M} \leq n^\alpha$.  For the case when $p$ is fixed, this also implies {\em strong} model selection consistency.  In the case when $p \to \infty$ and particularly for $p \gg n$, however, further justification is required because the number of candidate models to consider is $2^{n^{\alpha}}$.  Theorems \ref{Lemma: Ehtruemodel} and \ref{Lemma: Ehbigmodels} are able to manage this exponential-sized class of candidate model with the essential attribute that they provide concentration inequalities of tails that are uniform and vanish exponentially fast in $n$.  This fact is stated as our main result, Theorem \ref{Lemma: Mainresult}.  

\begin{theorem}\label{Lemma: Mainresult}
Assume the data generating model (\ref{eqn: truemodel}), and suppose that Conditions \ref{cond: eigenvalueV}, \ref{cond: supelementsB} and~\ref{cond: eigenvalX} are satisfied. Then for every $\epsilon > 0$ satisfying Conditions~\ref{cond: upperboundepsilon} and~\ref{cond: lowerboundepsilon},
\begin{equation*}
    \frac{r_\epsilon\fpr{M_{\text{o}} \vert \M{Y}}}{\sum_{M: \abs{M} \leq n^\alpha} r_\epsilon\fpr{M \vert \M{Y}}} \overset{\mathbb{P}_y}{\longrightarrow} 1,
\end{equation*}
as $n \to \infty$ or $n, p \to \infty$.
\end{theorem}

The proof of Theorem \ref{Lemma: Mainresult} and the proofs of all other results are organized in the Supplementary Material.  Note that Theorem \ref{Lemma: Mainresult} is the only non-asymptotic result in our theoretical developments, and so as long as the conditions are satisfied, it is expected that it is reasonably illustrative of the performance of our constructed EAS procedure on observed data.  We provide evidence to substantiate this claim in finite sample numerical studies, presented next in Section~\ref{sec: EASsimstudy}.

\section{Numerical Results}\label{sec: EASsimstudy}

In this section we demonstrate the performance of our EAS method in comparison to the state-of-the-art variable selection procedures for MLR. Very recently \cite{bai2018high} developed the MBSP method that is equipped to perform variable selection for MLR. They demonstrate a distinctly superior performance of MBSP over all the existing methods, especially in a high-dimensional setting.  To make standard the comparison between the MBSP and EAS approaches, we mimic the exact same synthetic data simulation study design constructed in \cite{bai2018high}. 

The simulation design can be broadly categorized into three parts, { \em low dimensional (LD) } ($n > p$), {\em high-dimensional (HD)} $(p > n)$ and {\em ultra high-dimensional (UHD) } ($p \gg n$).  Two sub-categories are considered within each of these categories,  to analyze performance for varying sizes of the true model, $\abs{M_{\text{o}}}$, versus the total number of predictors, $p$. Within each category the dimension of the multivariate response, $q$, is also varied to study the effect of $q$ on the model selection performance. In total, there are six experiments, summarized in Table~\ref{tab: simulationdesign}.

For each of the first six simulation designs, we generate synthetic data by the following mechanism: The $n$ columns of the design matrix $\bX$ are sampled from a multivariate normal distribution with mean zero and covariance matrix $\bGamma$, that has an AR$(1)$ structure with correlation coefficient $0.5$ (i.e., $\Gamma_{ij} = 0.5^{\abs{i-j}}$, for $i,j \in \{1, \dots, p\}$). The true model $M_{\text{o}}$ is constructed by randomly selecting $\abs{M_{\text{o}}}$ elements from $\{1, \dots, p\}$. Once the true model $M_{\text{o}}$ is constructed, each component of the $q \times \abs{M_{\text{o}}}$ true regression coefficient matrix $\bB_{\boldsymbol\cdot\,M_{\text{o}}}^{0}$ is set as a value generated from the random variable $U + \mathrm{I}(U > -0.5)$, with $U \sim \textrm{Uniform}(-5,4)$, so that the values always lie within $[-5,-0.5] \cup [0.5,5]$. The response vectors $\bY_1,\dots,\bY_n$ are independently generated from a multivariate normal distribution with mean $\bB_{\boldsymbol\cdot\,M_{\text{o}}}^{0}\bX_{M_{\text{o}}\,\boldsymbol\cdot}$ and covariance $\bV_{(M_{\text{o}})}^{0}$, where $\bV_{(M_{\text{o}})}^{0}$ also has an AR$(1)$ structure with $\bV_{(M_{\text{o}}), ij}^0 = \sigma^20.5^{\abs{i-j}}$ for $i,j \in \{1, \dots, q\}$ and $\sigma^2 = 2$.

\begin{table}[H]
\centering
\begin{tabular}{l l l l l l}
                       Dimension         & Sparsity        & $n$  & $p$  & $q$ & $\abs{M_{\text{o}}}$ \\ \hline \hline
\multirow{2}{*}{LD $(n > p)$ }      & Sparse         & 60     & 30    & 3    & 5           \\ 
                                                    & Dense           & 80    & 60    & 6    & 40          \\ \hline
\multirow{2}{*}{HD $(p > n)$}      & Sparse          & 50    & 200  & 5    & 20          \\ 
                                                    & Dense           & 60    & 100  & 6    & 40          \\ \hline
\multirow{2}{*}{UHD $(p \gg n)$} & Ultra-sparse  & 100  & 500  & 3    & 10          \\ 
                                                    & Sparse           & 150 & 1000 & 4    & 50          \\ \hline \hline
\end{tabular}
\caption{\footnotesize $n$ is the sample size, $p$ is the number of predictors, $q$ is the dimension of multivariate response vectors, and $\abs{M_{\text{o}}}$ is the size of the true model.}
\label{tab: simulationdesign}
\end{table}

%\textcolor{black}{Additional to the design considered in \cite{bai2018high}, we consider further experiments to demonstrate the performance of EAS under i) large $q$, ii) more difficult design matrix, and iii) dense error covariance matrix, as described in Table 2. At first, we increase the dimension of $q$ to $60$, which is ten times higher than what we considered in Table \ref{tab: simulationdesign}, while keeping the covariance matrix of the columns of design, $\bGamma$, and the error covariance matrix $\bV^{0}_{M_{\text{o}}}$ fixed at AR(1) (same as Table~\ref{tab: simulationdesign}), In the second experiment, we additionally increase dificulty in the design matrix by setting a non-decaying structure of $\bGamma$, while keeping the error covariance $\bV^{0}_{(M_{\text{o}})}$ fixed at AR(1), followed by a dense covariance structure of the error covariance in the third experiment. All the parameters of the design and the covariance matrices are presented in Table \ref{tab: additionalsimulationdesign}}.

\textcolor{black}{In order to further investigate the performance of the EAS method under challenging scenarios, we conduct additional experiments beyond those presented in \cite{bai2018high}. Specifically, we consider three additional experimental settings, described in Table \ref{tab: additionalsimulationdesign}. In the first experiment, we increase the dimension of the response variable to $q=60$, which is ten times larger than the dimension considered in Table~\ref{tab: simulationdesign}.  In the second experiment, we increase the difficulty of the design matrix by introducing a non-decaying structure for $\bGamma$, while keeping the error covariance $\bV^{0}_{(M_{\text{o}})}$ fixed at AR(1). In the third experiment, we introduce a dense structure for the error covariance matrix.}

%\begin{table}[H]
%\centering
%\begin{tabular}{l l l l l l l}
%   &  $n$  & $p$  & $q$ & $\abs{M_{\text{o}}}$ & $\bGamma_{ij}$ & $\bV_{M_{\text{o}}, ij}^0$ \\ \hline \hline
%  Experiment 7 & 150  & 1000  & 60    & 50  & $0.5^{\abs{i-j}}$ &  $2 \times 0.5^{\abs{i-j}}$   \\ 
%  Experiment 8 & 150  & 1000  & 60    & 50  & $0.5\{1 + \mathbb{I}(i = j)\}$ & $2 \times 0.5^{\abs{i-j}}$ \\
%  Experiment 9 & 150  & 1000  & 60    & 50  & $0.5\{1 + \mathbb{I}(i = j)\}$ & $1 + \mathbb{I}(i = j)$ 
% \end{tabular}
% \end{table}
\begin{table}[H]
\textcolor{black}{
\centering
\resizebox{\textwidth}{!}{%
\begin{tabular}{ccccc}
\multirow{2}{*}{Dimension} & \multirow{2}{*}{Set up} & \multirow{2}{*}{q} & \multirow{2}{*}{$\bGamma_{ij}$} & \multirow{2}{*}{$\bV_{(M_{\text{o}}), ij}^0$} \\
 &  &  &  &  \\ \hline \hline 
\multicolumn{1}{c|}{\multirow{7}{*}{\begin{tabular}[c]{@{}c@{}}$n = 150$,\\ \\ $p = 1000$, \\ \\ $M_{\text{o}} = 50$\end{tabular}}} & Large $q$ & 60 & $0.5^{\abs{i-j}}$ & $2 \times 0.5^{\abs{i-j}}$ \\ \cline{2-5} 
\multicolumn{1}{c|}{} & \multirow{3}{*}{\begin{tabular}[c]{@{}c@{}}Large $q$ and\\ Non-decaying correlation of \\ design matrix\end{tabular}} & \multirow{3}{*}{60} & \multirow{3}{*}{$0.5\{1 + \mathbb{I}(i = j)\}$} & \multirow{3}{*}{$2 \times 0.5^{\abs{i-j}}$} \\
\multicolumn{1}{c|}{} &  &  &  &  \\
\multicolumn{1}{c|}{} &  &  &  &  \\ \cline{2-5} 
\multicolumn{1}{c|}{} & \multirow{3}{*}{\begin{tabular}[c]{@{}c@{}}Large $q$, non-decaying correlation \\ of design matrix, and dense  \\ error covariance matrix\end{tabular}} & \multirow{3}{*}{60} & \multirow{3}{*}{$0.5\{1 + \mathbb{I}(i = j)\}$} & \multirow{3}{*}{$1 + \mathbb{I}(i = j)$} \\
\multicolumn{1}{c|}{} &  &  &  &  \\
\multicolumn{1}{c|}{} &  &  &  &  \\ \hline
\end{tabular}%
}
\caption{\footnotesize The notation is defined as follows: $n$ is the sample size, $p$ is the number of predictors, $q$ is the dimension of multivariate response vectors, $\abs{M_{\text{o}}}$ is the size of the true model, and $\bGamma_{ij}$ and $\bV_{(M_{\text{o}}), ij}^0$ are $(i,j)$th element of $\bGamma$ (i.e., the variance of the columns of the design matrix of $\bX$) and $\bV_{(M_{\text{o}})}^0$ (i.e., the error covariance matrix), respectively.}
\label{tab: additionalsimulationdesign}
}
\end{table}

The EAS procedure is implemented by computing Algorithm~\ref{algo: MCMC} described in Section~\ref{sec: EAScomputation} to draw MCMC samples from the space of all $2^p$ candidate models. Observe that Algorithm~\ref{algo: MCMC} is developed to work for a fixed $\epsilon$.  We propose two methods for selecting the tuning parameter $\epsilon$; (i) 10-fold CV and (ii) via Bayesian IC (BIC) by searching over a pre-specified grid of $\epsilon$ values. For both the CV and BIC routines, we take a uniform grid of 24 possible values for $\epsilon$, from 0.05 to 10 in all six experiments.  In the CV procedure, for each of the 10 folds we implement our EAS method on the training set by running the MCMC $500$ steps, discarding the first $200$ steps, and evaluating the performance on the validation set, as follows. The initial estimates from the multivariate LASSO (\verb|MLASSO|) \citep{glmnet} serve as the weights for proposing/removing predictors in the MCMC algorithm.  The expectation of $h_\epsilon(\bB_{\boldsymbol\cdot\,M})$ for any model $M$ is approximated by evaluating the $h$ at the least square estimator $\widehat{\bB}_{\boldsymbol\cdot\,M}$, which makes the computation very fast compared to the previous version of EAS procedures. The MAP estimated model from the MCMC sample is taken as a point estimator to compute the mean squared prediction error (MSPE) on the validation set.  The optimal $\epsilon$ is chosen as the one that minimizes the average of the MSPE over the 10 folds.  Finally, we re-run Algorithm~\ref{algo: MCMC} on the entire dataset using the optimal selected $\epsilon$ for 10,000 MCMC steps and discard the first 5,000.

For the BIC procedure, the computational cost is much less. In this case, for every $\epsilon$ in the grid, we run Algorithm~\ref{algo: MCMC} for 5,000 steps, discard the initial 2,000 in obtaining the MAP estimated model, and compute the BIC for the MAP model.  The $\epsilon$ corresponding to the minimum BIC value is selected as optimal. The advantage of using BIC is that we do not need to run the algorithm again for the optimally chosen $\epsilon$, we can simply use the MCMC chain from the initial runs as our estimated sample for the chosen $\epsilon$. 
 
We compare the performance of our EAS method with (1) the \verb|MBSP| approach as implemented in the R package MBSP \citep{mbsp}; (2) the multivariate group lasso with spike and slab prior (\verb|MBGL-SS|) method as implemented in the R package MBSGS \citep{mbsgs} with the natural grouping (i.e., each predictor represents one group); (3) the sparse reduced rank regression (\verb|SRRR|) method as implemented in the R package rrpack \citep{rrpack}, with pre-specified rank $q$ and adaptive group LASSO penalty; (4) the sparse partial least squares (\verb|SPLS|) approach as implemented in the R package spls \citep{spls}, with the thresholding parameter $\eta$ selected by CV, and the number of hidden components is set as $q$; (5) the multivariate sparse group LASSO (\verb|MSGLASSO|) method as implemented via R package MSGLasso \citep{msglasso} with each predictor representing its own group; (6) the \verb|MLASSO| method as implemented via the glmnet package \citep{glmnet} that penalizes the norm of each of the columns of the coefficient matrix; %The objective function for \verb|MLASSO| is specified in Section 4 of \cite[][]{liquet2017bayesian} 
and \textcolor{black}{(7) Multivariate square-root grouped LASSO} (\verb|MSRL|) \textcolor{black}{with each predictor representing its own group \cite{Molstad2022New}, with tuning parameter selected via 5-fold CV and the range of candidate tuning parameters, $\delta$, set at $0.1$.}   

For the frequentist procedures, the estimated model is defined to be the non-zero columns of the estimated coefficient matrix.  For \verb|MBGL-SS| the median thresholding estimator is used, and for \verb|MBSP| the coefficients selected in the estimated model are those for which the $95\%$ credible interval does not contain $0$. For our EAS method we take the least square estimator of the MAP model as the point estimator for $\bB_{\boldsymbol\cdot\,M_{\text{o}}}^{0}$. 

\textcolor{black}{The metrics we use to evaluate the performance of the various methods, over 1,000 synthetic data sets for each of the six experiments in Table~\ref{tab: simulationdesign} and three experiments in Table~\ref{tab: additionalsimulationdesign}, are the following.  We report median MSPE on an out-of-sample test set, $\bY_{new}$, that is of the same size as $\bY$.  We also report the average false discovery rate (FDR), the average false negative rate (FNR), average mis-classification probability (MP), average proportion of correct model selection (PCM), and median computation time. The results of the experiments in Table~\ref{tab: simulationdesign} are displayed in Table~\ref{tab: Simresults}, and the results of the experiments in Table~\ref{tab: additionalsimulationdesign} are displayed in Table~\ref{tab: Simresults_largeq}}. The results are a bit less noisy than those reported in the simulation study in \cite{bai2018high} because they only generated $100$ synthetic datasets, and they did {\em not} report the out-of-sample prediction performance.

The explicit formulas for computing the metrics are, $\textrm{MSE} := \lVert \bY- \widehat{\bY}\rVert_\textrm{F}^2/nq$, $\textrm{MSPE} := \lVert \bY_{new}- \widehat{\bY}_{new}\rVert _\textrm{F}^2/nq$,  $\textrm{FDR} := \textrm{FP}/(\textrm{FP} + \textrm{TP})$, $\textrm{FNR} := \textrm{FN}/(\textrm{FN} + \textrm{TN})$, and $\textrm{MP} := (\textrm{FP} + \textrm{FN})/pq$, where TP, FP, TN, and FN are, respectively, the number of true positives, false positives, true negatives, and false negatives.  Moreover, we also present the average estimated posterior probability of the true model, denoted $\mathbb{P}(M_{\text{o}} \mid \bY)$, for the Bayesian procedure and average fiducial probability, $r_\epsilon(M_{\text{o}} \mid \bY)$, for the EAS procedure. Note that neither $r_\epsilon(M_{\text{o}} \mid \bY)$ nor $\mathbb{P}(M_{\text{o}} \mid \bY)$ can be calculated for the frequentist methods, or the MAP or credible region based Bayesian methods, like \verb|MBSP|. 

From Table \ref{tab: Simresults}, in the generic $n > p$ case for both sparsity levels of the true model, irrespective of whether $\epsilon$ is chosen based on BIC or CV, our EAS method performs on par with all other methods in terms of predictive performance, except SPLS which tends to exhibit inferior level of accuracy. In terms of variable selection performance, our EAS method chooses the correct model with a high probability and very low FDR and FNR, similar to the Bayesian methods \verb|MBGL-SS|, \verb|MBSP| and frequentist method \verb|SRRR|. Other frequentist procedures tend to exhibit a lot of false positives.  our EAS method does an excellent job in assigning a very high (GF) probability to the true model.  Moreover, it is an advantage of our EAS method, and \verb|MBGL-SS|, that they provide a probabilistic assessment of the competing models $M$ so that inference can be made on how much better, say the MAP estimated model is from the second best model and so on.  For instance, if there are many models that are assigned similar probabilities, then the practitioner is warned not to over-interpret inference based on a single model.  \textcolor{black}{This situation would possibly happen if there is sufficient collinearity (as defined by $\epsilon$) among the important predictors.  In that case, it is not reasonable to think that there is a unique choice of correct model (for fixed sample size) or that the practitioner should choose a single model.  For example, if two covariates are perfectly correlated, then the notion of {\em best subset} of the two covariates is not meaningful.  Statistical inference guides data driven decisions, but such inference should also have the capability to suggest when there is not enough information in the data to make a decision; this embodies the EAS approach to model selection, and it also enables the EAS method to meaningfully be applied even in the absence of an underlying sparse data generating structure.}

Continuing on Table \ref{tab: Simresults}, in the $p > n$ and $p \gg n$ scenarios, our EAS method does fulfill the expectations consistent with the strong model selection consistency. This is demonstrated by the fact that the EAS procedure, either using the BIC or CV tuning selection procedures, outperforms \verb|MBSP| in terms of both the prediction performance and the average proportion of correct model selections. It also assigns a very high GF probability to the true model, which, again is consistent with our strong model selection consistency theoretical result. 

\begin{table}[H]
\centering
\scriptsize
\begin{tabular}{lrrrrrrrrrrrrr}
\hline
\hline
 \textbf{Method} &  \textbf{MSPE} & & \textbf{FDR} & & \textbf{FNR} &&  \textbf{MP} & & $\mathbb{P}(M_{\text{o}} \mid \bY)$ & & \textbf{PCM} && \begin{tabular}[c]{@{}l@{}} \textcolor{black}{\textbf{time}} \\  \textcolor{black}{(in sec)} \end{tabular} \\
 \hline 
 \hline
 \multicolumn{14}{c}{LD $(n > p)$, sparse : $n = 60, p = 30, q = 3, \abs{M_{\text{o}}} = 5$} \\
 \hline
EAS-BIC & 2.17 & & 0.0096 & & 0 & & 0.0007 & & 0.952 & & 0.952 && \textcolor{black}{143} \\
EAS-CV & 2.17 & & 0.0148 & & 0 & & 0.0014 & & 0.951 & & 0.951 && \textcolor{black}{714} \\
MBGL-SS & 2.17 & & 0.0042 & & 0 & & 0.0003 & & 0.903 & & 0.976 && \textcolor{black}{92}\\
MBSP & 2.28 & & 0.0194 & & 0 & & 0.0013 & & N/A & & 0.891 && \textcolor{black}{20}\\
MLASSO & 2.47 & & 0.6683 & & 0 & & 0.1245 & & N/A & & 0 && \textcolor{black}{0.6}\\
MSGLASSO & 2.65 & & 0.4393 & & 0 & & 0.0553 & & N/A & & 0.032 && \textcolor{black}{0.6}\\
SPLS & 6.66 & & 0.1419 & & 0.023 & & 0.0178 & & N/A & & 0.175 && \textcolor{black}{2.9}\\
SRRR & 2.17 & & 0.0113 & & 0 & & 0.0008 & & N/A & & 0.939 && \textcolor{black}{0.1}\\
\textcolor{black}{MSRL} & \textcolor{black}{2.49} & & \textcolor{black}{0.6603} & & \textcolor{black}{0} & & \textcolor{black}{0.1197} & & \textcolor{black}{N/A} & & \textcolor{black}{0} && \textcolor{black}{21} \\
\hline
\multicolumn{14}{c}{LD $(n > p)$, dense : $n = 80, p = 60, q = 6, \abs{M_{\text{o}}} = 40$} \\
\hline
EAS-BIC & 4.14 & & 0.0129 & & 0 & & 0.0021 & & 0.9309 & & 0.932 && \textcolor{black}{343}\\
EAS-CV & 4.12 & & 0 & & 0.0004 & & 0 & & 0.997 & & 0.997 &&  \textcolor{black}{1944}\\
MBGL-SS & 4.10 & & 0.0007 & & 0 & & 0.0001 & & 0.9383 & & 0.972 && \textcolor{black}{358}\\
MBSP & 4.25 & & 0.0028 & & 0 & & 0.0003 & & N/A & & 0.897 && \textcolor{black}{30}\\
MLASSO & 5.94 & & 0.331 & & 0 & & 0.055 & & N/A & & 0 && \textcolor{black}{1.1}\\
MSGLASSO & 6.23 & & 0.3118 & & 0 & & 0.0505 & & N/A & & 0 && \textcolor{black}{3.0}\\
SPLS & 145.08 & & 0.2738 & & 0.2291 & & 0.0467 & & N/A & & 0 && \textcolor{black}{7}\\
SRRR & 4.14 & & 0.0033 & & 0 & & 0.0004 & & N/A & & 0.886 && \textcolor{black}{0.4}\\
\textcolor{black}{MSRL} & \textcolor{black}{5.69} & & \textcolor{black}{0.3282} & & \textcolor{black}{0} & & \textcolor{black}{0.0543} & & \textcolor{black}{N/A} & & \textcolor{black}{0} && \textcolor{black}{32}\\
\hline
\multicolumn{14}{c}{HD $(p > n)$, sparse : $n = 50, p = 200, q = 5, \abs{M_{\text{o}}} = 20$} \\
\hline
EAS-BIC & 3.40 & & 0.0351 & & 0.0002 & & 0.0015 & & 0.8798 & & 0.882 && \textcolor{black}{368}\\
EAS-CV & 3.36 & & 0.0075 & & 0.0005 & & 0.0004 & & 0.9509 & & 0.951&&  \textcolor{black}{1811}\\
MBGL-SS & 52.29 & & 0.5503 & & 0.0076 & & 0.0561 & & 0.2029 & & 0.232 && \textcolor{black}{3158}\\
MBSP & 4.43 & & 0.0124 & & 0 & & 0.0003 & & N/A & & 0.778 && \textcolor{black}{95}\\
MLASSO & 15.85 & & 0.7818 & & 0.0001 & & 0.0721 & & N/A  & & 0 && \textcolor{black}{1.1}\\
MSGLASSO & 21.10 & & 0.737 & & 0.0049 & & 0.0568 & & N/A & & 0 && \textcolor{black}{7.3}\\
SPLS & 112.05  & & 0.5214 & & 0.0401 & & 0.0307 & & N/A & & 0 && \textcolor{black}{9.2}\\
SRRR & 15.54 & & 0.7448 & & 0.0016 & & 0.0582  & & N/A  & & 0 && \textcolor{black}{7.1}\\
\textcolor{black}{MSRL} & \textcolor{black}{16.69} & & \textcolor{black}{0.7885} & & \textcolor{black}{0.0002} & & \textcolor{black}{0.0758} & & \textcolor{black}{N/A} & & \textcolor{black}{0} && \textcolor{black}{65}\\
\hline
 \multicolumn{14}{c}{HD $(p > n)$, dense : $n = 60, p = 100, q = 6, \abs{M_{\text{o}}} = 40$} \\
 \hline
EAS-BIC & 6.91 & & 0.0589 & & 0.0135 & & 0.0065 & & 0.6673 & & 0.669 && \textcolor{black}{441}\\
EAS-CV & 6.45 & & 0.0203 & & 0.0089 & & 0.0026  & & 0.8599 & & 0.858 && \textcolor{black}{2252}\\
MBGL-SS & 6.31 & & 0.0026 & & 0.0002 & & 0.0002 & & 0.7697 & & 0.936 && \textcolor{black}{861}\\
MBSP & 9.78 & & 0.0297 & & 0.0002 & & 0.0021 & & N/A & & 0.339 && \textcolor{black}{56}\\
MLASSO & 34.86 & & 0.5281 & & 0.0006 & & 0.0748 & & N/A & & 0  && \textcolor{black}{1.3}\\
MSGLASSO & 32.95 & & 0.5265  & & 0.001 & & 0.0748 & & N/A & & 0 && \textcolor{black}{6.3}\\
SPLS & 197.30 & & 0.4721 & & 0.1437 & & 0.0659 & & N/A & & 0  && \textcolor{black}{8.7}\\
SRRR & 22.03 & & 0.474 & & 0.0031 & & 0.0603 & & N/A & & 0 && \textcolor{black}{2.7}\\
\textcolor{black}{MSRL} & \textcolor{black}{34.48} & & \textcolor{black}{0.5511} & & \textcolor{black}{0.001} & & \textcolor{black}{0.0821} & & \textcolor{black}{N/A} & & \textcolor{black}{0} && \textcolor{black}{40}\\
\hline
 \multicolumn{14}{c}{UHD $(p \gg n)$, ultra-sparse : $n = 100, p = 500, q = 3, \abs{M_{\text{o}}} = 10$} \\
 \hline
EAS-BIC & 2.23 & & 0.0143 & & 0 & & 0.0004 & & 0.9609 & & 0.9609 && \textcolor{black}{241}\\
EAS-CV & 3.98 & & 0.6135 & & 0 & & 0.0222 & & 0.2116 & & 0.2116 && \textcolor{black}{1645}\\
MBGL-SS & 2.23 & & 0.0032 & & 0 & & 0.0001 & & 0.7691 & & 0.9729 && \textcolor{black}{18901}\\
MBSP & 2.86 & & 0.0666  & & 0 & & 0.0005 & & N/A & & 0.5366 && \textcolor{black}{491}\\
MLASSO & 3.09 & & 0.841 & & 0 & & 0.0396 & & N/A & & 0 && \textcolor{black}{1.3}\\
MSGLASSO & 16.71 & & 0.7683 & & 0.0013 & & 0.0228 & & N/A & & 0 && \textcolor{black}{7.7}\\
SPLS & 27.54 & & 0.2405 & & 0.0055 & & 0.0039 & & N/A & & 0.001 && \textcolor{black}{14}\\
SRRR & 5.91 & & 0.9363 & & 0 & & 0.0981 & & N/A & & 0 && \textcolor{black}{20}\\
\textcolor{black}{MSRL} & \textcolor{black}{3.10} & & \textcolor{black}{0.7890} & & \textcolor{black}{0} & & \textcolor{black}{0.0277} & & \textcolor{black}{N/A} & & \textcolor{black}{0} && \textcolor{black}{163}\\
\hline
 \multicolumn{14}{c}{UHD $(p \gg n)$, sparse : $n = 150, p = 1000, q = 4, \abs{M_{\text{o}}} = 50$} \\
 \hline
EAS-BIC & 3.00 & & 0 & & 0 & & 0 & & 0.9944 & & 0.995 && \textcolor{black}{1144}\\
EAS-CV & 3.00 & & 0.0013 & & 0 & & 0.0001 & & 0.9913 & & 0.991 && \textcolor{black}{7095}\\
MBGL-SS & 354.19 & & 0.932 & & 0.0175 & & 0.1532 & & 0 & & 0 && \textcolor{black}{116522}\\
MBSP & 3.46 & & 0.0026 & & 0 & & 0 & & N/A & & 0.871 && \textcolor{black}{2316}\\
MLASSO & 19.00 & & 0.8143 & & 0 & & 0.055 & & N/A & & 0 && \textcolor{black}{2.3}\\
MSGLASSO & 90.39 & & 0.7995 & & 0.0081 & & 0.0455 & & N/A & & 0 && \textcolor{black}{74}\\
SPLS & 306.63 & & 0.5796 & & 0.0221 & & 0.0193 & & N/A & & 0 && \textcolor{black}{91}\\
SRRR & 34.99 & & 0.8052 & & 0.0006 & & 0.0514 & & N/A & & 0 && \textcolor{black}{88}\\
\textcolor{black}{MSRL} & \textcolor{black}{17.82} & & \textcolor{black}{0.8323} & & \textcolor{black}{0} & & \textcolor{black}{0.0643} & & \textcolor{black}{N/A} & & \textcolor{black}{0} && \textcolor{black}{346}\\
\hline
\end{tabular}\caption{\footnotesize Performance of EAS compared other methods over 1,000 replications.}
\label{tab: Simresults}
\vspace{-.1in}
\end{table}

The selection performance of the frequentist procedures seems to degrade rapidly when one moves from small to large $p$. They tend to select a lot of false signals in the estimated model. An interesting remark in support of \verb|MBGL-SS| is that it seems to perform on par with the EAS procedure and outperform \verb|MBSP| in some of the high-dimensional designs considered, arguably when either $p$ is not so much larger than $n$ or when the true model is ultra-sparse. This contradicts the numerical results presented in \cite{bai2018high} where \verb|MBGL-SS| is shown to perform poorly in all of the high-dimensional scenarios. In the $p \gg n$ and the ultra-sparse, when the optimal $\epsilon$ is chosen by CV, it seems to commit a lot of false discoveries. This might be attributed to the lack of identifiability of the true model in some of the folds, thus choosing an $\epsilon$ that it is relatively smaller than the optimal one. However, choosing $\epsilon$ through BIC seems to mitigate these computational bottlenecks observed with the CV procedure.

\begin{table}[H]
\textcolor{black}{
\centering
\scriptsize
\begin{tabular}{lrrrrrrrrrrrrr}
\hline
\hline
  \textcolor{black}{\textbf{Method}} &  \textcolor{black}{\textbf{MSPE}} & &  \textcolor{black}{\textbf{FDR}} & &  \textcolor{black}{\textbf{FNR}} &&   \textcolor{black}{\textbf{MP}} & &  \textcolor{black}{$\mathbb{P}(M_{\text{o}} \mid \bY)$} & &  \textcolor{black}{\textbf{PCM}} && \begin{tabular}[c]{@{}l@{}} \textcolor{black}{\textbf{time}} \\  \textcolor{black}{(in sec)} \end{tabular} \\
 \hline 
 \hline \\
 \multicolumn{14}{c}{\textcolor{black}{Large dimension of response: $q = 60$}} \\ 
\multicolumn{14}{c}{\textcolor{black}{AR(1) correlation for the column of design matrix $\bX$: $\Gamma_{ij} = 0.5^{\abs{i-j}}$}} \\
\multicolumn{14}{c}{\textcolor{black}{AR(1) error covariance matrix: $\bV_{(M_{\text{o}}), ij}^0 = 2 \times 0.5^{\abs{i-j}}$}} \\
 \hline
 \textcolor{black}{EAS-BIC} & \textcolor{black}{3.01} & & \textcolor{black}{0.000} & & \textcolor{black}{0.000} & & \textcolor{black}{0.000} & & \textcolor{black}{0.99} & & \textcolor{black}{1} & & \textcolor{black}{4680} \\
 \textcolor{black}{EAS-CV} & \textcolor{black}{3.01} & & \textcolor{black}{0.000} & & \textcolor{black}{0.000} & & \textcolor{black}{0.000} & & \textcolor{black}{0.99} & & \textcolor{black}{1} & & \textcolor{black}{22896} \\
 \textcolor{black}{MBSP} & \textcolor{black}{3.79} & & \textcolor{black}{0.054} & & \textcolor{black}{0.000} & & \textcolor{black}{0.000} & & \textcolor{black}{N/A} & & \textcolor{black}{0.09} && \textcolor{black}{4587}\\
 \textcolor{black}{MLASSO} & \textcolor{black}{6.52} & & \textcolor{black}{0.914} & & \textcolor{black}{0.000} & & \textcolor{black}{0.009} & & \textcolor{black}{N/A} & & \textcolor{black}{0.00} && \textcolor{black}{40}\\
 \textcolor{black}{MSGLASSO} & \textcolor{black}{8.37} & & \textcolor{black}{0.871} & & \textcolor{black}{0.001} & & \textcolor{black}{0.006} & & \textcolor{black}{N/A} & & \textcolor{black}{0.00} && \textcolor{black}{618}\\
 \textcolor{black}{SPLS} & \textcolor{black}{3.43} & & \textcolor{black}{0.167} & & \textcolor{black}{0.000} & & \textcolor{black}{0.0002} & & \textcolor{black}{N/A} & & \textcolor{black}{0} && \textcolor{black}{4280}\\
 \textcolor{black}{SRRR} & \textcolor{black}{32.76} & & \textcolor{black}{0.927} & & \textcolor{black}{0.000} & & \textcolor{black}{0.011} & & \textcolor{black}{N/A} & & \textcolor{black}{0} && \textcolor{black}{1922}\\
 \textcolor{black}{MSRL} & \textcolor{black}{6.91} & & \textcolor{black}{0.9334} & & \textcolor{black}{0} & & \textcolor{black}{0.0117} & & \textcolor{black}{N/A} & & \textcolor{black}{0} && \textcolor{black}{262}\\
 \hline \\
\multicolumn{14}{c}{\textcolor{black}{Large dimension of response: $q = 60$}} \\
\multicolumn{14}{c}{\textcolor{black}{Non-decaying correlation for the column of design matrix $\bX$: $\Gamma_{ij} = 0.5\{1 + \mathbb{I}(i = j)\}$}} \\
\multicolumn{14}{c}{\textcolor{black}{AR(1) error covariance matrix: $\bV_{(M_{\text{o}}), ij}^0 = 2 \times 0.5^{\abs{i-j}}$}} \\
 \hline
 \textcolor{black}{EAS-BIC} & \textcolor{black}{3.01} & & \textcolor{black}{0.000} & & \textcolor{black}{0.000} & & \textcolor{black}{0.000} & & \textcolor{black}{0.99} & & \textcolor{black}{1} && \textcolor{black}{3174}\\
 \textcolor{black}{EAS-CV} & \textcolor{black}{3.01} & & \textcolor{black}{0.000} & & \textcolor{black}{0.000} & & \textcolor{black}{0.000} & & \textcolor{black}{0.99} & & \textcolor{black}{1} && \textcolor{black}{17680}\\
  \textcolor{black}{MBSP} & \textcolor{black}{3.97} & & \textcolor{black}{0.053} & & \textcolor{black}{0.000} & & \textcolor{black}{0.000} & & \textcolor{black}{N/A} & & \textcolor{black}{0.08} && \textcolor{black}{4521}\\
 \textcolor{black}{MLASSO} & \textcolor{black}{7.98} & & \textcolor{black}{0.496} & & \textcolor{black}{0.000} & & \textcolor{black}{0.001} & & \textcolor{black}{N/A} & & \textcolor{black}{0.00} && \textcolor{black}{36}\\
 \textcolor{black}{MSGLASSO} & \textcolor{black}{6.56} & & \textcolor{black}{0.78} & & \textcolor{black}{0.000} & & \textcolor{black}{0.003} & & \textcolor{black}{N/A} & & \textcolor{black}{0.00} && \textcolor{black}{7298}\\
  \textcolor{black}{SPLS} & \textcolor{black}{5.14} & & \textcolor{black}{0.185} & & \textcolor{black}{0.000} & & \textcolor{black}{0.0002} & & \textcolor{black}{N/A} & & \textcolor{black}{0} && \textcolor{black}{4799}\\
 \textcolor{black}{SRRR} & \textcolor{black}{15.19} & & \textcolor{black}{0.931} & & \textcolor{black}{0.000} & & \textcolor{black}{0.011} & & \textcolor{black}{N/A} & & \textcolor{black}{0} && \textcolor{black}{2848}\\
  \textcolor{black}{MSRL} & \textcolor{black}{6.56} & & \textcolor{black}{0.9406} & & \textcolor{black}{0} & & \textcolor{black}{0.0132} & & \textcolor{black}{N/A} & & \textcolor{black}{0} && \textcolor{black}{7098}\\
 \hline \\
 \multicolumn{14}{c}{\textcolor{black}{Large dimension of response: $q = 60$}} \\
  \multicolumn{14}{c}{\textcolor{black}{Non-decaying correlation for the column of design matrix $\bX$: $\Gamma_{ij} = 0.5\{1 + \mathbb{I}(i = j)\}$}} \\
 \multicolumn{14}{c}{\textcolor{black}{Dense error covariance matrix:  $\bV_{(M_{\text{o}}), ij}^0 = 1 + \mathbb{I}(i = j)$}} \\
 \hline
 \textcolor{black}{EAS-BIC} & \textcolor{black}{3.00} & & \textcolor{black}{0.000} & & \textcolor{black}{0.000} & & \textcolor{black}{0.000} & & \textcolor{black}{0.99} & & \textcolor{black}{1} && \textcolor{black}{3214} \\
 \textcolor{black}{EAS-CV} & \textcolor{black}{3.11} & & \textcolor{black}{0.119} & & \textcolor{black}{0.000} & & \textcolor{black}{0.0002} & & \textcolor{black}{0.72} & & \textcolor{black}{0.72} && \textcolor{black}{17808}\\
  \textcolor{black}{MBSP} & \textcolor{black}{2.98} & & \textcolor{black}{0.052} & & \textcolor{black}{0.000} & & \textcolor{black}{0.000} & & \textcolor{black}{N/A} & & \textcolor{black}{0.12} && \textcolor{black}{4720}\\
 \textcolor{black}{MLASSO} & \textcolor{black}{7.27} & & \textcolor{black}{0.531} & & \textcolor{black}{0.000} & & \textcolor{black}{0.001} & & \textcolor{black}{N/A} & & \textcolor{black}{0.00} && \textcolor{black}{41}\\
 \textcolor{black}{MSGLASSO} & \textcolor{black}{5.47} & & \textcolor{black}{0.734} & & \textcolor{black}{0.000} & & \textcolor{black}{0.002} & & \textcolor{black}{N/A} & & \textcolor{black}{0.00} && \textcolor{black}{7184}\\
  \textcolor{black}{SPLS} & \textcolor{black}{4.57} & & \textcolor{black}{0.186} & & \textcolor{black}{0.0001} & & \textcolor{black}{0.0002} & & \textcolor{black}{N/A} & & \textcolor{black}{0} && \textcolor{black}{4800}\\
 \textcolor{black}{SRRR} & \textcolor{black}{10.34} & & \textcolor{black}{0.928} & & \textcolor{black}{0.000} & & \textcolor{black}{0.011} & & \textcolor{black}{N/A} & & \textcolor{black}{0} && \textcolor{black}{2876}\\
  \textcolor{black}{MSRL} & \textcolor{black}{4.51} & & \textcolor{black}{0.9374} & & \textcolor{black}{0} & & \textcolor{black}{0.0125} & & \textcolor{black}{N/A} & & \textcolor{black}{0} && \textcolor{black}{47963}\\
\hline
\hline
\end{tabular}\caption{\footnotesize Performance of EAS compared other methods for large dimension of the multivariate response, dense (non-decaying) correlation of the columns of the design matrix and error covariance matrix. Throughout this experiment, we consider UHD ($p \gg n$) and sparse case, i.e., the sample size, the number of predictors, and the size of the true model is fixed at $n = 150, p = 1000$, and $\abs{M_{\text{o}}} = 50$.}
\label{tab: Simresults_largeq}
}
\end{table}

\textcolor{black}{Table~\ref{tab: Simresults_largeq} compares performance of the methods with respect to challenges relating to larger dimension of the response, non-decaying correlation structure among the columns of the design matrix, and dense error covariance matrix.  We do not present the results of} \verb|MBGL-SS| \textcolor{black}{method because it takes more than 96 hours for a single dataset when $q$ is increased to $60$. Our results demonstrate that using either the BIC or CV tuning selection procedures, our EAS method outperforms all other methods in terms of both prediction performance and the average proportion of correct model selections. Furthermore, our EAS approach assigns a high GF probability to the true model, which aligns with our theoretical result regarding strong model selection consistency. In comparison to the frequentist approaches,} \verb|SPLS| \textcolor{black}{is the only method that performs relatively well, whereas all other procedures tend to have a high FDR.}

\textcolor{black}{Lastly, observe that when the error covariance matrix is dense and the design matrix is complex---in addition to a large number of predictors---our EAS approach tends to result in higher false discoveries when optimal $\epsilon$ is selected by CV. As mentioned before, this issue may be due to a lack of identifiability of the true model in certain folds, resulting in the selection of a suboptimal $\epsilon$ that is relatively smaller than the optimal one. Nevertheless, selecting the regularization parameter using BIC appears to mitigate these computational bottlenecks observed with the CV procedure. Overall, our findings demonstrate that our EAS method offers superior performance in these challenging scenarios, and may prove beneficial in various applied settings.}

\section{Yeast cell data analysis} \label{sec: EASrealdata}

\textcolor{black}{This section presents the results of implementing the EAS algorithm on the yeast cell cycle dataset \citep{lee2002transcriptional}. The first part of the data contains yeast cell cycle gene expression data consisting of 18 measurements of messenger ribonucleic acid (mRNA) levels which are taken every 7 minutes of 119 minutes covering two cell cycle periods, for 542 cell cycle-related genes. The second part of the data contains binding information for a total of 106 transcription factors (TFs). TFs are essential regulators of gene expression, and the binding of these proteins to specific DNA sequences controls the transcription of genes into mRNA, and ultimately the synthesis of functional proteins. The levels of mRNA, and therefore protein, produced by a gene are determined by the activity of the associated TFs \citep{phillips2008regulation,wang2015understanding}. Therefore, it is important to identify the key TFs that regulate cell cycles.  The particular dataset we analyze here is taken from the R package} spls \citep{spls}. \textcolor{black}{The response matrix $\bY$ of gene-expression data is $18 \times 542$ dimensional and the design matrix $\bX$ of TFs is $106 \times 542 $ dimensional, i.e., $n=542, p=106$, and $q=18$. The data has been previously analyzed for variable selection of TFs in \citep{bai2018high, chun2010sparse}}.

\textcolor{black}{As not all TFs are significantly contributing to the gene expressions; we aim to find a parsimonious set of cell-cycle regulating TFs by applying our EAS algorithm to this data. We use only the CV method to select the optimal tuning parameter $\epsilon$ for the real data application because we cannot rely on BIC due to the probable violation of the Gaussian assumption.  We use a uniform grid of 16 values ranging from $0.01$ to $0.2$ for selection of $\epsilon$ via CV. We do not extend the grid endpoint over $0.2$ because an implementation of Algorithm~\ref{algo: MCMC} did not select any TFs as admissible for epsilon greater than $0.2$. We run the EAS method on the training set with the MCMC algorithm for 500 steps, discarding the initial 200 steps, and evaluate its performance on the validation set. The multivariate LASSO estimates serve as weights for proposing or removing predictors in the MCMC algorithm.  We use the MAP estimated model from the MCMC sample as a point estimator to calculate the MSPE on the validation set. We select the optimal $\epsilon$ that minimizes the average MSPE across the 10 folds. The CV method selects optimal $\epsilon$ as $0.09$. Finally,  using $\epsilon=0.09$, we implement Algorithm~\ref{algo: MCMC} on the entire yeast cell cycle data ten times to account for random variation in the MCMC chains, and each chain is run for 10,000 MCMC steps, discarding the initial 5,000 steps. Thus, we arrive at ten MCMC chains, each containing 5,000 samples from the GF distribution of the $\epsilon$-admissible models for the yeast data.}

\begin{table}[H]
\centering
\scriptsize
\textcolor{black}{
\begin{tabular}{l|r|l|r|r|r|r}
Method & Model & Selected TFs & $\#$ TFs & \begin{tabular}[c]{@{}c@{}}GF probability\\ $r_{\epsilon}(M \mid \bY)$\end{tabular} & MSPE & MAPE \\ \hline \hline
 & 1     & \begin{tabular}[c]{@{}l@{}}ACE2 GAT3 HIR1 MBP1 \\ MCM1 NDD1 STE12 SWI5 SWI6\end{tabular}                                         & 9    & $48.9\%$ & 19.0 & 23.9  \\
\cmidrule(l){2-7}
\multirow{5}{*}{EAS-CV} & 2     & \begin{tabular}[c]{@{}l@{}}ACE2 GAT3 HIR1 MBP1 \\ MCM1 NDD1 STE12 SWI5 \\ SWI6 YAP5\end{tabular}                                 & 10   & $18.9\%$ & 19.0 & 23.9 \\ \cmidrule(l){2-7}
& 3    & \begin{tabular}[c]{@{}l@{}}ACE2 FKH2 GAT3 HIR1 \\ MBP1 MCM1 NDD1 STE12\\ SWI4 SWI5\end{tabular}                                  & 10   & $10\%$ & 19.0 & 23.6 \\ \midrule
MBSP &    & \begin{tabular}[c]{@{}l@{}}ACE2 FKH2 GAT3 HIR1 \\ HIR2 MBP1 MET4 NDD1 \\ REM1 STE12 SWI5 SWI6\end{tabular}                                         & 12    & NA & 18.6 & 23.5 \\  \midrule
MBGL-SS &    & \begin{tabular}[c]{@{}l@{}}GAT3  NDD1 SWI5 SWI6\end{tabular}                                         & 4    & NA & 20.1 & 23.6  \\ \hline
\end{tabular}}\cprotect\caption{\footnotesize \textcolor{black}{Selected TFs for models with more than $.1$ estimated GF probability based on the EAS,} \verb|MBSP| \textcolor{black}{and} \verb|MBGL-SS| \textcolor{black}{methods, respectively. The optimal tuning parameter is selected via CV.}}\label{tab: yeast_model_prob}
\end{table}

\textcolor{black}{The first five rows of Table~\ref{tab: yeast_model_prob} present the GF probabilities of the model, based on the ten MCMC chains. We truncate the table to display models with estimated GF probabilities greater than $.1$. The MAP estimated model identifies a rather parsimonious model, containing only 9 significant TFs. Interestingly, except for HIR1 and GAT3, all of the selected TFs in the MAP model are among the 21 experimentally confirmed cell-cycle-related TFs \citep{wang2007group}, supporting the relevance of our EAS approach to select a parsimonious model in real applications. Moreover, our research has identified HIR1 and GAT3 as two novel TFs. To evaluate the predictive performance of each of the models in Table~\ref{tab: yeast_model_prob}, we employ ten-fold CV. Specifically, we use $90\%$ of the data as the training set and obtain the least squares estimator. We then calculate the MSE and median absolute deviation (MAD) of the residuals on the remaining $10\%$ of the data that we held out. We repeat this process 1,000 times, each time using different training and test sets, and compute the average MSE as MSPE and average MAD as mean absolute prediction error (MAPE). Finally, we scale the MSPE and MAPE by a factor of 100 for better clarity. The number of TFs in the optimally selected model for most of the competing methods are presented in Table 3 of \cite{bai2018high}. We present the selected TFs in the optimal model for the Bayesian methods only}, i.e., \verb|MBSP| and \verb|MBGL-SS|, \textcolor{black}{in the last two rows of Table~\ref{tab: yeast_model_prob}. The last two columns of the table suggest that our EAS method does an outstanding job in terms of prediction accuracy in comparison to the competing methods while simultaneously maintaining parsimony in the selected model.}

%\textcolor{black}{The estimated GF probabilities of the models selected by our EAS method presents a classic example of non-existence of a single best model in real data applications. The ability of our EAS method to provide a probabilistic assessment of competing models stands out to be the most attractive feature of our EAS method. In the case of the yeast cell cycle data, our EAS method assigns similar probabilities to Model 1 and Model 2, without favoring a single model. Furthermore, both models exhibit comparable prediction accuracy. Therefore, we recommend considering both Model 1 and Model 2 to draw meaningful inferences from the data, rather than solely relying on one model.} 

\textcolor{black}{While models 2 and 3, each having 10 significant TFs, and the model selected by} \verb|MBSP| \textcolor{black}{ (with 12 significant TFs) exhibit similar or slightly better MSPE and MAPE compared to model 1, our EAS approach demonstrates remarkable performance by assigning a higher GF probability to the more parsimonious model 1. As determined by the EAS approach, the marginal improvement in prediction accuracy obtained with additional TFs is not substantial enough to justify the increase in model complexity.}

%\textcolor{black}{The GF inference construction that we employ allows for parametric inference analogous to Bayesian
%inference. With respect to uncertainty quantification of the unknown parameters, conditional
%on a given mode (i.e., set of active coefficients), the standard deviation of the
%estimated GF (i.e., posterior) distribution for each coefficient would serve a similar role as a
%standard error, if the posterior mean is taken as the point estimate. Instead of using p-values
%to determine statistical significance, we use the marginal inclusion probability of a particular
%coefficient over all models/graphs in the GF distribution of 𝐺 to determine the significance of
%the coefficient.}

\begin{table}[H]
\centering
\textcolor{black}{
\begin{tabular}{lrr}
\begin{tabular}[c]{@{}l@{}}Transcription\\  factors\end{tabular} & & \begin{tabular}[c]{@{}l@{}}Marginal inclusion \\ probability (in $\%$) \end{tabular}  \\ \hline \hline
ACE2  & & 100 \\
HIR1 &  &100 \\
NDD1  &  & 100     \\
STE12 & & 100    \\
 SWI15 & &  100   \\
  GAT3 & &   99.9   \\ \midrule
  MBP1 & &   98.9  \\
  SWI16 &  &  95.0  \\
   MCM1& &   82.2  \\  \midrule
   FKH2 &   &  58.5 \\
  RME1 & &   53.5  \\
  HIR2 &  &   52.6 \\
   ARG81 & & 50.0    \\
   \hline \hline
 \end{tabular}\caption{\footnotesize Marginal inclusion probabilities of the TFs based on ten independent MCMC chains on the yeast cell cycle data. The optimal tuning parameter is selected via CV.}
\label{tab: yeast_TF_prob}
}
\end{table}

\textcolor{black}{The GF construction used enables parametric inference similar to Bayesian inference. In terms of quantifying uncertainty of the unknown parameters, the standard deviation of the estimated GF distribution of regression coefficient for each TFs, conditioned on a given model, plays a similar role to that of a standard error if the posterior mean is used as the point estimate. In this setting, we do not rely on p-values to determine statistical significance, but instead utilize the marginal inclusion probability of a specific TFs across all models in the GF distribution to establish its significance. The marginal inclusion probabilities of all TFs are presented in Table~\ref{tab: yeast_TF_prob}. We have shortened the table to only include TFs that have inclusion probabilities of at least $.5$. Remarkably, all the TFs with marginal inclusion probabilities exceeding $.8$ are the components of the MAP model listed in Table~\ref{tab: yeast_model_prob}.}

\textcolor{black}{In this exposition, the real data analysis of yeast cell cycle data is not intended as a comprehensive investigation. Rather, it serves as a proof of concept for the practical utility of the EAS methodology for analyzing real data. It should be noted that obtaining a probability distribution of all possible models that are $\epsilon$-admissible, as presented in Table~\ref{tab: yeast_model_prob}, along with the marginal inclusion probabilities of the TFs, as in Table~\ref{tab: yeast_TF_prob}, is not feasible using frequentist or Bayesian point estimation-based procedures, such as} \verb|MBSP|.  \textcolor{black}{The ability of our EAS method to provide a probabilistic assessment of competing models stands out as an attractive feature of our EAS method. Although MCMC-based approaches are computationally more expensive, they provide more comprehensive information for uncertainty quantification.
}

\section{Concluding remarks}
 
The theoretical results presented in this article assume that the dimension $q$ of response is fixed. However, we kept careful account of all instances of $q$ in all of the non-asymptotic results presented, leaving an indication for the reader to understand the influence of $q$ in the consistency rates.  An obvious extension of our work is to allow $q$ to grow, and in that case a careful account of the role of $q$ in the theory is critical to determine the circumstances in which the EAS method remains a consistent model selection procedure.  A field where this extended theory could be applied is functional data analysis (FDA), where the response is measured very densely for each subject and naturally the dimension of the response grows. Although smoothness in the mean and the covariance function is fundamental to the analysis of FDA, many FDA procedures simply rely on techniques that are developed for multivariate response data. Thus, we view this article as a promising first step on the pathway to a novel functional variable selection procedure.

\appendix

\section{Generalized fiducial distribution for multivariate linear regression} \label{sec: suppGFI}

Here, we will derive the GF distribution of $(\bB_{\cdot\,M}   \M{A}_{(M)})$ presented in (\ref{eqn: EASMassfunctionmodelM}). The data generating equation corresponding to $i$-th data, $\bY_i$ is,
$$
G_i := G(\bB_{\cdot\,M},\M{A}_{(M)}, \M{U}_i) = \bB_{\cdot\,M}\bX_{M\,\cdot,i} + \M{A}_{(M)}\M{U}_i \quad i=1,2,\dots,n,
$$
where, $\bX_{M\,\cdot,i}$ is the $i$th column of the matrix $\bX_{M\,\cdot}$. Then, 
\begin{align*}
    \frac{\partial G_i}{\partial \bB_{\cdot\,M}} &= 
    \begin{bmatrix}
    \frac{\partial \fpr{\bB_{\cdot\,M}\bX_{M\,\cdot,i}}}{\partial b_{11} }, & \dots, & \frac{\partial \fpr{\bB_{\cdot\,M}\bX_{M\,\cdot,i}}}{\partial b_{q\abs{M}} }
    \end{bmatrix} \\
    &= \begin{bmatrix}
    \M{J}^{11}\bX_{M\,\cdot,i} & \cdots & \M{J}^{1\abs{M}}\bX_{M\,\cdot,i} & \cdots & \M{J}^{q1}\bX_{M\,\cdot,i} & \cdots &\M{J}^{q\abs{M}}\bX_{M\,\cdot,i} 
    \end{bmatrix} \\
&= \M{I}_q \otimes \bX_{M\,\cdot,i}^\top  \in \Real^{q \times q\abs{M}},
\end{align*}
where $\M{J}^{kl} \in \Real^{q \times \abs{M}}$ is a sparse matrix of which $(k,l)$th element is $1$ and all others are zero. Similarly,
 $   \frac{\partial G_i}{\partial \M{A}_{(M)}} = \M{I}_q \otimes \M{U}_i^\top \in \Real^{q \times q^2}$. So the matrix of derivatives corresponding to the $i$-th data generating equation is, 
\begin{align*}
   \M{D}_{(M),i} =  \begin{bmatrix}
      \frac{\partial G_i}{\partial \bB_{\cdot\,M}} & \vdots &  \frac{\partial G_i}{\partial \M{A}_{(M)}} 
   \end{bmatrix} 
   = \begin{bmatrix}
   \M{I}_q \otimes \bX_{M\,\cdot,i}^\top & \vdots & \M{I}_q \otimes \M{U}_i^\top
    \end{bmatrix} \in \Real^{q \times \overline{q\abs{M}+q^2}}.
\end{align*}
The entire matrix of derivative corresponding to all observations is,
\begin{align*}
    \M{D}_{(M)} 
    &= \begin{bmatrix}
       \M{I}_q \otimes \M{X}_{M\,\cdot,1}^\top & \M{I}_q \otimes \M{U}_1^\top \\
       \vdots & \vdots \\
       \M{I}_q \otimes \M{X}_{M\,\cdot,n}^\top & \M{I}_q \otimes \M{U}_n^\top
    \end{bmatrix}
    \in \Real^{qn \times \overline{q\abs{M}+q^2}},
\end{align*}
which is almost surely of full column rank if $n > \abs{M} + q$. Now, define,
$ \M{P}_{(M)} 
    := \begin{bmatrix}
    \M{I}_q \otimes \bX_{M\,\cdot}^\top & \M{I}_q \otimes \M{U}^\top
    \end{bmatrix}$,
and evaluate the derivatives at  $\M{U} = \M{A}_{(M)}^{-1}(\M{Y} - \bB_{\cdot\,M}\bX_{M\,\cdot}) =  \M{A}_{(M)}^{-1}\widetilde{\M{U}} $ with $\widetilde{\M{U}} = 
\M{Y} - \bB_{\cdot\,M}\bX_{M\,\cdot} \in \Real^{q \times n} $.
Because $\M{P}_{(M)}$ is obtained rearranging rows of $\M{D}_{(M)}$, $\M{D}_{(M)}^\top \M{D}_{(M)} = \M{P}_{(M)}^\top \M{P}_{(M)}$ and 
\begin{align*}
    &\M{P}_{(M)}^\top \M{P}_{(M)} \\
    &= \begin{bmatrix}
    \M{I}_q \otimes \bX_{M\,\cdot}\bX_{M\,\cdot}^\top & & \M{I}_q \otimes \bX_{M\,\cdot}\M{U}^\top \\ \\
    \M{I}_q \otimes \M{U}\bX_{M\,\cdot}^\top & & \M{I}_q \otimes \M{U}\M{U}^\top
    \end{bmatrix} \\
    &= \begin{bmatrix}
    \M{I}_q \otimes \bX_{M\,\cdot}\bX_{M\,\cdot}^\top & &  \M{I}_q \otimes \fpr{\bX_{M\,\cdot}\widetilde{\M{U}}^\top{\M{A}_{(M)}^{-1}}^\top} \\
    \\
    \M{I}_q \otimes \fpr{\M{A}_{(M)}^{-1}\tilde{ \M{U}}\bX_{M\,\cdot}^\top} & &  \M{I}_q \otimes  \M{A}_{(M)}^{-1}\widetilde{\M{U}} \widetilde{\M{U}}^\top {\M{A}_{(M)}^{-1}}^\top 
    \end{bmatrix} \\ \\
    &=\begin{bmatrix}
    \M{I}_{q} \otimes \M{I}_{\abs{M}} & \M{0} \\
    \\
    \M{0} & \M{I}_q \otimes \M{A}_{(M)}^{-1}
    \end{bmatrix}
    \begin{bmatrix}
    \M{I}_q \otimes \bX_{M\,\cdot}\bX_{M\,\cdot}^\top & & \M{I}_q \otimes \bX_{M\,\cdot}\widetilde{\M{U}}^\top \\
    \\
    \M{I}_q \otimes \widetilde{\M{U}}\bX_{M\,\cdot}^\top & & \M{I}_q \otimes \widetilde{\M{U}}\widetilde{\M{U}}^\top  
    \end{bmatrix}
    \begin{bmatrix}
    \M{I}_{q} \otimes \M{I}_{\abs{M}} & \M{0} \\
    \\
    \M{0} & \M{I}_q \otimes {\M{A}_{(M)}^{-1}}^\top
    \end{bmatrix}.
\end{align*}
Then,
$$
\det \M{D}_{(M)}^\top\M{D}_{(M)} = \fpr{\det \M{A}_{(M)}\M{A}_{(M)}^\top}^{-q}\det \begin{bmatrix}
    \M{I}_q \otimes \bX_{M\,\cdot}\bX_{M\,\cdot}^\top & & \M{I}_q \otimes \bX_{M\,\cdot}\widetilde{\M{U}}^\top \\
    \\
    \M{I}_q \otimes \tilde{\M{U}}\bX_{M\,\cdot}^\top & & \M{I}_q \otimes \widetilde{\M{U}}\widetilde{\M{U}}^\top  
    \end{bmatrix}.
$$
After row and column operations,
\begin{align*}
    \det \begin{bmatrix}
    \M{I}_q \otimes \bX_{M\,\cdot}\bX_{M\,\cdot}^\top & & \M{I}_q \otimes \bX_{M\,\cdot}\widetilde{\M{U}}^\top \\
    \\
    \M{I}_q \otimes \widetilde{\M{U}}\bX_{M\,\cdot}^\top & & \M{I}_q \otimes \M{\widetilde{U}}\widetilde{\M{U}}^\top  
    \end{bmatrix} & =
    \det \left\{\M{I}_q \otimes 
   \begin{bmatrix}
   \bX_{M\,\cdot}\bX_{M\,\cdot}^\top & &  \bX_{M\,\cdot}\widetilde{\M{U}}^\top \\
    \\
    \widetilde{\M{U}}\bX_{M\,\cdot}^\top & &  \widetilde{\M{U}}\widetilde{\M{U}}^\top  
    \end{bmatrix} \right\} \\
    & =
    \fpr{\det \begin{bmatrix}
   \bX_{M\,\cdot}\bX_{M\,\cdot}^\top & &  \bX_{M\,\cdot}\widetilde{\M{U}}^\top \\
    \\
    \widetilde{\M{U}}\bX_{M\,\cdot}^\top & &  \widetilde{\M{U}}\widetilde{\M{U}}^\top  
    \end{bmatrix} }^q.
\end{align*}
Further, by property of determinant of block matrices,
\begin{align*}
    \det \begin{bmatrix}
   \bX_{M\,\cdot}\bX_{M\,\cdot}^\top & &  \bX_{M\,\cdot}\widetilde{\M{U}}^\top \\
    \\
    \widetilde{\M{U}}\bX_{M\,\cdot}^\top & &  \widetilde{\M{U}}\widetilde{\M{U}}^\top  
    \end{bmatrix} &= \fpr{\det \bX_{M\,\cdot}\bX_{M\,\cdot}^\top} \det \tpr{\widetilde{\M{U}}\widetilde{\M{U}}^\top - \widetilde{\M{U}}\bX_{M\,\cdot}^\top\fpr{\bX_{M\,\cdot}\bX_{M\,\cdot}^\top}^{-1}\bX_{M\,\cdot}\widetilde{\M{U}}^\top}. 
\end{align*}
Analogous to univariate linear regression,
\begin{align*}
    \M{Y} - \widehat{\bB}_{\cdot\,M}\bX_{M\,\cdot} = \M{Y} - \M{Y}\bX_{M\,\cdot}^\top \fpr{\bX_{M\,\cdot}\bX_{M\,\cdot}^\top}^{-1}\bX_{M\,\cdot} = \M{Y} - \M{Y}\bH_{(M)} = \M{Y}\fpr{\M{I}-\bH_{(M)}},
\end{align*}
and,
\begin{align*}
    &\widetilde{\M{U}}\widetilde{\M{U}}^\top - \widetilde{\M{U}}\bX_{M\,\cdot}^\top\fpr{\bX_{M\,\cdot}\bX_{M\,\cdot}^\top}^{-1}\bX_{M\,\cdot}\widetilde{\M{U}}^\top \\
    &= \fpr{\widetilde{\M{U}} -\widehat{\bB}_{\cdot\,M}\bX_{M\,\cdot} + \widehat{\bB}_{\cdot\,M}\bX_{M\,\cdot} }\tpr{\M{I}-\bH_{(M)}}\fpr{\widetilde{\M{U}}-\widehat{\bB}_{\cdot\,M}\bX_{M\,\cdot} + \widehat{\bB}_{\cdot\,M}\bX_{M\,\cdot} }^\top \\
    &= \fpr{\M{Y}-\widehat{\bB}_{\cdot\,M}\bX_{M\,\cdot}}\tpr{\M{I}-\bH_{(M)}}\fpr{\M{Y}-\widehat{\bB}_{\cdot\,M}\bX_{M\,\cdot}}^\top  \\
    &= \widehat{\M{\Sigma}}_{(M)}.
\end{align*}
% where
% \begin{align*}
%     \M{D}(\V{\beta}) = \begin{bmatrix}
%     \M{I}_n \otimes \M{X}\M{X}^\top & & \M{\tilde{U}} \otimes \M{X} \\
%     \\
%     \M{\tilde{U}}^\top \otimes \M{X}^\top & & \M{\tilde{U}}^\top \M{\tilde{U}} \otimes \M{I}_n
%     \end{bmatrix}
% \end{align*}
The Jacobian is obtained as
\begin{align*}
    J(\bY, (\bB_{\cdot\,M},\M{A}_{(M)})) := \sqrt{\det \M{D}_{(M)}^\top\M{D}_{(M)}} = C_M \fpr{\det\M{V}_{(M)}}^{-q/2}
\end{align*}
where $C_M := \fpr{\det\M{\Omega}_{(M)}}^{q/2} \fpr{\det \widehat{\M{\Sigma}}_{(M)}}^{q/2}$  and $\M{\Omega}_{(M)} = \bX_{M\,\cdot}\bX_{M\,\cdot}^\top$. Note that $C_M$ does not depend on either of $\bB_{\cdot\,M}$ or $\M{A}_{(M)}$ . The likelihood function is,
\begin{align*}
    f(\bY, (\bB_{\cdot\,M},\M{A}_{(M)})) \propto  \fpr{\det\M{V}_{(M)}}^{-\frac{n}{2}} \exp\tpr{-\frac{1}{2}\tr\fpr{\M{R}_{(M)}\M{V}_{(M)}^{-1}}},
\end{align*}
where,
\begin{align*}
    \M{R}_{(M)} &= \fpr{\M{Y}- \bB_{\cdot\,M}\bX_{M\,\cdot}} \fpr{\M{Y}- \bB_{\cdot\,M}\bX_{M\,\cdot}}^\top \\
    &= \widehat{\M{\Sigma}}_{(M)} + \fpr{\bB_{\cdot\,M}- \widehat{\bB}_{\cdot\,M}} \bX_{M\,\cdot}\bX_{M\,\cdot}^\top \fpr{\bB_{\cdot\,M}- \widehat{\bB}_{\cdot\,M}}^\top.
\end{align*}
Under the matrix change of variables $\M{V}_{(M)} =\M{A}_{(M)}\M{A}_{(M)}^\top$, we use Theorem 1.4.10 of \cite{gupta2018matrix_sub}, to derive the marginal GF distribution of model $M$ as,  
\begin{align*}
    &r_\epsilon(M \vert \M{Y}) \propto  \int_{\bB_{\cdot\,M} }  \int_{\M{A}_{(M)} } f(\bY, (\bB_{\cdot\,M},\M{A}_{(M)})) J(\bY, (\bB_{\cdot\,M},\M{A}_{(M)})) h_\epsilon(\bB_{\cdot\,M}) d(\M{A}_{(M)}, \bB_{\cdot\,M})\\
    &\propto  C_M \int_{\bB_{\cdot\,M}}  \int_{\M{A}_{(M)} } \fpr{\det\M{V}_{(M)}}^{-\frac{q+n}{2}}\exp\tpr{-\frac{1}{2}\tr\fpr{\M{R}_{(M)}\M{V}_{(M)}^{-1}}}h_\epsilon\fpr{\bB_{\cdot\,M}} d\M{A}_{(M)} d\bB_{\cdot\,M} \\
    &\propto C_M \int_{\bB_{\cdot\,M}} \int_{\M{V}_{(M)} > \M{0}} \fpr{\det\M{V}_{(M)}}^{-\frac{q+n+1}{2}}\exp\tpr{-\frac{1}{2}\tr\fpr{\M{R}_{(M)}\M{V}_{(M)}^{-1}}} h_\epsilon\fpr{\bB_{\cdot\,M}} d\M{V}_{(M)}  d\bB_{\cdot\,M}   \\
    &\propto C_M  \int_{\bB_{\cdot\,M}} \tpr{\det \M{R}_{(M)}}^{-\frac{n}{2}}h_\epsilon\fpr{\bB_{\cdot\,M}} d\bB_{\cdot\,M} \quad (\textrm{w.r.t inverse-wishart kernel})  \\
    &= C_M \int_{} \fpr{\det \tpr{\widehat{\M{\Sigma}}_{(M)} + \fpr{\bB_{\cdot\,M}- \widehat{\bB}_{\cdot\,M}} \M{\Omega}_{(M)} \fpr{\bB_{\cdot\,M}- \widehat{\bB}_{\cdot\,M}}^\top}}^{-\frac{n}{2}}h_\epsilon\fpr{\bB_{\cdot\,M}} d\bB_{\cdot\,M} \\
    &= C_M \int_{} \fpr{\det \tpr{\M{I}_M + \M{\Omega}_{(M)}\fpr{\bB_{\cdot\,M}- \widehat{\bB}_{\cdot\,M}}^\top \widehat{\M{\Sigma}}_{(M)}^{-1} \fpr{\bB_{\cdot\,M}- \widehat{\bB}_{\cdot\,M}}}\det \,\widehat{\M{\Sigma}}_{(M)}}^{-\frac{n}{2}}
    h_\epsilon\fpr{\bB_{\cdot\,M}} d\bB_{\cdot\,M} \\
    &\propto\Gamma_q\fpr{\frac{n-\abs{M}}{2}}\pi^{\frac{q\abs{M}}{2}} \fpr{\det \widehat{\M{\Sigma}}_{(M)}}^{-\fpr{\frac{n-\abs{M}-q}{2}}}\; \mathbb{E}\fpr{h_\epsilon\fpr{\bB_{\cdot\,M}}}  \quad (\textrm{w.r.t matrix-t kernel}) .
\end{align*}

This completes the derivation of the GF distribution of $(\bB_{\cdot\,M}   \M{A}_{(M)})$ for multivariate linear regression setting under the model~(\ref{eqn: EASdgm}).

\section{Proof of Lemmas and Theorems} \label{sec: suppPROOF} %% if no title is needed, leave empty \section*{}.

%**************************************************************************************************************************************************************************%
%.                                                                                   PROOF OF MINIMUM EIGEN VALUE BOUND OF RSS                                                                 % 
%**************************************************************************************************************************************************************************% 

\begin{proof}[Proof of Lemma \ref{Lemma: MinEigenofRSSMatrix}]
\begin{align*}
 \ubar{\lambda}_{v}^{-1}\lambda_{\min}\big(\widehat{\bSigma}_{(M)} \big) &=  \lVert  \bAMo^{{0}^{-1}} \rVert^2\lambda_{\min}(\bY(\bI_n - \bH_{(M)})\bY^\top) \\
& \geq  \lambda_{\min}(\bAMo^{{0}^{-1}} \bY(\bI_n - \bH_{(M)})\bY^\top \bAMo^{{0}^{-\top}}) \\
&=  \lambda_{\min}(\bZ(\bI_n - \bH_{(M)})\bZ^\top) 
\end{align*}
where $\bZ :=  \bAMo^{{0}^{-1}}\bY = \mathbb{E}(\bZ) +   \bU$, with $\mathbb{E}(\bZ) = \bAMo^{{0}^{-1}}\bBMo^{0}\bXMo$.  Recall that $\bH_{(M)}$ is the symmetric projection matrix onto row space of $\bX_{M\,\cdot}$, and let $r := \textrm{rank}(\bH_{(M)})$. Since $\bI_n - \bH_{(M)}$ is symmetric and idempotent with rank $n-r$, there exists a $\bG_{(M)} \in \Real^{n \times (n-r)}$ such that $\bG_{(M)}\bG_{(M)}^\top = (\bI - \bH_{(M)})$ and $\bG_{(M)}^\top\bG_{(M)} = \bI_{n-r}$. Define $\widetilde{\bZ}_M := \bZ\bG_{(M)} = \mathbb{E}(\bZ\bG_{(M)}) + \bU\bG_{(M)}$. Then, $\bZ(\bI_n - \bH_{(M)})\bZ^\top = \btZ_{(M)}\btZ_{(M)}^\top$ and the minimum eigenvalue of $\btZ_{(M)}\btZ_{(M)}^\top$ will be the minimum singular value of $\btZ_{(M)}$ (in the compact SVD notation).   Notice that, $(\bU\bG_{(M)})^\top \sim \textrm{Matrix-Normal}_{n-r,q}\fpr{\M{0}, \M{I}_{n - r},\bI_q}$. Applying Theorm 2.1 of \cite{Tu2020} for independent non-centered Gaussian design for sufficiently large n with $t = \sqrt{n - r-q} - \sqrt{q} - 1 - \sqrt{\tau}\sqrt{n-r} > 0 $, $n - r > 2q$, and $0 < \tau < 1$, we get
\begin{align*}
 \mathbb{P}_y\fpr{\sigma_{\min}\big(\btZ_{(M)} \big) \geq \sqrt{\tau(n-r)}} > 1 -  e^{-\left\{\sqrt{n - r-q} - \sqrt{q} - 1 - \sqrt{\tau}\sqrt{n-r}\right\}^2}.
\end{align*}
Since, $ (\sqrt{n - r-q} - \sqrt{q} - 1 - \sqrt{\tau}\sqrt{n-r}) / (1-\sqrt{\tau})\sqrt{n-r} \to 1$ as $n \to \infty$, for sufficiently large $n$, the exponent in the right hand of the above expression will be larger than $(1-\sqrt{\tau})\sqrt{n-r}/\sqrt{2}$. The proof is completed by observing that $n-r \geq n-\abs{M}$ and $\ubar{\lambda}_v > 0$, by Condition \ref{cond: eigenvalueV}.
\end{proof}

%**************************************************************************************************************************************************************************%
%.                                                                                      PROOF OF RATIO OF DETERMINANT OF RSS                                                                         % 
%**************************************************************************************************************************************************************************%                   

\begin{proof}[Proof of Theorem \ref{Lemma: RatioofRSS}]
\textit{Case 1}. First we consider models $M$ with $M \subsetneq M_{\text{o}}$. Let $r_{n,M} = 1 - \fpr{1+\xi_{n,\abs{M}}\ubar{\lambda}_v/\bar{\lambda}_v}/2$. Then for sufficiently large $n$, $r_{n,M} \in (0, 1/2)$. We take $\tau =  r_{n,M}$ in Lemma~\ref{Lemma: MinEigenofRSSMatrix} to get,
\begin{align*}
    \mathbb{P}_y\fpr{\lambda_{\min}\fpr{\widehat{\M{\Sigma}}_{(M)}} < r_{n,M}\ubar{\lambda}_v(n-\abs{M}) } \leq e^{-(1-1/\sqrt{2})^2(n-\abs{M})/2} <  e^{-0.04(n-\abs{M})} ,
\end{align*}
 Then, 
\begin{align*}
     &\mathbb{P}_y\left(\frac{\det \sighat{(M_{\text{o}})}}{\det \sighat{(M)}}  > \xi_{n,\abs{M}}^q \right) \\
     &\leq \mathbb{P}_y\fpr{ \det \sighat{(M_{\text{o}})} \geq \xi_{n,\abs{M}}^q\fpr{\lambda_{\min}\fpr{\widehat{\M{\Sigma}}_{(M)}}}^q } \hspace{70 mm} \\
     &\leq \mathbb{P}_y\fpr{ \det \sighat{(M_{\text{o}})} \geq \xi_{n,\abs{M}}^q\fpr{r_{n,M}\ubar{\lambda}_v(n-\abs{M})}^q} + \mathbb{P}\fpr{\lambda_{\min}\fpr{\widehat{\M{\Sigma}}_{(M)}} < r_{n,M}\ubar{\lambda}_v(n-\abs{M}) } \\
        &\leq   \mathbb{P}_y\fpr{ \frac{\det \sighat{(M_{\text{o}})}}{\det \bVMo^0} \geq \xi_{n,\abs{M}}^q\fpr{\frac{\ubar{\lambda}_v}{\bar{\lambda}_v}}^q\fpr{r_{n,M}(n-\abs{M})}^q} + e^{-0.04(n-\abs{M})},
     \end{align*}
and by Theorem $3.3.22$ of \cite{gupta2018matrix_sub} with $\xi^*_{n,\abs{M}} = \xi_{n,\abs{M}}\ubar{\lambda}_v/\bar{\lambda}_v \in (0,1)$,
     \begin{align*}
     &\mathbb{P}_y\left(\frac{\det \sighat{(M_{\text{o}})}}{\det \sighat{(M)}}  > \xi_{n,\abs{M}}^q \right)  \\
     &\leq \sum_{i=1}^q \mathbb{P}_y\fpr{\bigchi^2_{n- \abs{M_{\text{o}}} - i + 1} \geq r_{n,M}\xi^*_{n,\abs{M}} (n - \abs{M})} + e^{-0.04(n-\abs{M})}\\
     &\leq \sum_{i=1}^q\exp\fpr{-s\;r_{n,M} \;\xi^*_{n,\abs{M}} (n - \abs{M}) - \fpr{\frac{n- \abs{M_{\text{o}}} - i + 1}{2}} \log(1-2s)} +  e^{-0.04(n-\abs{M})} \\
     &\leq q\exp\fpr{-(n - \abs{M})\fpr{s\;r_{n,M} \;\xi^*_{n,\abs{M}}  + \frac{1}{4} \log(1-2s)}}  + e^{-0.04(n-\abs{M})},
\end{align*}
where the second to last inequality holds for any $s < 1/2$ by Chernoff's bound, and the last inequality holds for any $n > 2\abs{M_{\text{o}}} + 2q - \abs{M} \geq  2q + 1$. Taking $s =\frac{1}{2}\fpr{1- \frac{1}{2r_{n,M} \;\xi^*_{n,\abs{M}}}} < 0$,
\begin{align*}
   &s\;r_{n,M} \;\xi^*_{n,\abs{M}}  + \frac{1}{4} \log(1-2s)\\
    &= \frac{1}{4}\fpr{\xi^*_{n,\abs{M}}(1-\xi^*_{n,\abs{M}}) - 1 - \log\fpr{\xi^*_{n,\abs{M}}(1-\xi^*_{n,\abs{M}})}} \\
    &\geq \frac{1}{4}\fpr{\xi^*_{n,\abs{M}}(1-\xi^*_{n,\abs{M}}) - 1 + \log 4} \geq \frac{\xi_{n,\abs{M}}\ubar{\lambda}_vn^\alpha \log \abs{M_{\text{o}}}}{2\bar{\lambda}_v(n - \abs{M})} + 0.09 > 0,
\end{align*}
which implies that,
\begin{equation*}
    \mathbb{P}_y\left(\frac{\det \sighat{(M_{\text{o}})}}{\det \sighat{(M)}}  > \xi_{n,\abs{M}}^q \right) \leq q\exp\fpr{-\frac{\xi_{n,\abs{M}}n^\alpha \log \abs{M_{\text{o}}}\ubar{\lambda}_v}{2\bar{\lambda}_v} - 0.09 \fpr{n - \abs{M}}} + e^{-0.04(n-\abs{M})}.
\end{equation*}
Therefore, 
\begin{align*}
    &\mathbb{P}_y\fpr{\bigcup_{M \subsetneq M_{\text{o}}} \spr{\M{Y} : \left(\frac{\det \sighat{(M_{\text{o}})}}{\det \sighat{(M)}}\right)^{\frac{n-\abs{M}-q}{2}}  > \xi_{n,\abs{M}}^{\frac{q(n-\abs{M}-q)}{2}}}} \\
    &\leq \sum_{j=1}^{M_{\text{o}}} \binom{\abs{M_{\text{o}}}}{j} \underset{M \subsetneq M_{\text{o}}: \abs{M} = j}{\max} \;\mathbb{P}_y\fpr{\M{Y} : \left(\frac{\det \sighat{(M_{\text{o}})}}{\det \sighat{(M)}}\right)^{\frac{n-j-q}{2}}  > \xi_{n,j}^{\frac{q(n-j-q)}{2}}} \\
    &\leq \sum_{j=1}^{M_{\text{o}}} e^{j \log(\abs{M_{\text{o}}})}  \underset{M \subsetneq M_{\text{o}}: \abs{M} = j}{\max} \spr{ qe^{-\frac{\xi_{n,j}n^\alpha \log \abs{M_{\text{o}}}\ubar{\lambda}_v}{2\bar{\lambda}_v} - 0.09 \fpr{n - j} } + e^{-0.04(n-j)}  } \\
    &\leq \abs{M_{\text{o}}}q\exp\fpr{-\frac{\xi_{n,\abs{M_{\text{o}}}}n^\alpha \log \abs{M_{\text{o}}}\ubar{\lambda}_v}{2\bar{\lambda}_v} - 0.09 \fpr{n - \abs{M_{\text{o}}}} + \abs{M_{\text{o}}} \log(\abs{M_{\text{o}}})} \\
    & +  2\abs{M_{\text{o}}}\exp\fpr{-0.04(n-\abs{M_{\text{o}}}) + \abs{M_{\text{o}}}\log\abs{M_{\text{o}}}},
\end{align*}
The proof for Case 1 completes by noting that,
\begin{equation*}
    \left(\frac{\det \sighat{(M_{\text{o}})}}{\det \sighat{(M)}}\right)^{\frac{n-\abs{M}-q}{2}}  \leq  \xi_{n,\abs{M}}^{\frac{q(n-\abs{M}-q)}{2}} \leq  \fpr{1 - \frac{2n^\alpha\log\abs{M_\text{o}}}{n - \abs{M} - q}}^{q(n-\abs{M}-q)/2} \leq e^{-qn^\alpha\log\abs{M_\text{o}}}.
\end{equation*}

\textit{Case 2}. Fix an arbitrary model $M$ such that $M \not \subseteq M_{\text{o}}$ and $|M| \le n^{\alpha}$, and construct a new model $M^\prime := M \cup M_{\text{o}}$ where $\abs{M^{\prime}} = \abs{M} + \ell$ for some $\ell \in \spr{1,\dots,\abs{M_{\text{o}}}}$. Further, $M^\prime \supset M_{\text{o}}$ implies $\fpr{\M{H}_{(M^\prime)}-\bHMo}\fpr{\M{I}-\M{H}_{(M^\prime)}}=\M{0}$, which means $\widehat{\M{\Sigma}}_{(M_{\text{o}})}-\widehat{\M{\Sigma}}_{(M^\prime)}$ is independent of $\widehat{\M{\Sigma}}_{(M^\prime)}$. For sufficiently large $n$, $\sighat{(M^\prime)} \sim \textrm{Wishart}_q\fpr{n -r-\ell, \bVMo^0}$ and $\widehat{\M{\Sigma}}_{(M_{\text{o}})}-\widehat{\M{\Sigma}}_{(M^\prime)} \sim \textrm{Wishart}_q\fpr{r + \ell - \abs{M_{\text{o}}}, \bVMo^0}$, where $r := \text{rank}(\bX_{M\,\cdot}) \leq \abs{M}$. Without loss of generality for the bounds we derive it suffices to work with $r = \abs{M}$. By Theorem $10.5.3$ of \cite{muirhead2009aspects}, 
\begin{equation*}
\frac{\det\sighat{(M^\prime)}}{\det\sighat{(M_{\text{o}})}} \sim 
\begin{cases}
    \prod_{i=1}^q  V^{(1)}_i \quad \textrm{if } \abs{M} + \ell - \abs{M_{\text{o}}} \geq q \\ \\
    \prod_{i=1}^{\abs{M} + \ell - \abs{M_{\text{o}}}} V^{(2)}_i  \quad \textrm{if } \abs{M} + \ell - \abs{M_{\text{o}}} < q
\end{cases}
\end{equation*}
where $V^{(1)}_i$ are independently distributed  $\textrm{Beta}\fpr{\frac{1}{2}\fpr{n - \abs{M}-\ell-i + 1}, \frac{1}{2}\fpr{\abs{M} + \ell - \abs{M_{\text{o}}}}}$ random variables and $V^{(2)}_i$ are independently distributed $\textrm{Beta}\fpr{\frac{1}{2}\fpr{n - \abs{M_{\text{o}}} - q -i + 1}, \frac{q}{2}}$. We will handle the two cases separately. 

\textbf{Case 2a}: Suppose $\abs{M} + \ell - \abs{M_{\text{o}}} \geq q$. Because, $M^\prime \supset M$, $\widehat{\M{\Sigma}}_{(M)} - \widehat{\M{\Sigma}}_{(M^\prime)}$ is a positive definite matrix and that implies $\det{\widehat{\M{\Sigma}}_{(M)}} \geq \det{\widehat{\M{\Sigma}}_{(M^\prime)}}$. Hence,
\begin{align*}
&\mathbb{P}_y\left(\frac{\det \sighat{(M_{\text{o}})}}{\det \sighat{(M)}}  >  \zeta_{n,\abs{M}}^q \right)  \leq \mathbb{P}_y\left(\frac{\det \sighat{(M_{\text{o}})}}{\det \sighat{(M^\prime)}}  > \zeta_{n,\abs{M}}^q  \right) \leq \sum_{i=1}^q \mathbb{P}_y\left(V^{(1)}_i  <  \zeta_{n,\abs{M}}^{-1} \right).  
\end{align*}
The probabilities in the sum on the right side are bounded by approximating the CDF of the Beta density, and approximating of the beta function by Theorem 2 of \cite{jameson2013inequalities} as,
\begin{align*}
     &\mathbb{P}_y\left(V^{(1)}_i  <  \zeta_{n,\abs{M}}^{-1} \right) \\
     &= \frac{1}{B\fpr{\frac{n - \abs{M}-\ell-i + 1}{2}, \frac{\abs{M} + \ell - \abs{M_{\text{o}}}}{2}}} \int_{0}^{ \zeta_{n,\abs{M}}^{-1}} y^{\frac{1}{2}\fpr{n - \abs{M}-\ell-i + 1}-1}(1-y)^{\frac{1}{2}\fpr{\abs{M} + \ell - \abs{M_{\text{o}}}}-1} \;dy \\
     \leq& \frac{\fpr{\frac{n - \abs{M_{\text{o}}}-i + 1}{2}}^{\fpr{\frac{\abs{M} + \ell - \abs{M_{\text{o}}}}{2}}}}{\Gamma\fpr{\frac{\abs{M} + \ell - \abs{M_{\text{o}}}}{2}}}\zeta_{n,\abs{M}}^{-\fpr{\frac{n - \abs{M}-\ell-i - 1}{2}}}\tpr{1- \fpr{\frac{\zeta_{n,\abs{M}}-1}{\zeta_{n,\abs{M}}}}^{\frac{\abs{M} + \ell - \abs{M_{\text{o}}}}{2}}} \\
     &\leq \fpr{\frac{n - \abs{M_{\text{o}}}}{2}}^{\frac{\abs{M} }{2}}\zeta_{n,\abs{M}}^{-\fpr{\frac{n - \abs{M}-q}{2}} + \frac{\abs{M_{\text{o}}}+1}{2}} \cdot 1
     \end{align*}
% holds for $N \geq 2\abs{M_{\text{o}}} + \abs{M} + n$ 
\\

\textbf{Case 2b}: Suppose $\abs{M} + \ell - \abs{M_{\text{o}}} < q$.  Then similar to the previous case,
\begin{align*}
&\mathbb{P}_y\left(\frac{\det \sighat{(M_{\text{o}})}}{\det \sighat{(M)}}  >  \zeta_{n,\abs{M}}^q \right)  \leq \mathbb{P}_y\left(\frac{\det \sighat{(M_{\text{o}})}}{\det \sighat{(M^\prime)}}  > \zeta_{n,\abs{M}}^{\abs{M} + \ell - \abs{M_{\text{o}}}}  \right) \leq \sum_{i=1}^{\abs{M} + \ell - \abs{M_{\text{o}}}} \mathbb{P}_y\left(V^{(2)}_i  <  \zeta_{n,\abs{M}}^{-1} \right),  
\end{align*}
and,
\begin{align*}
     &\mathbb{P}_y\left(V^{(2)}_i  < \zeta_{n,\abs{M}}^{-1} \right)\\
     &= \frac{1}{B\fpr{\frac{n - \abs{M_{\text{o}}}-q-i + 1}{2}, \frac{q}{2}}} \int_{0}^{\zeta_{n,\abs{M}}^{-1}} y^{\frac{1}{2}\fpr{n - \abs{M_{\text{o}}}-q-i + 1}-1}(1-y)^{\frac{q}{2}-1} \;dy \\
     &\leq \frac{\fpr{\frac{n - \abs{M_{\text{o}}}-i + 1}{2}}^{\frac{q}{2}}}{\Gamma\fpr{\frac{q}{2}}}\zeta_{n,\abs{M}}^{-\fpr{\frac{n - \abs{M_{\text{o}}}-q-i - 1}{2}}}\tpr{1- \fpr{\frac{\zeta_{n,\abs{M}}-1}{\zeta_{n,\abs{M}}}}^{\frac{q}{2}}} \\
&\leq \fpr{\frac{n - \abs{M_{\text{o}}}}{2}}^{\frac{q}{2}}\zeta_{n,\abs{M}}^{-\fpr{\frac{n - \abs{M}-q}{2}}+\frac{\abs{M_{\text{o}}}+1}{2}}.
\end{align*}
Thus, in any case, for sufficiently large $n$,
\begin{align*}
    \mathbb{P}_y\left(\frac{\det \sighat{(M_{\text{o}})}}{\det \sighat{(M)}}  > \zeta_{n,\abs{M}}^q \right) &\leq  2q\exp\left(-n^{\alpha}\log\fpr{n-\abs{M_{\text{o}}}} + \max\spr{\frac{\abs{M}}{2}, \frac{q}{2}}\log\fpr{\frac{n-\abs{M_{\text{o}}}}{2}} \right. \\
& \hspace{3 cm} \left. - \abs{M}\log p  + \frac{(\abs{M_{\text{o}}}+1)}{2}\log\fpr{\zeta_{n,\abs{M}}}\right).
\end{align*}
Therefore,
\begin{align*}
    &\mathbb{P}_y\fpr{\bigcup_{M \not\subseteq M_{\text{o}} : \abs{M} \leq n^{\alpha}} \spr{\M{Y} : \left(\frac{\det \sighat{(M_{\text{o}})}}{\det \sighat{(M)}}\right)^{\frac{n-\abs{M}-q}{2}}  > \zeta_{n,\abs{M}}^{\frac{q(n-\abs{M}-q)}{2}}}} \\
    &\leq \sum_{j=1}^{n^\alpha} \binom{p}{j} \underset{M \not\subseteq M_{\text{o}}: \abs{M} = j}{\max} \;\mathbb{P}_y\fpr{\left(\frac{\det \sighat{(M_{\text{o}})}}{\det \sighat{(M)}}\right)^{\frac{n-j-q}{2}}  > \zeta_{n,j}^{\frac{q(n-j-q)}{2}}} \\
    &\leq \sum_{j=1}^{n^\alpha}  \left\{2q\exp\left(-n^{\alpha}\log\fpr{n-\abs{M_{\text{o}}}} + \max\spr{\frac{j}{2}, \frac{q}{2}}\log\fpr{n-\abs{M_{\text{o}}}} \right. \right. \\
& \hspace{4 cm} \left. \left. + \frac{(\abs{M_{\text{o}}}+1)}{2}\log\fpr{\zeta_{n,j}} \right) \right\} \\
    &\leq 2q\exp\fpr{- \frac{n^\alpha}{2}\log\fpr{n-\abs{M_{\text{o}}}} + \alpha\log n + \frac{(\abs{M_{\text{o}}}+1)}{2}\log\fpr{\zeta_{n,n^\alpha}} }.
\end{align*}
The proof is completed by noting that
$
  \zeta_{n,\abs{M}}^{\frac{q(n-\abs{M}-q)}{2}}  \leq e^{q\fpr{n^{\alpha}\log\fpr{n-\abs{M_{\text{o}}}} +\abs{M}\log p} } .
$

\end{proof}

\begin{proof}[Proof of Theorem \ref{Lemma: Ehtruemodel}]

To prove Theorem~\ref{Lemma: Ehtruemodel} and Theorem~\ref{Lemma: Ehbigmodels} we need two additional results, which are stated below. The proofs of Lemma~\ref{Lemma: FiducialDistBhat} and Lemma~\ref{Lemma: LeastSqEstDist} are provided after the proof of the Theorem~\ref{Lemma: Ehbigmodels}. Since these two lemmas are only necessary to bound the $\mathrm{E}(h_\epsilon(\bB_{\cdot\,M}))$ from above for the large models, without loss of generality we assume that $\bX_{M\,\cdot}$ is of full row rank, as in the statement of both the lemmas. If  $\bX_{M\,\cdot}$ is not of full rank, $h_\epsilon(\bB_{\cdot\,M})$ is automatically zero. 

\begin{lemma} \label{Lemma: FiducialDistBhat}
For model $M$ with $\abs{M} < n - 4q$,
\begin{equation*}
    \bB_{\cdot\,M}  \sim T_{q,M}\fpr{n-\abs{M}-q+1,\widehat{\bB}_{\cdot\,M},  \widehat{\M{\Sigma}}_{(M)}, \fpr{\bX_{M\,\cdot}\bX_{M\,\cdot}^\top}^{-1}},
\end{equation*}
where $\widehat{\bB}_{\cdot\,M} = \M{Y}\bX_{M\,\cdot}^\top \fpr{\bX_{M\,\cdot}\bX_{M\,\cdot}^\top}^{-1}$, and for any $\epsilon > 0$,
\begin{align*}
    \mathbb{P}\fpr{\frac{1}{2}\left\lVert \widehat{\M{\Sigma}}_{(M)}^{-1/2} \fpr{\bB_{\cdot\,M}\bX_{M\,\cdot} - \widehat{\bB}_{\cdot\,M}\bX_{M\,\cdot}}\right\rVert_{\textrm{F}}^2 \;\;\leq \frac{\epsilon}{9}} &\geq 1-V_{5,n,M}, 
\end{align*}
with, 
\begin{equation*}
    V_{5,n,M} := \exp\fpr{-\frac{\epsilon\fpr{n - \abs{M}}}{36} +   \frac{q\abs{M}}{2}} + 2\exp\fpr{-\frac{\fpr{\sqrt{n - \abs{M}}-2\sqrt{q}}^2}{8}}.
\end{equation*}
\end{lemma}

\begin{lemma}\label{Lemma: LeastSqEstDist}
For any model $M$ such that $\abs{M} < n - 4q$, 
%assume that $\M{Y} \vert \bB_{\cdot\,M}, \M{A}_{(M)} \sim \textrm{Matrix-Normal}_{q,n}\fpr{\bB_{\cdot\,M}\bX_{M\,\cdot},\M{V}_{(M)}, \M{I}_n}$, 
then the least-squared estimator, \\ $$\widehat{\bB}_{\cdot\,M}\sim \textrm{Matrix-Normal}_{q,\abs{M}}\fpr{\mathbb{E}_y\fpr{\widehat{\bB}_{\cdot\,M}}, \bVMo^0, \fpr{\bX_{M\,\cdot}\bX_{M\,\cdot}^\top}^{-1}},$$ and for sufficiently large $n$, 
\begin{align*}
    \mathbb{P}_y\fpr{\frac{1}{2}\left\lVert \widehat{\M{\Sigma}}_{(M)}^{-1/2} \tpr{\widehat{\bB}_{\cdot\,M} - \mathbb{E}_y\fpr{\widehat{\bB}_{\cdot\,M}}}\bX_{M\,\cdot}\right\rVert_{\textrm{F}}^2 \leq \frac{\epsilon}{9}} &\geq 1 - V_{6,n,M},
\end{align*}
where, 
\begin{align*}
    V_{6,n,M} &:=  \exp\fpr{-\frac{\epsilon\fpr{n-\abs{M}}\ubar{\lambda}_v}{36\;\bar{\lambda}_v} + \frac{q\abs{M}}{2}}  +  \exp\fpr{-0.04(n-\abs{M})}.
    \end{align*}
\end{lemma}

Let $\widetilde{\M{B}}_{\min}$ minimize the objective function $\frac{1}{2}\left\lVert \bAMo^{0^{-1}} \fpr{\bB^{0}_{\cdot\,M_{\text{o}}}\bX_{M_{\text{o}}\,\cdot} - \M{B}\M{X}}\right\rVert_{\textrm{F}}^2$ subject to $\abs{\spr{j: \norm{\bB_j} \neq 0}} \leq \abs{M_\text{o}}-1$. Also, suppose  $\M{B}_{\min}$ minimizes $\left\lVert \widehat{\M{\Sigma}}_{(M_{\text{o}})}^{-1/2} \fpr{\bB_{\cdot\,M_{\text{o}}}\bX_{M_{\text{o}}\,\cdot} - \M{B}\M{X}}\right\rVert_{\textrm{F}}^2$ subject to $\abs{\spr{j: \norm{\bB_j} \neq 0}} \leq \abs{M_{\text{o}}}-1$. Then,
\begin{align*}
    \left\lVert \bAMo^{0^{-1}} \fpr{\bB^{0}_{\cdot\,M_{\text{o}}}\bX_{M_{\text{o}}\,\cdot} - \widetilde{\M{B}}_{\min}\M{X}}\right\rVert_{\textrm{F}} &\leq   \left\lVert \bAMo^{0^{-1}} \fpr{\bB^{0}_{\cdot\,M_{\text{o}}}\bX_{M_{\text{o}}\,\cdot} - \M{B}_{\min}\M{X}}\right\rVert_{\textrm{F}} \\
    &\leq \left\lVert \bAMo^{0^{-1}}\widehat{\M{\Sigma}}_{(M_{\text{o}})}^{1/2} \right\rVert_{\textrm{F}}\left\lVert \widehat{\M{\Sigma}}_{(M_{\text{o}})}^{-1/2} \fpr{\bB^{0}_{\cdot\,M_{\text{o}}}\bX_{M_{\text{o}}\,\cdot} - \M{B}_{\min}\M{X}}\right\rVert_{\textrm{F}}
\end{align*}
Now, $\left\lVert \bAMo^{0^{-1}}\widehat{\M{\Sigma}}_{(M_{\text{o}})}^{1/2} \right\rVert_{\textrm{F}}^2 = \tr\fpr{\bAMo^{0^{-1}}\widehat{\M{\Sigma}}_{(M_{\text{o}})}\bAMo^{0^{-\top}}}$. Now, $\widehat{\M{\Sigma}}_{(M_{\text{o}})} \sim \textrm{Wishart}_q\fpr{n - \abs{M_{\text{o}}}, \bVMo^0}$ implies $\bAMo^{0^{-1}}\widehat{\M{\Sigma}}_{(M_{\text{o}})}\bAMo^{0^{-T}} \sim \textrm{Wishart}_q\fpr{n - \abs{M_{\text{o}}}, \M{I}_q}$. This means, $\left\lVert \bAMo^{0^{-1}}\widehat{\M{\Sigma}}_{(M_{\text{o}})}^{1/2} \right\rVert_{\textrm{F}}^2 \sim \bigchi^2$ distribution with $q\fpr{n-\abs{M_{\text{o}}}}$ degrees of freedom.
Using triangle inequality,
\begin{align}
    \nonumber&\left\lVert \widehat{\M{\Sigma}}_{(M_{\text{o}})}^{-1/2} \fpr{\bB^{0}_{\cdot\,M_{\text{o}}}\bX_{M_{\text{o}}\,\cdot} - \M{B}_{\min}\M{X}}\right\rVert_{\textrm{F}} \\
    \nonumber &\leq \left\lVert \widehat{\M{\Sigma}}_{(M_{\text{o}})}^{-1/2} \tpr{\bB_{\cdot\,M_{\text{o}}} - \widehat{\bB}_{\cdot\,M_{\text{o}}}}\bX_{M_{\text{o}}\,\cdot}\right\rVert_{\textrm{F}} +
    \left\lVert \widehat{\M{\Sigma}}_{(M_{\text{o}})}^{-1/2} \fpr{\widehat{\bB}_{\cdot\,M_{\text{o}}} - \bB^{0}_{\cdot\,M_{\text{o}}}}\bX_{M_{\text{o}}\,\cdot}\right\rVert_{\textrm{F}} \\
&\hspace{20pt} + 
    \left\lVert \widehat{\M{\Sigma}}_{(M_{\text{o}})}^{-1/2} \fpr{\bB_{\cdot\,M_{\text{o}}}\bX_{M_{\text{o}}\,\cdot} - \M{B}_{\min}\M{X}}\right\rVert_{\textrm{F}}.
\end{align}
Then,
\begin{align*}
    &\M{I}\fpr{\left\lVert \bAMo^{0^{-1}} \fpr{\bB^{0}_{\cdot\,M_{\text{o}}}\bX_{M_{\text{o}}\,\cdot} - \widetilde{\M{B}}_{\min}\M{X}}\right\rVert_{\textrm{F}}^2 > 36q\;(n-\abs{M_{\text{o}}})\epsilon } \\
    &\leq \M{I}\fpr{\frac{1}{2}\left\lVert \widehat{\M{\Sigma}}_{(M_{\text{o}})}^{-1/2} \tpr{\bB_{\cdot\,M_{\text{o}}} - \widehat{\bB}_{\cdot\,M_{\text{o}}}}\bX_{M_{\text{o}}\,\cdot}\right\rVert_{\textrm{F}}^2 > \epsilon} + 
    \M{I}\fpr{\frac{1}{2}\left\lVert \widehat{\M{\Sigma}}_{(M_{\text{o}})}^{-1/2} \fpr{\widehat{\bB}_{\cdot\,M_{\text{o}}} - \bB^{0}_{\cdot\,M_{\text{o}}}}\bX_{M_{\text{o}}\,\cdot}\right\rVert_{\textrm{F}}^2 > \epsilon} \\
    & \qquad + 
    \M{I}\fpr{\frac{1}{2}\left\lVert \widehat{\M{\Sigma}}_{(M_{\text{o}})}^{-1/2} \fpr{\bB_{\cdot\,M_{\text{o}}}\bX_{M_{\text{o}}\,\cdot} - \M{B}_{\min}\M{X}}\right\rVert_{\textrm{F}}^2 > \epsilon} + \M{I}\fpr{\left\lVert \bAMo^{0^{-1}}\widehat{\M{\Sigma}}_{(M_{\text{o}})}^{1/2} \right\rVert_{\textrm{F}}^2 > 2q(n-\abs{M_{\text{o}}}) }.
\end{align*}
Taking expectation with respect to Fiducial distribution of $\bB_{\cdot\,M_{\text{o}}}$ given $\M{Y}$,
\begin{align}
    \nonumber &\M{I}\fpr{\left\lVert \bAMo^{0^{-1}} \fpr{\bB^{0}_{\cdot\,M_{\text{o}}}\bX_{M_{\text{o}}\,\cdot} - \widetilde{\M{B}}_{\min}\M{X}}\right\rVert_{\textrm{F}}^2 > 36q\fpr{n-\abs{M_{\text{o}}}}\epsilon } \\
    \nonumber &\leq 
    \mathbb{P}\fpr{\frac{1}{2}\left\lVert \widehat{\M{\Sigma}}_{(M_{\text{o}})}^{-1/2} \tpr{\bB_{\cdot\,M_{\text{o}}} - \widehat{\bB}_{\cdot\,M_{\text{o}}}}\bX_{M_{\text{o}}\,\cdot}\right\rVert_{\textrm{F}}^2 > \epsilon} 
    + 
    \M{I}\fpr{\frac{1}{2}\left\lVert \widehat{\M{\Sigma}}_{(M_{\text{o}})}^{-1/2} \fpr{\widehat{\bB}_{\cdot\,M_{\text{o}}} - \bB^{0}_{\cdot\,M_{\text{o}}}}\bX_{M_{\text{o}}\,\cdot}\right\rVert_{\textrm{F}}^2 > \epsilon} \\
    & \qquad 
    + \mathbb{E}\fpr{h_\epsilon\fpr{\bB_{\cdot\,M_{\text{o}}}}}
    +
    \M{I}\fpr{\left\lVert \bAMo^{0^{-1}}\widehat{\M{\Sigma}}_{(M_{\text{o}})}^{1/2} \right\rVert_{\textrm{F}}^2 > 2q(n-\abs{M_{\text{o}}}) } \label{eqn: proofEhtruemodels2}.
\end{align}
Next,
\begin{align*}
    &\mathbb{P}_y\left(\left\{ \M{I}\fpr{\frac{1}{2}\left\lVert \widehat{\M{\Sigma}}_{(M_{\text{o}})}^{-1/2} \fpr{\widehat{\bB}_{\cdot\,M_{\text{o}}} - \bB^{0}_{\cdot\,M_{\text{o}}}}\bX_{M_{\text{o}}\,\cdot}\right\rVert_{\textrm{F}}^2 > \epsilon} = 0\right\} \;\;\; \bigcap \right. \\
& \hspace{5 cm}\left. \left \{\M{I}\fpr{\left\lVert \bAMo^{0^{-1}}\widehat{\M{\Sigma}}_{(M_{\text{o}})}^{1/2} \right\rVert_{\textrm{F}}^2 > 2q(n-\abs{M_{\text{o}}}) } = 0 \right \}\right) \\
    &\geq \mathbb{P}_y\fpr{\frac{1}{2}\left\lVert \widehat{\M{\Sigma}}_{(M_{\text{o}})}^{-1/2} \fpr{\widehat{\bB}_{\cdot\,M_{\text{o}}} - \bB^{0}_{\cdot\,M_{\text{o}}}}\bX_{M_{\text{o}}\,\cdot}\right\rVert_{\textrm{F}}^2 \leq \epsilon}
    +
    \mathbb{P}_y\fpr{\left\lVert \bAMo^{0^{-1}}\widehat{\M{\Sigma}}_{(M_{\text{o}})}^{1/2} \right\rVert_{\textrm{F}}^2 \leq 2q(n-\abs{M_{\text{o}}})} - 1 \\
    &\geq 1 - \exp\fpr{-\frac{\epsilon\fpr{n-\abs{M_{\text{o}}}}\ubar{\lambda}_{v}}{4\;\bar{\lambda}_{v}} + \frac{q\abs{M_{\text{o}}}}{2}} -  \exp\fpr{-0.04(n-\abs{M_{\text{o}}})}  -  \exp\fpr{-0.15q(n-\abs{M_{\text{o}}})}, 
\end{align*}
where the first probability is obtained by an application of Lemma~\ref{Lemma: LeastSqEstDist} and the second probability is computed by the Chernoff bound for the $\chi^2$ distribution with $q(n-\abs{M_{\text{o}}})$ degrees of freedom evaluated at $1/4$. The proof is complete by applying Lemma~\ref{Lemma: FiducialDistBhat} to the first term in~(\ref{eqn: proofEhtruemodels2}). 
\end{proof}

\begin{proof}[Proof of Theorem \ref{Lemma: Ehbigmodels}]
Let, $j^* := \underset{j}{\argmin} \;\; \lVert \widehat{\bB}_{\cdot\,M,j} \rVert_2 $, where $\widehat{\bB}_{\cdot\,M,j}$ is the $j$th column of the least square coefficient matrix $\widehat{\bB}_{\cdot\,M}$ for model $M$.  Construct the model $M(-1) := M \;\backslash\; \{j^*\}$ with $\abs{M}-1$ covariates from model $M$. Suppose $\widehat{\bB}_{\cdot\,M(-1)}$ be the least-squared estimator
corresponding to the model $M(-1)$. Because $\M{B}_{\min}$ minimizes the objective function corresponding to the $h$ function,
\begin{equation*}
    \frac{1}{2}\left\lVert \widehat{\M{\Sigma}}_{(M)}^{-1/2} \fpr{\bB_{\cdot\,M}\bX_{M\,\cdot} - \M{B}_{\min}\M{X}}\right\rVert_{\textrm{F}}^2 \leq \frac{1}{2}\left\lVert \widehat{\M{\Sigma}}_{(M)}^{-1/2} \fpr{\bB_{\cdot\,M}\bX_{M\,\cdot} - \mathbb{E}_y\fpr{\widehat{\bB}_{\cdot\,M(-1)}}\bX_{M(-1)\,\cdot}}\right\rVert_{\textrm{F}}^2.
\end{equation*}
By triangle inequality for the Frobenius norm,
\begin{align}
    \nonumber \mathbb{E}\fpr{h_\epsilon(\bB_{\cdot\,M})} &= \mathbb{P}\fpr{\frac{1}{2}\left\lVert \widehat{\M{\Sigma}}_{(M)}^{-1/2} \fpr{\bB_{\cdot\,M}\bX_{M\,\cdot} - \M{B}_{\min}\M{X}}\right\rVert_{\textrm{F}}^2 \;\;\geq \epsilon} \\
    \nonumber &\leq \mathbb{P}\fpr{\frac{1}{2}\left\lVert \widehat{\M{\Sigma}}_{(M)}^{-1/2} \fpr{\bB_{\cdot\,M}\bX_{M\,\cdot} - \mathbb{E}_y\fpr{\widehat{\bB}_{\cdot\,M(-1)}}\bX_{M(-1)\,\cdot}}\right\rVert_{\textrm{F}}^2 \geq \epsilon} \\
   \nonumber &\leq  \mathrm{I}\fpr{\frac{1}{2}\left\lVert \widehat{\M{\Sigma}}_{(M)}^{-1/2} \fpr{\widehat{\bB}_{\cdot\,M}\bX_{M\,\cdot} - \mathbb{E}_y\tpr{\widehat{\bB}_{\cdot\,M}}\bX_{M\,\cdot}}\right\rVert_{\textrm{F}}^2 \;\;\geq \frac{\epsilon}{9}} \\
    \nonumber & \qquad + \mathrm{I}\fpr{\frac{1}{2}\left\lVert \widehat{\M{\Sigma}}_{(M)}^{-1/2} \fpr{\mathbb{E}_y\tpr{\widehat{\bB}_{\cdot\,M}}\bX_{M\,\cdot} - \mathbb{E}_y\fpr{\widehat{\bB}_{\cdot\,M(-1)}}\bX_{M(-1)\,\cdot}}\right\rVert_{\textrm{F}}^2 \;\;\geq \frac{\epsilon}{9}}  \\
 \nonumber & + \mathbb{P}\fpr{\frac{1}{2}\left\lVert \widehat{\M{\Sigma}}_{(M)}^{-1/2} \fpr{\bB_{\cdot\,M}\bX_{M\,\cdot} - \widehat{\bB}_{\cdot\,M}\bX_{M\,\cdot}}\right\rVert_{\textrm{F}}^2 \;\;\geq \frac{\epsilon}{9}} \\
    \nonumber &\leq  \mathrm{I}\fpr{\frac{1}{2}\left\lVert \widehat{\M{\Sigma}}_{(M)}^{-1/2} \fpr{\widehat{\bB}_{\cdot\,M}\bX_{M\,\cdot} - \mathbb{E}_y\tpr{\widehat{\bB}_{\cdot\,M}}\bX_{M\,\cdot}}\right\rVert_{\textrm{F}}^2 \;\;\geq \frac{\epsilon}{9}} \\
\nonumber & + \mathrm{I}\fpr{\frac{1}{2}\left \lVert \widehat{\M{\Sigma}}_{(M)}^{-1/2} \bB_{\cdot\,M_{\text{o}}}^0 \bX_{M_{\text{o}}\,\cdot}\fpr{\bH_{(M)} - \M{H}_{(M)(-1)}} \right \rVert_{\textrm{F}}^2 \;\;\geq \frac{\epsilon}{9}} \\
&+  \exp\fpr{-\frac{\epsilon\fpr{n - \abs{M}}}{36} +  \frac{q\abs{M}}{2}} + 2\exp\fpr{-\frac{\fpr{\sqrt{n - \abs{M}}-2\sqrt{q}}^2}{8}}  , \label{eqn: proofEhbigmodels1}
\end{align}
where the first probability is computed by Lemma \ref{Lemma: FiducialDistBhat}. Next, 
{\small
\begin{align}
    \nonumber & \mathbb{P}_y \left(\spr{\left\lVert \widehat{\M{\Sigma}}_{(M)}^{-1/2} \fpr{\widehat{\bB}_{\cdot\,M}\bX_{M\,\cdot} - \mathbb{E}\tpr{\widehat{\bB}_{\cdot\,M}}\bX_{M\,\cdot}}\right\rVert_{\textrm{F}}^2 \;\; < \frac{2\epsilon}{9}} \;\;\bigcap\;\; \right. \\
 \nonumber &\hspace{5 cm} \left. \spr{\left \lVert \widehat{\M{\Sigma}}_{(M)}^{-1/2} \bB_{\cdot\,M_{\text{o}}}^0 \bX_{M_{\text{o}}\,\cdot}\fpr{\bH_{(M)} - \M{H}_{(M)(-1)}} \right \rVert_{\textrm{F}}^2 \;\; < \frac{2\epsilon}{9}} \right)  \\
\nonumber   & \geq 1 -  \mathbb{P}_y\fpr{\frac{1}{2}\left\lVert \widehat{\M{\Sigma}}_{(M)}^{-1/2} \tpr{\widehat{\bB}_{\cdot\,M} - \mathbb{E}\fpr{\widehat{\bB}_{\cdot\,M}}}\bX_{M\,\cdot}\right\rVert_{\textrm{F}}^2 \geq \frac{\epsilon}{9}} \\
& \hspace{5cm} - \mathbb{P}_y\fpr{\frac{1}{2}\left \lVert \widehat{\M{\Sigma}}_{(M)}^{-1/2} \bB_{\cdot\,M_{\text{o}}}^0 \bX_{M_{\text{o}}\,\cdot}\fpr{\bH_{(M)} - \M{H}_{(M)(-1)}} \right \rVert_{\textrm{F}}^2 \geq \frac{\epsilon}{9}}. \label{eqn: proofEhbigmodels2}
\end{align}
}
By Condition~\ref{cond: lowerboundepsilon} and putting $\tau = 1/2$ in Lemma~\ref{Lemma: MinEigenofRSSMatrix},
\begin{align}
    \nonumber &\mathbb{P}_y\fpr{\frac{1}{2}\left \lVert \widehat{\M{\Sigma}}_{(M)}^{-1/2} \bB_{\cdot\,M_{\text{o}}}^0 \bX_{M_{\text{o}}\,\cdot}\fpr{\bH_{(M)} - \M{H}_{(M)(-1)}} \right \rVert_{\textrm{F}}^2 \geq \frac{\epsilon}{9}} \\ 
    \nonumber &\leq \mathbb{P}_y\fpr{\lambda_{\min}\fpr{\widehat{\M{\Sigma}}_{(M)}}\leq \frac{9}{2\epsilon}\left \lVert \bB_{\cdot\,M_{\text{o}}}^0 \bX_{M_{\text{o}}\,\cdot}\fpr{\bH_{(M)} - \M{H}_{(M)(-1)}} \right \rVert_{\textrm{F}}^2} \\
     &\leq \mathbb{P}_y\fpr{\lambda_{\textrm{min}}\fpr{\widehat{\M{\Sigma}}_{(M)}} \leq  \frac{\ubar{\lambda}_v(n-\abs{M})}{2}} 
     \leq e^{-0.04(n-\abs{M})}. \label{eqn: proofEhbigmodels3}
\end{align}
Bounding the first probability in equation~(\ref{eqn: proofEhbigmodels2}) by Lemma~\ref{Lemma: LeastSqEstDist} and combining equation~(\ref{eqn: proofEhbigmodels1}),(\ref{eqn: proofEhbigmodels2}), and (\ref{eqn: proofEhbigmodels3}),
\small{\begin{equation*}
    \mathbb{P}_y\tpr{\mathbb{E}\fpr{h_\epsilon\fpr{\bB_{\cdot\,M}}} \geq \exp\fpr{-\frac{\epsilon\fpr{n - \abs{M}}}{36} +  \frac{q\abs{M}}{2}} + 2\exp\fpr{-\frac{1}{8}\spr{\sqrt{n - \abs{M}}-2\sqrt{q}}^2}} 
    \leq V_{7,n,M} \hspace{50 mm}
\end{equation*} }
where 
\begin{align*}
    V_{7,n,M} &:=  \exp\fpr{-\frac{\epsilon\fpr{n-\abs{M}}\ubar{\lambda}_v}{36\;\bar{\lambda}_v} + \frac{q\abs{M}}{2}}  + 2\exp\fpr{-0.04(n-\abs{M})}
\end{align*}
Finally,
\begin{align*}
        &\mathbb{P}_y\tpr{\bigcup_{\substack{M \not\subset M_{\text{o}} \\ \abs{M} \leq n^\alpha}} \spr{\bY:\mathbb{E}\fpr{h_\epsilon\fpr{\bB_{\cdot\,M}}} \geq \exp\fpr{-\frac{\epsilon\fpr{n - \abs{M}}}{36} +  \frac{q\abs{M}}{2}} + 2\exp\fpr{-\frac{1}{8}\spr{\sqrt{n - \abs{M}}-2\sqrt{q}}^2}}} 
        \\ &\leq   \sum_{j=1}^{n^\alpha} \sum_{\substack{M: \abs{M}=j \\ M \not \subseteq M_{\text{o}}}} \spr{\exp\fpr{-\frac{\epsilon\fpr{n-\abs{M}}\ubar{\lambda}_v}{36\;\bar{\lambda}_v} + \frac{q\abs{M}}{2}} + 2\exp\fpr{-0.04(n-\abs{M})} } \\ % here new line
    &\leq  \sum_{j=1}^{n^\alpha}  \spr{\exp\fpr{-\frac{\epsilon\fpr{n-j}\ubar{\lambda}_v}{36\;\bar{\lambda}_v} + \frac{qj}{2} + j\log p} + 2\exp\fpr{ -0.04(n-j) + j\log p}} \\  %new line here
    &\leq    n^\alpha\spr{\exp\fpr{-\frac{\epsilon\fpr{n-n^\alpha}\ubar{\lambda}_v}{36\;\bar{\lambda}_v} + \frac{qn^\alpha}{2} + n^\alpha\log p} + 2\exp\fpr{ -0.04(n-n^{\alpha}) + n^\alpha\log p}}. 
    \end{align*} 
This concludes the proof of the theorem.
\end{proof}

\begin{proof}[Proof of Lemma \ref{Lemma: FiducialDistBhat}]
 Defining $\widetilde{\bB}_{\cdot\,M}=\widehat{\M{\Sigma}}_{(M)}^{-1/2} \fpr{\bB_{\cdot\,M} - \widehat{\bB}_{\cdot\,M}}$ and $\M{\Omega}_{(M)} = \bX_{M\,\cdot}\bX_{M\,\cdot}^\top$, observe that,
\begin{align*}
    \left\lVert \widehat{\M{\Sigma}}_{(M)}^{-1/2} \fpr{\bB_{\cdot\,M}\bX_{M\,\cdot} - \widehat{\bB}_{\cdot\,M}\bX_{M\,\cdot}}\right\rVert_{\textrm{F}}^2 &= \left\lVert \widehat{\M{\Sigma}}_{(M)}^{-1/2} \fpr{\bB_{\cdot\,M} - 
 \widehat{\bB}_{\cdot\,M}}\bX_{M\,\cdot}\right\rVert_{\textrm{F}}^2 \\ 
    &= \tr\fpr{\widetilde{\bB}_{\cdot\,M} \M{\Omega}_{(M)} \widetilde{\bB}_{\cdot\,M}^\top} \\
    &=\left\lVert \widehat{\M{\Sigma}}_{(M)}^{-1/2} \fpr{\bB_{\cdot\,M} - \widehat{\bB}_{\cdot\,M}}\M{\Omega}_{(M)}^{1/2}\right\rVert_{\textrm{F}}^2.
\end{align*}
By the property of Matrix-t distribution, (see Theorem 4.3.5, page 137 of \cite{gupta2018matrix_sub})
\begin{equation*}
    \widehat{\M{\Sigma}}_{(M)}^{-1/2} \fpr{\bB_{\cdot\,M} - \widehat{\bB}_{\cdot\,M}}\M{\Omega}_{(M)}^{1/2} \sim T_{q,M}\fpr{n-\abs{M}-q+1,\M{0},  \M{I}_q, \M{I}_{\abs{M}}}.
\end{equation*}
Additionally, by Theorem $4.2.1$ (page $134$) of \cite{gupta2018matrix_sub},
\begin{equation*}
    \widehat{\M{\Sigma}}_{(M)}^{-1/2} \fpr{\bB_{\cdot\,M} - \widehat{\bB}_{\cdot\,M}}\M{\Omega}_{(M)}^{1/2} \overset{d}{=} \fpr{\M{W}^{-1/2}}^\top\M{Z},
\end{equation*}
where $\M{W} \sim \text{Wishart}_q\fpr{n-\abs{M}, \M{I}_q}$ and $\M{Z} \sim \text{Matrix-Normal}\fpr{\M{0}, \M{I}_q, \M{I}_{\abs{M}}}$. Then,
\begin{align*}
    \left\lVert \widehat{\M{\Sigma}}_{(M)}^{-1/2} \fpr{\bB_{\cdot\,M}\bX_{M\,\cdot} - \widehat{\bB}_{\cdot\,M}\bX_{M\,\cdot}}\right\rVert_{\textrm{F}}^2 &= \tr\fpr{\M{Z}^\top\M{W}^{-1/2}\M{W}^{{-1/2}^\top}\M{Z}}
    \leq \lambda_{\min}^{-1}\fpr{\M{W}} \tr\fpr{\M{Z}\M{Z}^\top},
\end{align*}
where $\lambda_{\min}\fpr{\M{W}} > 0$ with probability $1$ as $n-\abs{M} > q$. Since, $\tr\fpr{\M{Z}\M{Z}^\top} \sim \chi^2_{q\abs{M}}$,
\begin{align*}
     &\mathbb{P}\fpr{\frac{1}{2}\left\lVert \widehat{\M{\Sigma}}_{(M)}^{-1/2} \fpr{\bB_{\cdot\,M}\bX_{M\,\cdot} - \widehat{\bB}_{\cdot\,M}\bX_{M\,\cdot}}\right\rVert_{\textrm{F}}^2 \;\;\geq \frac{\epsilon}{9}} \\
     &\leq \mathbb{P}\fpr{\frac{\tr\fpr{\M{Z}\M{Z}^\top}}{\lambda_{\min}\fpr{\M{W}}}  \;\geq \frac{2\epsilon}{9}} \hspace{70 mm} \\
     &\leq \mathbb{P}\fpr{\frac{\tr\fpr{\M{Z}\M{Z}^\top}}{\lambda_{\min}\fpr{\M{W}}}  \;\geq \frac{2\epsilon}{9} \;,\;  \lambda_{\min}\fpr{\M{W}} \geq \frac{\fpr{n - \abs{M}}}{4}}  \\
& \hspace{2 cm} +   \mathbb{P}\fpr{\frac{\tr\fpr{\M{Z}\M{Z}^\top}}{\lambda_{\min}\fpr{\M{W}}}  \;\geq \frac{2\epsilon}{9} \;\;,\;\;  \lambda_{\min}\fpr{\M{W}} < \frac{\fpr{n - \abs{M}}}{4}} \\
     &\leq  \mathbb{P}\fpr{\tr\fpr{\M{Z}\M{Z}^\top} \;\geq \frac{\epsilon\fpr{n - \abs{M}}}{18}} + \mathbb{P}\fpr{\lambda_{\min}\fpr{\M{W}} < \frac{1}{4}\fpr{n - \abs{M}}} \\
     &\leq \mathbb{P}\fpr{\chi^2_{q\abs{M}} \;\geq \frac{\epsilon\fpr{n - \abs{M}}}{18}} + \mathbb{P}\fpr{\lambda_{\min}\fpr{\M{W}} < \frac{1}{4}\fpr{n - \abs{M}}} \\
     &\leq \exp\fpr{-\frac{\epsilon\fpr{n - \abs{M}}}{72} +  \frac{q\abs{M}}{2}} + 2\exp\fpr{-\frac{\fpr{\sqrt{n - \abs{M}}-2\sqrt{q}}^2}{8}},
\end{align*}
where the first probability in the second to last line is obtained by using the Chernoff's bound for $\chi^2$ distribution and and the second probability by substituting $t = (1/2)\sqrt{n-\abs{M}} - \sqrt{q} > 0$ (by the condition on $\abs{M}$ in the statement of the lemma) in the corollary $5.35$ (page $21$) of \cite{vershynin2010introduction_sub} and noting that any central $\textrm{Wishart}_q\fpr{n-\abs{M}, \M{I}_q}$ matrix is identically distributed as $\widetilde{\M{Z}}^{\top}\widetilde{\M{Z}}$ where $\widetilde{\M{Z}} \sim \textrm{Matrix-Normal}_{n-\abs{M},q}\fpr{0, \M{I}_{n-\abs{M}}, \M{I}_q}$, provided $n-\abs{M} > q$ which is again true by the specified condition in the lemma. This completes the proof. 
\end{proof}

\begin{proof}[Proof of Lemma \ref{Lemma: LeastSqEstDist}]
Defining, $\M{\Omega}_{(M)} = \bX_{M\,\cdot}\bX_{M\,\cdot}^\top$, the quadratic form can be alternatively written as
\begin{align*}
    &\left\lVert \widehat{\M{\Sigma}}_{(M)}^{-1/2} \tpr{\widehat{\bB}_{\cdot\,M}\bX_{M\,\cdot} - \mathbb{E}_y\fpr{\widehat{\bB}_{\cdot\,M}}\bX_{M\,\cdot}}\right\rVert_{\textrm{F}}^2 \\
    &= \tr\tpr{\widehat{\M{\Sigma}}_{(M)}^{-1} \fpr{\widehat{\bB}_{\cdot\,M} - \mathbb{E}_y\fpr{\widehat{\bB}_{\cdot\,M}}}\M{\Omega}_{(M)} \fpr{\widehat{\bB}_{\cdot\,M} - \mathbb{E}_y\fpr{\widehat{\bB}_{\cdot\,M}}}^\top} \\
    &\leq \lambda_{\textrm{min}}^{-1}\fpr{\widehat{\M{\Sigma}}_{(M)}}\bar{\lambda}_v \left\lVert \bAMo^{0^{-1}}\fpr{\widehat{\bB}_{\cdot\,M} - \mathbb{E}_y\fpr{\widehat{\bB}_{\cdot\,M}}}\M{\Omega}_{(M)}^{1/2} \right\rVert_{\textrm{F}}^2. 
\end{align*}
Since, $\bAMo^{0^{-1}}\fpr{\widehat{\bB}_{\cdot\,M} - \mathbb{E}_y\fpr{\widehat{\bB}_{\cdot\,M}}}\M{\Omega}_{(M)}^{1/2} \sim \textrm{Matrix-Normal}_{q,\abs{M}}\fpr{\M{0}, \M{I}_q, \M{I}_{\abs{M}}}$,
\begin{equation*}
\left\lVert \bAMo^{0^{-1}}\fpr{\widehat{\bB}_{\cdot\,M} - \mathbb{E}_y\fpr{\widehat{\bB}_{\cdot\,M}}}\M{\Omega}_{(M)}^{1/2} \right\rVert_{\textrm{F}}^2 \sim \chi^2_{q\abs{M}}.
\end{equation*}
Choosing $\tau = 1/2$ in Lemma~\ref{Lemma: MinEigenofRSSMatrix},
\begin{align*}
    \mathbb{P}_y\fpr{\lambda_{\min}\fpr{\widehat{\M{\Sigma}}_{(M)}} < \ubar{\lambda}_v(n-\abs{M})/2 } \leq \exp\fpr{-0.4(n-\abs{M})}.
\end{align*}
Finally,
\begin{align*}
     &\mathbb{P}_y\fpr{\frac{1}{2}\left\lVert \widehat{\M{A}}_{(M)}^{-1} \tpr{\widehat{\bB}_{\cdot\,M} - \mathbb{E}\fpr{\widehat{\bB}_{\cdot\,M}}}\bX_{M\,\cdot}\right\rVert_{\textrm{F}}^2 \geq \frac{\epsilon}{9}}\\
     &\leq 
     \mathbb{P}_y\fpr{\lambda_{\textrm{min}}^{-1}\fpr{\widehat{\M{\Sigma}}_{(M)}}\bar{\lambda}_v \left\lVert \bVMo^{0^{-1/2}}\fpr{\widehat{\bB}_{\cdot\,M} - \mathbb{E}_y\fpr{\widehat{\bB}_{\cdot\,M}}}\M{\Omega}_{(M)}^{1/2} \right\rVert_{\textrm{F}}^2 \geq \frac{2\epsilon}{9}} \qquad \qquad \qquad \hspace{20 mm} \\
     &\leq \mathbb{P}_y\fpr{\frac{\bar{\lambda}_v\bigchi^2_{q\abs{M}}}{\lambda_{\textrm{min}}\fpr{\widehat{\M{\Sigma}}_{(M)}}}  \geq  \frac{2\epsilon}{9}} \\
     &\leq \mathbb{P}_y\fpr{\frac{\bar{\lambda}_v\bigchi^2_{q\abs{M}}}{\lambda_{\textrm{min}}\fpr{\widehat{\M{\Sigma}}_{(M)}}}  \geq  \frac{2\epsilon}{9}, \;\;\lambda_{\textrm{min}}\fpr{\widehat{\M{\Sigma}}_{(M)}} > \frac{\ubar{\lambda}_v(n-\abs{M})}{2}}  \\
  & \hspace{2 cm} + \mathbb{P}_y\fpr{\lambda_{\textrm{min}}\fpr{\widehat{\M{\Sigma}}_{(M)}} \leq  \frac{\ubar{\lambda}_v(n-\abs{M})}{2}} \\ 
     &\leq \mathbb{P}_y\fpr{\bigchi^2_{q\abs{M}}  \geq  \frac{\epsilon\fpr{n-\abs{M}}\ubar{\lambda}_v}{9\bar{\lambda}_v}} + \mathbb{P}_y\fpr{\lambda_{\textrm{min}}\fpr{\widehat{\M{\Sigma}}_{(M)}} \leq  \frac{\ubar{\lambda}_v(n-\abs{M})}{2}} \\ \\
     &\leq \exp\fpr{-\frac{\epsilon\fpr{n-\abs{M}}\ubar{\lambda}_v}{36\;\bar{\lambda}_v} + \frac{q\abs{M}}{2}}  + \exp\fpr{-0.04(n-\abs{M})},
    \end{align*}
    where the first quantity in the last line is obtained by Chernoff's bound for chi-square distribution. This completes the proof of Lemma \ref{Lemma: LeastSqEstDist}. 
\end{proof}

\begin{proof}[Proof of Theorem \ref{Lemma: Mainresult}]
The statement of the theorem is equivalent to showing that,
\begin{equation}\label{eqn: maintheoremequivform}
   \sum_{\substack{M: \abs{M} \leq n^\alpha \\ M \neq M_{\text{o}}}} \frac{r_\epsilon\fpr{M \mid \M{Y}}}{r_\epsilon\fpr{M_{\text{o}} \mid \M{Y}}} \overset{\mathbb{P}_y}{\longrightarrow} 0,
\end{equation}
as $n \to \infty$ or $n, p \to \infty$.  To show this, observe that the ratio in sum has the form,
\begin{align*}
    \frac{r_\epsilon(M \mid \M{Y})}{r_\epsilon(M_{\text{o}} \mid \M{Y})} &=\pi^{\frac{q(\abs{M}-\abs{M_{\text{o}}})}{2}} \frac{\Gamma_q\fpr{\frac{n-\abs{M}}{2}}}{\Gamma_q\fpr{\frac{n-\abs{M_{\text{o}}}}{2}}} 
   \tpr{\frac{\det \sighat{(M_{\text{o}})}}{\det \sighat{(M)}}}^{\frac{n- \abs{M}- q}{2}} \times \\
& \hspace{5 cm} \fpr{\det \sighat{(M_{\text{o}})}}^{\frac{\abs{M}-\abs{M_{\text{o}}}}{2}}
 \frac{\mathbb{E}(h_\epsilon(\bB_{\cdot\,M}))}{\mathbb{E} (h_\epsilon(\bB_{\cdot\,M_{\text{o}}}))} .
\end{align*}
The multivariate gamma function is defined as the product of univariate gamma functions, and so, using the gamma function inequalities in \cite{jameson2013inequalities}, the ratio of the multivariate gamma functions is bounded by,
\begin{align*}
    \frac{\Gamma_q\fpr{\frac{n-\abs{M}}{2}}}{\Gamma_q\fpr{\frac{n-\abs{M_{\text{o}}}}{2}}} &=
        \prod_{j=1}^q \Gamma\fpr{\frac{n - \abs{M} - j + 1}{2}}\tpr{\Gamma\fpr{\frac{n - \abs{M_{\text{o}}} - j + 1}{2}}}^{-1}  \\
    &\leq \begin{cases}
        \prod_{j=1}^q \fpr{\frac{n - \abs{M_{\text{o}}} - j + 1}{2}} \fpr{\frac{n - \abs{M} - j + 1}{2}}^{\frac{\abs{M_{\text{o}}}-\abs{M}}{2}-1} \quad &\text{if} \abs{M_{\text{o}}} \geq \abs{M} \\ \\
        \prod_{j=1}^q  \fpr{\frac{n - \abs{M} - j + 1}{2}-1}^{\frac{\abs{M_{\text{o}}}-\abs{M}}{2}} \quad &\text{otherwise}
    \end{cases}
\end{align*}
Moreover, for the true model $M_{\text{o}}$, because $\widehat{\M{\Sigma}}_{(M_{\text{o}})} \sim \textrm{Wishart}_q\fpr{n - \abs{M_{\text{o}}},\bVMo^0}$,
\begin{equation*}
    \frac{\det \sighat{(M_{\text{o}})}}{\det \bVMo^{0}} \sim \prod_{j=1}^q \bigchi^2_{n-\abs{M_{\text{o}}}-j+1},
\end{equation*}
by Theorem 3.2.15 of \cite{muirhead2009aspects}. Applying the Chernoff bound for the chi-square distribution and using the sub-additivity property of the probability measure,
\begin{equation*}
    \mathbb{P}_y\fpr{\det \sighat{(M_{\text{o}})} > \prod_{j=1}^q 3\bar{\lambda}_v\fpr{n-\abs{M_{\text{o}}}-j+1}} \leq q\exp\fpr{-\frac{n-\abs{M_{\text{o}}}-q+1}{4}} := V_{5,n}.
\end{equation*}

Expanding the quantity in~(\ref{eqn: maintheoremequivform}), 
\begin{align}
 \nonumber \sum_{\substack{M: \abs{M} \leq n^\alpha \\ M \neq M_{\text{o}}}} \frac{r_\epsilon\fpr{M \mid \M{Y}}}{r_\epsilon\fpr{M_{\text{o}} \vert \M{Y}}} &= \sum_{\substack{M: M \subsetneq M_{\text{o}}}} \frac{r_\epsilon\fpr{M \mid \M{Y}}}{r_\epsilon\fpr{M_{\text{o}} \mid \M{Y}}} + \sum_{\substack{M: \abs{M} \leq n^\alpha \\ M \not \subseteq M_{\text{o}}}} \frac{r_\epsilon\fpr{M \mid \M{Y}}}{r_\epsilon\fpr{M_{\text{o}} \mid \M{Y}}} \\
  &=  \underbrace{\sum_{j=1}^{\abs{M_{\text{o}}}} \sum_{\substack{M: \abs{M}=j \\ M \subsetneq M_{\text{o}}}} \frac{r_\epsilon\fpr{M \mid \M{Y}}}{r_\epsilon\fpr{M_{\text{o}} \mid \M{Y}}}}_{ :=\;\; T_1}  + \underbrace{\sum_{j=1}^{n^\alpha} \sum_{\substack{M: \abs{M}=j \\ M \not \subseteq M_{\text{o}}}} \frac{r_\epsilon\fpr{M \mid \M{Y}}}{r_\epsilon\fpr{M_{\text{o}} \mid \M{Y}}}}_{ :=\;\; T_2}. \label{eqn: maintheoremexpansion} 
\end{align}
Denote the two terms on the right side as $T_1$ and $T_2$, respectively. First consider $T_1$. By Theorem~\ref{Lemma: Ehtruemodel}, with probability exceeding $1-V_{3,n}$, $\mathbb{E}(h_\epsilon(\bB_{\cdot\,M_{\text{o}}}))$ is bounded from below by $1-g_n(M_{\text{o}}, \epsilon)$ with,
$$
g_n(M_{\text{o}}, \epsilon) := \exp\fpr{-\frac{\epsilon\fpr{n - \abs{M_{\text{o}}}}}{36} + \frac{q\abs{M_{\text{o}}}}{2}} + 2\exp\fpr{-\frac{1}{8}\spr{\sqrt{n - \abs{M_{\text{o}}}}-2\sqrt{q}}^2}.
$$
Since the quantity $g_n(M_{\text{o}}, \epsilon)$ vanishes as $n \to \infty$, for sufficiently large $n$, $g_n(M_{\text{o}}, \epsilon) < K$ for some $K \in (0,1)$. Bounding the ratio of the determinants of the residual matrices by Case 1 of Lemma \ref{Lemma: RatioofRSS}, choosing $n$ large enough so that $\mathbb{E}(h_\epsilon(\bB_{\cdot\,M_{\text{o}}})) > 1-K$, and bounding the $\mathbb{E}(h_\epsilon(\bB_{\cdot\,M}))$ for any model $M \not\subseteq M_{\text{o}}$ by $1$,
\begin{align*}
    \frac{r_\epsilon(M \mid \M{Y})}{r_\epsilon(M_{\text{o}} \mid \M{Y})} &\leq \frac{e^{-qn^\alpha\log\abs{M_{\text{o}}}}}{1-K}\prod_{j=1}^q \frac{\fpr{\frac{n - \abs{M_{\text{o}}} - j + 1}{2}} \fpr{3\pi\bar{\lambda}_v \fpr{n-\abs{M_{\text{o}}}-j + 1}}^{\frac{ \abs{M}-\abs{M_{\text{o}}}}{2}}}{\fpr{\frac{n - \abs{M} - j + 1}{2}}^{1-\frac{ \abs{M_{\text{o}}}-\abs{M}}{2}}} \\
    &\leq \frac{e^{-qn^\alpha\log\abs{M_{\text{o}}}}}{1-K}\prod_{j=1}^q \tpr{\frac{n - \abs{M} - j + 1}{6\pi\bar{\lambda}_v\fpr{n - \abs{M_{\text{o}}} - j + 1}}}^{\frac{ \abs{M_{\text{o}}}-\abs{M}}{2}}\fpr{\frac{n - \abs{M_{\text{o}}} - j + 1}{n - \abs{M} - j + 1}} \\
    &\leq \frac{e^{-qn^\alpha\log\abs{M_{\text{o}}}}}{1-K}\prod_{j=1}^q \tpr{\frac{n }{6\pi\bar{\lambda}_v\fpr{n - \abs{M_{\text{o}}} - j + 1}}}^{\frac{ \abs{M_{\text{o}}}-\abs{M}}{2}} \\
     &\leq \frac{1}{1-K}e^{-qn^\alpha\log\abs{M_{\text{o}}} - q\log(\bar{\lambda}_v)(\abs{M_{\text{o}}}-\abs{M})/2},
\end{align*}
where the last inequality holds for $n > \frac{\abs{M_{\text{o}}}+q}{1-1/6\pi}$. This implies that with probability exceeding $1-V_{1,n} - V_{3,n} - V_{5,n}$,
\begin{align*}
    T_1 &= \sum_{j=1}^{\abs{M_{\text{o}}}} \sum_{\substack{M: \abs{M}=j \\ M \subsetneq M_{\text{o}}}} \frac{r_\epsilon\fpr{M \mid \M{Y}}}{r_\epsilon\fpr{M_{\text{o}} \mid \M{Y}}}  \leq \sum_{j=1}^{\abs{M_{\text{o}}}} \binom{\abs{M_{\text{o}}}}{j} \underset{M \subsetneq M_{\text{o}}; \abs{M}=j}{\max}\frac{r_\epsilon(M \mid \M{Y})}{r_\epsilon(M_{\text{o}} \mid \M{Y})} \\
     &\leq\sum_{j=1}^{\abs{M_{\text{o}}}}  \fpr{1-K}^{-1}\exp\fpr{-qn^\alpha\log\abs{M_{\text{o}}} - q\log(\bar{\lambda}_v)(\abs{M_{\text{o}}}-j)/2 + j\log \abs{M_{\text{o}}}} \\
    &\leq \fpr{1-K}^{-1}\abs{M_{\text{o}}}\exp\fpr{-qn^\alpha\log\abs{M_{\text{o}}} + \abs{M_{\text{o}}}\log\abs{M_{\text{o}}} -q\mathrm{I}(\bar{\lambda}_v < 1)\log(\bar{\lambda}_v)\abs{M_{\text{o}}}/2}.
\end{align*}
By Condition~\ref{cond: supelementsB}, $1-V_{1,n}-V_{3,n}-V_{5,n} \to 1$ as $n \to \infty$ or $n, p \to \infty$, and so $T_1 \to 0$ in probability as $n \to \infty$ or $n, p \to \infty$.

Next, consider $T_2$. Note that for models $M$ such that $M \not\subset M_{\text{o}}$ and $|M| \le n^{\alpha}$, By Theorem~\ref{Lemma: Ehbigmodels}, for large $n$, with probability exceeding $1-V_{4,n}$, $\mathbb{E}(h_{\epsilon}(\bB_{\cdot\,M})) \leq \tilde{g}_n(\epsilon,M)$ where,
$$
\tilde{g}_n(\epsilon,M) := \exp\fpr{-\frac{\epsilon\fpr{n - \abs{M}}}{36} +  \frac{q\abs{M}}{2}} + 2\exp\fpr{-\frac{1}{8}\spr{\sqrt{n - \abs{M}}-2\sqrt{q}}^2}.
$$ 
Bounding the ratio of the determinants of the residual matrices as in Case 2 of Theorem \ref{Lemma: RatioofRSS}, and choosing $n$ large enough so that $\mathbb{E}(h_{\epsilon}(\bB_{\cdot\,M_{\text{o}}})) > 1-K$ with high probability, for all $M \not \subseteq M_{\text{o}}$ with $|M| \le n^{\alpha}$,
\begin{align*}
    \frac{r_\epsilon(M \mid \M{Y})}{r_\epsilon(M_{\text{o}} \mid \M{Y})} 
    &\leq \begin{cases}
            \tpr{\frac{n }{6\pi\bar{\lambda}_v\fpr{n - \abs{M_{\text{o}}} - q}}}^{\frac{ q(\abs{M_{\text{o}}}-\abs{M})}{2}}\frac{\tilde{g}_n(\epsilon,M)}{1-K} e^{q\fpr{n^{\alpha}\log\fpr{n-\abs{M_{\text{o}}}}+\abs{M}\log p}} \;\; &\textrm{if} \;\; \abs{M} \leq \abs{M_{\text{o}}} \\ \\
    \fpr{\frac{6\pi\bar{\lambda}_v \fpr{n - \abs{M_{\text{o}}}}}{n-\abs{M}-q-1}}^{q\fpr{\frac{ \abs{M}-\abs{M_{\text{o}}}}{2}}}\frac{\tilde{g}_n(\epsilon,M)}{1-K} e^{q\fpr{n^{\alpha}\log\fpr{n-\abs{M_{\text{o}}}}+\abs{M}\log p}} \;\; &\textrm{if} \;\; \text{otherwise} 
    \end{cases}
\end{align*}
Then,
\[
T_2 = \underbrace{\sum_{j=1}^{\abs{M_{\text{o}}}} \sum_{\substack{M: \abs{M}=j \\ M \not \subseteq M_{\text{o}}}} \frac{r_\epsilon\fpr{M \mid \M{Y}}}{r_\epsilon\fpr{M_{\text{o}} \mid \M{Y}}}}_{ := \;\; T_{21}} \;\;+
\underbrace{\sum_{j=\abs{M_{\text{o}}}+1}^{n^\alpha} \sum_{\substack{M: \abs{M}=j \\ M \not \subseteq M_{\text{o}}}} \frac{r_\epsilon\fpr{M \mid \M{Y}}}{r_\epsilon\fpr{M_{\text{o}} \mid \M{Y}}}}_{ := \;\; T_{22}}
\]
Consider $T_{21}$ and $T_{22}$ separately.  For sufficiently large $n$, using $\tilde{g}_n(\epsilon,M)$ to bound $\mathbb{E}(h_{\epsilon}(\bB_{\cdot\,M}))$, with probability exceeding $1-V_{2,n} - V_{3,n} - V_{4,n} - V_{5,n}$,
\small{\begin{align*}
    T_{21} &\leq \sum_{j=1}^{\abs{M_{\text{o}}}} \binom{p}{j} \underset{M \not\subseteq M_{\text{o}}; \abs{M}=j}{\max}\;\frac{r_\epsilon(M \mid \M{Y})}{r_\epsilon(M_{\text{o}} \mid \M{Y})} \\
    &\leq \frac{2}{1-K}\sum_{j=1}^{\abs{M_{\text{o}}}} \exp\fpr{-\frac{n-j}{8} -q\log(\bar{\lambda}_v)(\abs{M_{\text{o}}}-j)/2  + qn^{\alpha}\log\fpr{n-\abs{M_{\text{o}}}}+ (q+1)j\log p} \\
    &  +  \frac{1}{1-K}\sum_{j=1}^{\abs{M_{\text{o}}}} \exp\fpr{-\frac{\epsilon(n-j)}{36} +  \frac{qj}{2} -q\log(\bar{\lambda}_v)(\abs{M_{\text{o}}}-j)/2  + qn^{\alpha}\log\fpr{n-\abs{M_{\text{o}}}}+ (q+1)j\log p} \\
    &\leq \frac{\abs{M_{\text{o}}}}{1-K} \left\{2 \exp\fpr{-\frac{n-\abs{M_{\text{o}}}}{8} - \frac{q\abs{M_{\text{o}}}}{2}\mathrm{I}(\bar{\lambda}_v < 1)\log(\bar{\lambda}_v) + qn^{\alpha}\log\fpr{n-\abs{M_{\text{o}}}}+ (q+1)\abs{M_{\text{o}}}\log p} \right. \\
    & \left. + \exp\left(-\frac{\epsilon(n-\abs{M_{\text{o}}})}{36} +qn^{\alpha}\log\fpr{n-\abs{M_{\text{o}}}}+ (q+1)\abs{M_{\text{o}}}\log p  + \frac{q\abs{M_{\text{o}}}}{2}\{1-\mathrm{I}(\bar{\lambda}_v < 1)\log(\bar{\lambda}_v)\} \right) \right\}.
\end{align*}}

For $T_{22}$, because  $(n-\abs{M_{\text{o}}}/(n-\abs{M}-q-1)$ converges to $1$ for all $M$ with $\abs{M} \leq n^\alpha$, we can choose $n$ large enough so that $(n-\abs{M_{\text{o}}}/(n-\abs{M}-q-1) < 2$.  Then, again using the bound $\tilde{g}_n(\epsilon,M)$, we obtain that for all $M \not \subseteq M_{\text{o}}$ such that $\abs{M_{\text{o}}} \leq \abs{M} \leq n^\alpha$, with probability exceeding $1-V_{2,n} - V_{3,n} - V_{4,n} - V_{5,n}$,
\small{\begin{align*}
    T_{22} := &\sum_{j=\abs{M_{\text{o}}}+1}^{n^\alpha} \sum_{\substack{M: \abs{M}=j \\ M \not \subseteq M_{\text{o}}}} \frac{r_\epsilon\fpr{M \mid \M{Y}}}{r_\epsilon\fpr{M_{\text{o}} \mid \M{Y}}} \\
    &\leq \sum_{j=1}^{n^{\alpha}} \binom{p}{j} \underset{M \not\subseteq M_{\text{o}}; \abs{M}=j}{\max}\;\frac{r_\epsilon(M \mid \M{Y})}{r_\epsilon(M_{\text{o}} \mid \M{Y})} \\
    &\leq \frac{2}{1-K}\sum_{j=1}^{n^{\alpha}} \exp\fpr{-\frac{n-j}{8} + \frac{qj}{2}\log\fpr{12\pi\bar{\lambda}_v} + qn^{\alpha}\log\fpr{n-\abs{M_{\text{o}}}}+ (q+1)j\log p} \\
    & +  \frac{1}{1-K} \sum_{j=1}^{n^{\alpha}} \exp\fpr{-\frac{\epsilon(n-j)}{36} +  \frac{qj}{2}(1+\log\fpr{12\pi\bar{\lambda}_v}) + qn^{\alpha}\log\fpr{n-\abs{M_{\text{o}}}}+ (q+1)j\log p} \\
    &\leq \frac{2n^\alpha}{1-K}  \exp\left(-\frac{n-n^\alpha}{8} + n^{\alpha}\spr{q\log\fpr{n-\abs{M_{\text{o}}}}+ (q+1)\log p} \right. \\
& \hspace{5 cm} \left. + \frac{q}{2}\log\fpr{12\pi\bar{\lambda}_v}\{n^\alpha\mathrm{I}(\bar{\lambda}_v > \frac{1}{12\pi})+1\} \right) \\
    & \quad + \frac{n^\alpha}{1-K} \exp\left(-\frac{\epsilon(n-n^\alpha)}{36} +qn^{\alpha}\log\fpr{n-\abs{M_{\text{o}}}}+ (q+1)n^\alpha\log p  \right. \\
    & \hspace{50 mm} + \left. \frac{q}{2}\fpr{1+\log\fpr{12\pi\bar{\lambda}_v}}\{n^\alpha\mathrm{I}(\bar{\lambda}_v > 1/12\pi e)+1\} \right).
\end{align*}}

Lastly, by Condition~\ref{cond: supelementsB} since $1-V_{2,n} - V_{3,n} - V_{4,n} - V_{5,n} \to 1$ as $n \to \infty$ or $n, p \to \infty$, $T_{2} = T_{21} + T_{22} \to 0$ in probability as $n \to \infty$ or $n, p \to \infty$. This completes the proof of Theorem~\ref{Lemma: Mainresult}. 

\end{proof}

\begin{acks}{
The authors would like to thank the associate editor and two reviewers for their constructive comments which led to a significantly improved version of the manuscript. The authors would also like to thank Ms. Sukanya Bhattacharyya for providing additional computational resources without which the extensive numerical studies presented in the paper would not have been possible. 

Research reported in this publication was supported by the National Heart, Lung, and
Blood Institute of the National Institutes of Health under Award Number R56HL155373.
The content is solely the responsibility of the authors and does not necessarily represent
the official views of the National Institutes of Health.}
\end{acks}

%\begin{supplement}
%  \sname{Supplement to}
%  \stitle{``The EAS approach to variable selection for multivariate response data in high-dimensional settings''.}
%\sdescription{For conciseness of the manuscript, longer derivations, additional lemmas, and proofs have been moved to these supplementary material.}
%\end{supplement}

%%%%%%%%%%%%%%%%%%%%%%%%%%%%%%%%%%%%%%%%%%%%%%%%%%%%%
%\begin{supplement}
%  \sname{Supplement to}
%  \stitle{``The EAS approach to variable selection for multivariate response data in high-dimensional settings''.}
%\sdescription{For conciseness of the manuscript, longer derivations, additional lemmas, and proofs have been moved to these supplementary material.}
%\end{supplement}

\bibliographystyle{imsart-number}
\bibliography{bibliography}

\end{document}